\documentclass[preprint,nocopyrightspace]{sigplanconf}

\usepackage{amsfonts,amsmath,amsthm,amscd,amssymb,float}
\usepackage[german,british,english]{babel}
\usepackage{stmaryrd}
\usepackage{comment}
\usepackage{xspace}
\usepackage{color,graphicx}
\usepackage{mathpartir}
\usepackage[utf8]{inputenc}
\usepackage{newunicodechar}
\usepackage{bussproofs}
\usepackage{hyperref}
\usepackage{MnSymbol}
\usepackage{microtype}
\usepackage{tikz}
\usepackage[inline]{enumitem}
\usepackage{adjustbox}
\usepackage{tkz-berge}
\usepackage{verbatimbox}
\usepackage{textpos}
\usepackage{multicol}

\usetikzlibrary{shapes,positioning,decorations.markings,arrows}

\tikzstyle{->>>} =
  [shorten >= 0.3mm,
   decoration =
     {markings,
      mark=at position -2.4mm with {\arrow{>}},
      mark=at position -1.2mm with {\arrow{>}},
      mark=at position 1 with {\arrow{>}}},
   postaction={decorate}]

\pgfdeclaredecoration{funbisim}{final}{
  \state{final}[width=\pgfdecoratedpathlength]{
    \draw[->]
      (0,\pgfdecorationsegmentamplitude+0.1mm) -- +(\pgfdecoratedpathlength,0);
    \draw[-]
      (0,-\pgfdecorationsegmentamplitude-0.1mm) -- +(\pgfdecoratedpathlength,0);}}
\tikzset{
  funbisim/.style={
    decoration={funbisim, amplitude=0.25ex},
    decorate,
    funbisim options/.style={#1}    
  }}

\pgfdeclaredecoration{bisim}{final}{
  \state{final}[width=\pgfdecoratedpathlength]{
    \draw[<->]
      (0,\pgfdecorationsegmentamplitude+0.1mm) -- +(\pgfdecoratedpathlength,0);
    \draw[-]
      (0,-\pgfdecorationsegmentamplitude-0.1mm) -- +(\pgfdecoratedpathlength,0);}}
\tikzset{
  bisim/.style={
    decoration={bisim, amplitude=0.25ex},
    decorate,
    bisim options/.style={#1}    
  }}

\makeatletter
\pgfdeclareshape{transbox}{
  \inheritsavedanchors[from=rectangle]
  \inheritanchor[from=rectangle]{center}
  \inheritanchor[from=rectangle]{north}
  \inheritanchor[from=rectangle]{south}
  \inheritanchor[from=rectangle]{west}
  \inheritanchor[from=rectangle]{east}
  \inheritbehindbackgroundpath[from=rectangle]
  \inheritbackgroundpath[from=rectangle]
  \inheritbeforebackgroundpath[from=rectangle]
  \inheritbehindforegroundpath[from=rectangle]
  \inheritforegroundpath[from=rectangle]
  \inheritbeforeforegroundpath[from=rectangle]
}
\makeatother

\newcommand\tikzscale{1}

\newcommand\transpicture[1]{\vcentered{\scalebox{\tikzscale}{\begin{tikzpicture}[>=stealth,node distance=4mm] #1 \end{tikzpicture}}}}
\newcommand\rbpicture[1]{\vcentered{\begin{tikzpicture}[>=stealth,node distance=9mm,inner sep=0.5pt] #1 \end{tikzpicture}}}

\newcommand\ltgnode[3][]{\node[#1,draw,shape=circle,inner sep=0.1mm,minimum size=5mm](#2){\large #3}}

\newcommand{\addPrefix}[3][]{\node[node distance=1mm,#1,left=of #2.north]{\inMath{\scriptstyle(#3)}}}
\newcommand{\addPos}[2]{\node[node distance=1.2mm,right=of #1.north]{\inMath{\scriptstyle #2}}}
\newcommand{\addPrefixswphantom}[3][]{\node[node distance=1mm,#1,right=of #2.south]{\inMath{\scriptstyle\phantom{(#3)}}}}

\newcommand{\translation}[3]{\translationif{#1}{#2}{#3}{}{}}
\newcommand{\rbedge}[4][]{\Edge[style={->,thin,#1},label={#2}](#3)(#4)}

\newcommand{\transname}[1]{\vspace{-1ex}\noindent{#1}:\hspace*{0.5ex}}
\newcommand{\translationif}[5]{\transname{#1} \transpicture{#2}\vcentered{$\overset{#4}{\underset{#5}{\implies}}$}\transpicture{#3}}

\newcommand{\translationifarrowspacing}[6]{\transname{#1} \transpicture{#2}\vcentered{$\overset{#5}{\underset{#6}{#3}\implies{#3}}$}\transpicture{#4}}

\newcommand{\transinit}[1]{\hspace*{0ex}\hfill\transpicture{#1}\hfill}
\newcommand{\transstep}[2]{$\implies_{#1}$\hspace*{-1cm}\hfill\transpicture{#2}\hfill}
\newcommand{\transbreak}{\\[3.5mm]}

\newcommand{\fusionstep}[4]{&{#1}:&  
                            & &       
                            \begin{gathered}[c]
                                \transpicture{#2}
                                \\[-1.75ex]
                                \rotatebox{270}{$\Rightarrow$}
                                \\[0.75ex]
                                \transpicture{#3}
                            \end{gathered}
                            & 
                            \hspace*{3ex}\Longleftarrow
                            & &
                              \transpicture{#4}
                            }
\newcommand{\fusionstepplus}[6]{&{#1}:&  
                            & &      
                            \begin{gathered}[c]
                                \transpicture{#2}
                                \\[-1.75ex]
                                \rotatebox{270}{$\Rightarrow$}
                                \\[0.75ex]
                                \transpicture{#3}
                            \end{gathered}
                            & 
                            \hspace*{3ex}\Longleftarrow
                            & &
                            \begin{gathered}[c]
                              {\mbox{}\hspace*{20ex}\scalebox{\tikzscale}{#6}}
                              \\
                              \transpicture{#5}
                              \\
                              \phantom{\scalebox{\tikzscale}{#6}}
                            \end{gathered}  
                            }
\newcommand{\fusiontwostep}[5]{&{#1}:&  
                            & &     
                            \begin{gathered}[c]
                                \transpicture{#2}
                                \\[-1.75ex]
                                \rotatebox{270}{$\Rightarrow$}
                                \\[0.75ex]
                                \transpicture{#3}
                                \\[-1.75ex]
                                \rotatebox{270}{$\Rightarrow$}
                                \\[0.75ex]
                                \transpicture{#4}
                            \end{gathered}
                            & 
                            \hspace*{3ex}\Longleftarrow
                            & &
                              \transpicture{#5}
                            }
\newcommand{\fusionthreestep}[6]{&{#1}:&  
                            & &     
                            \begin{gathered}[c]
                                \transpicture{#2}
                                \\[-1.75ex]
                                \rotatebox{270}{$\Rightarrow$}
                                \\[0.75ex]
                                \transpicture{#3}
                                \\[-1.75ex]
                                \rotatebox{270}{$\Rightarrow$}
                                \\[0.75ex]
                                \transpicture{#4}
                                \\[-1.75ex]
                                \rotatebox{270}{$\Rightarrow$}
                                \\[0.75ex]
                                \transpicture{#5}
                            \end{gathered}
                            & 
                            \hspace*{3ex}\Longleftarrow
                            & &
                              \transpicture{#6}
                            }

\newcommand{\linetwofusiononezerosteps}[7]{&{#1}:&  
                            & &       
                            \begin{gathered}[c]
                                \transpicture{#2}
                                \\[-1.75ex]
                                \rotatebox{270}{$\Rightarrow$}
                                \\[0.75ex]
                                \transpicture{#3}
                            \end{gathered}
                            \hspace*{3ex}\Longleftarrow\hspace*{3ex}
                              \transpicture{#4}
                            & \hspace*{3ex}\phantom{\Longleftarrow}\hspace*{3ex}
                            & &   
                            \begin{gathered}[c]
                                \transpicture{#6}
                            \end{gathered}
                            \hspace*{3ex}\Longleftarrow\hspace*{1ex}
                              \transpicture{#7}  
                            }

\makeatletter
\def\blfootnote{\xdef\@thefnmark{}\@footnotetext}
\makeatother

\newcommand{\rulestranslambdaletreccaltolhotgs}{{\cal R}}
\newcommand{\rulestranslambdaletreccaltoltgs}{{\cal R}_{\snlvarsucc}}

\newcommand{\rulestranslambdaletreccaltolhotgsgenerated}{$\rulestranslambdaletreccaltolhotgs$\nb-ge\-ne\-ra\-ted} 

\newcommand\inMath[1]{\ensuremath{#1}\xspace}
\newcommand\vcentered[1]{\raisebox{-0.5\height}{#1}}

\newcommand{\nb}{\nobreakdash}

\newcommand{\scdots}{\hspace*{1pt}\cdots\hspace*{1.25pt}}

\newcommand{\myparagraph}[1]{\vspace{0.45ex}\noindent\emph{#1}.}
\newcommand{\myparagraphbf}[1]{\vspace{0.45ex}\noindent\emph{\textbf{{#1}.}}}

\newcommand\sep[1]{&#1&}

\newcommand{\funin}{\mathrel{:}}
\newcommand{\funap}[2]{#1({#2})}
\newcommand{\bfunap}[3]{{#1}({#2},\hspace*{0.5pt}{#3})}
\newcommand{\indap}[2]{#1_{#2}}

\newcommand{\sdefdby}{{:=}}
\newcommand{\defdby}{\mathrel{\sdefdby}}
\newcommand{\length}[1]{{\left|{#1}\right|}}

\newcommand{\sgraphsize}{\text{\normalfont size}}
\newcommand{\graphsize}{\funap{\sgraphsize}}
\newcommand{\termsize}{\length}

\newcommand{\simage}{\textrm{im}}
\newcommand{\image}{\funap{\simage}}

\newcommand{\factor}[2]{{#1}\hspace*{-1pt}/_{\hspace*{-1pt}{#2}}}

\makeatletter
\def\namedlabel#1#2{\begingroup
    #2%
    \def\@currentlabel{#2}%
    \phantomsection\label{#1}\endgroup
}
\makeatother

\newcommand\vs[1]{\mathit{vs}({#1})} 
\newcommand\vszero[1]{\mathit{v}({#1})} 

\newcommand\partialTo\rightharpoonup

\newcommand{\scoll}{{{\text{\small\textbar}}\hspace{-0.73ex}\downarrow}}
\newcommand{\coll}[1]{{#1}\hspace{0.17ex}{\scoll}}

\newcommand\tuple[1]{\langle #1 \rangle}
\newcommand\tuplespace{\hspace*{0.5pt}}
\newcommand\pair[2]{\tuple{#1, \tuplespace #2}}

\newcommand{\sfuncomp}{\circ}
\newcommand{\scompfuns}[2]{{#1}\mathrel{\sfuncomp}{#2}}
\newcommand{\compfuns}[2]{\funap{\scompfuns{#1}{#2}}}

\newcommand{\eqclofwrt}[2]{\indap{[{#1}]}{\hspace*{-0.5pt}#2}}

\newcommand{\existsst}[2]{\exists{#1}.\;{#2}}

\newcommand{\emptyword}{\epsilon}

\newcommand{\wordsover}[1]{{#1^*}}

\newcommand{\mcdots}{\mathrel{\cdots}}
\newcommand{\mcdotsspace}{\mathrel{\cdots\,}}

\newcommand{\srestrictfunto}[2]{{#1}{\mid}_{#2}}
\newcommand{\restrictfunto}[2]{\funap{\srestrictfunto}}

\newcommand{\sidfun}{\text{\normalfont id}}

\newcommand{\sidfunon}{\indap{\sidfun}}

\newcommand{\srestrictto}[2]{{#1}\!\mid_{#2}}
\newcommand{\restrictto}[2]{\funap{\srestrictto}}

\newcommand{\niks}{}

\newcommand{\slogand}{{\wedge}}
\newcommand{\logand}{\mathrel{\slogand}}

\newcommand{\subap}[2]{#1_{#2}}
\newcommand{\supap}[2]{#1^{#2}}
\newcommand{\bpap}[3]{{#1}_{#2}^{#3}}
\newcommand{\pbap}[3]{{#1}^{#2}_{#3}}

\let\oldlambda\lambda
\renewcommand\lambda{\inMath\oldlambda}
\newunicodechar{λ}{{\tt\lambda}}
\let\oldalpha\alpha
\renewcommand\alpha{\inMath\oldalpha}
\let\oldmu\mu
\renewcommand\mu{\inMath\oldmu}

\newcommand{\lambdaletreccal}{\txtlambdaletreccal} 

\newcommand{\ssin}{\ensuremath{\mathsf{in}}}

\newcommand{\sslet}{\ensuremath{\mathsf{let}}}

\newcommand{\letin}[2]{\txtlet\;{#1}\;\txtin\;{#2}}

\newcommand{\txtin}{{\text{\sf in}}}
\newcommand{\txtlet}{{\text{\sf let}}}
\newcommand{\txtletrec}{\text{\normalfont\sf letrec}}
\newcommand{\txtrec}{\text{\sf rec}}

\newcommand{\txtlambdacal}{\ensuremath{\lambda}}
\newcommand{\txtlambdabhcal}{\ensuremath{\lambda_{\blackhole}}}
\newcommand{\txtinflambdacal}{\ensuremath{\lambda^{\hspace*{-1pt}\infty}}}
\newcommand{\txtinflambdabhcal}{\ensuremath{\lambda^{\hspace*{-1pt}\infty}_{\blackhole}}}
\newcommand{\txtlambdaletreccal}{\ensuremath{\indap{\lambda}{\txtletrec}}}
\newcommand{\txtlambdaletrecbhcal}{\ensuremath{\indap{\lambda}{\txtletrec,\blackhole}}}
\newcommand{\txtlambdamucal}{\ensuremath{\indap{\lambda}{\mu}}}

\newcommand{\abindgroup}{B}

\newcommand{\abindgroupi}{\indap{\abindgroup}}
\newcommand{\abindgroupacc}{B'}

\newcommand{\abindgroupacci}{\indap{\abindgroupacc}}

\newcommand{\abindgroupbar}{\bar{B}}

\newcommand{\abindgroupbari}{\indap{\abindgroupbar}}

\newcommand{\arecvar}{f}
\newcommand{\brecvar}{g}
\newcommand{\crecvar}{h}
\newcommand{\arecvari}{\indap{\arecvar}}
\newcommand{\brecvari}{\indap{\brecvar}}

\newcommand{\llrecvars}{{\cal R}}

\newcommand{\bfundef}[2]{{#1},\,{#2}}
\newcommand{\tfundef}[3]{{#1},\,{#2},\,{#3}}

\newcommand{\sbgsupseteq}{\supseteq'}
\newcommand{\bgsupseteq}{\mathrel{\sbgsupseteq}}
\newcommand{\sbgunion}{\cup'}
\newcommand{\bgunion}{\mathrel{\sbgunion}}

\newcommand{\unfrulelabs}{\sslabs}
\newcommand{\unfrulelapp}{\sslapp}
\newcommand{\unfruleletin}{\sslet\_\ssin}
\newcommand{\unfruleletrec}{\sslet\text{-rec}}
\newcommand{\unfrulegarbage}{\text{gc}}
\newcommand{\unfruletighten}{\text{tighten}}
\newcommand{\unfruleblackhole}{\blackhole}

\newcommand{\derivrulelabs}{\sslabs}
\newcommand{\derivrulelapp}{@}
\newcommand{\derivrulelet}{\txtlet}
\newcommand{\derivrulerec}{\txtrec}
\newcommand{\derivrulenlvar}{\snlvar}
\newcommand{\derivrulenlvarsucc}{\snlvarsucc}

\newcommand{\alphaconversion}{$\alpha$\nb-con\-ver\-sion}

\newcommand{\bindcaptchains}{bin\-ding--cap\-tu\-ring chains}
\newcommand{\betareduction}{$\beta$\nb-re\-duc\-tion}
\newcommand{\eagscope}{eager-scope}
\newcommand{\lambdacalculus}{$\lambda$\nb-calculus}

\newcommand{\lambdaterm}{$\lambda$\nb-term}
\newcommand{\lambdaterms}{\lambdaterm{s}}

\newcommand{\inflambdabhterm}{$\lambda_{\bh}^{\hspace*{-1.5pt}\infty}$\nb-term}
\newcommand{\inflambdabhterms}{\inflambdabhterm{s}}
\newcommand{\lambdavariable}{$\lambda$\nb-va\-ri\-able}
\newcommand{\lambdavariables}{\lambdavariable{s}}
\newcommand{\lambdaabstraction}{$\lambda$\nb-ab\-strac\-tion}
\newcommand{\lambdaabstractions}{\lambdaabstraction{s}}

\newcommand{\lambdaexpressions}{\lambdaexpressions}
\newcommand{\inflambdaterm}{$\lambda^{\hspace*{-1.5pt}\infty}$\hspace*{-2pt}\nb-term}
\newcommand{\inflambdaterms}{\inflambdaterm{s}}
\newcommand{\lambdaletrecterm}{$\txtlambdaletreccal$\nb-term}
\newcommand{\lambdaletrecterms}{\lambdaletrecterm{s}}
\newcommand{\lambdaletreccalterm}{$\txtlambdaletreccal$\nb-term}
\newcommand{\lambdaletreccalterms}{\lambdaletreccalterm{s}}
\newcommand{\lambdaletreccalexpression}{$\txtlambdaletreccal$\nb-ex\-pres\-sion}

\newcommand{\lambdamuterm}{$\txtlambdamucal$\protect\nb-term}
\newcommand{\lambdamuterms}{\lambdamuterm{s}}
\newcommand{\nondeterminism}{non-deter\-mi\-nism}
\newcommand{\picalculus}{$\pi$\nb-cal\-cu\-lus}

\newcommand{\termgraph}{term graph}    
\newcommand{\termgraphs}{term graphs}  
\newcommand{\Termgraphs}{Term graphs}  
\newcommand{\termdashgraph}{term-graph}

\newcommand{\slappvertex}{$\sslapp$\nb-ver\-tex}

\newcommand{\snlvarvertex}{$\snlvar$\nb-ver\-tex}

\newcommand{\slabsvertices}{$\sslabs$\nb-ver\-ti\-ces}

\newcommand{\snlvarsuccvertex}{$\snlvarsucc$\nb-ver\-tex}
\newcommand{\snlvarsuccvertices}{$\snlvarsucc$\nb-ver\-ti\-ces}

\newcommand{\letbinding}{$\txtlet$\nb-bin\-ding}
\newcommand{\letbindings}{\letbinding{s}}

\newcommand{\letfloating}{let-floa\-ting}

\newcommand{\letrecexpression}{$\txtletrec$\nb-ex\-pres\-sion}

\newcommand{\inpart}{$\txtin$-part}

\newcommand{\lambdalifting}{lambda-lif\-ting}

\newcommand{\letexpression}{$\txtlet$\nb-ex\-pres\-sion}
\newcommand{\letexpressions}{$\txtlet$\nb-ex\-pres\-sions}

\newcommand{\CSE}{CSE}

\newcommand{\unfoldingequivalent}{un\-fol\-ding-equi\-va\-lent}
\newcommand{\unfoldingequivalence}{un\-fol\-ding equi\-va\-lence}

\newcommand{\alphaequivalence}{$\alpha$\nb-equi\-va\-lence}

\newcommand{\sinvAck}{\alpha}
\newcommand{\invAck}{\funap{\sinvAck}}

\newcommand{\lambdahotg}{$\lambda$\nb-ho-\termdashgraph{}}
\newcommand{\lambdahotgs}{\lambdahotg{s}}
\newcommand{\Lambdahotg}{Lambda higher-order \termgraph{}}
\newcommand{\Lambdahotgs}{\Lambdahotg{s}}
\newcommand{\alhotg}{{\cal G}}
\newcommand{\alhotgi}[1]{{\cal G}_{#1}}

\newcommand{\alhotreetg}{{\cal T}\hspace*{-1.75pt}r}
\newcommand{\alhotreetgi}[1]{\alhotreetg_{#1}}

\newcommand{\alaphotgi}[1]{{\cal G}_{#1}}

\newcommand{\lambdatg}{$\lambda$\nb-\termdashgraph{}{}}
\newcommand{\lambdatgs}{\lambdatg{s}}
\newcommand{\Lambdatg}{Lambda \termgraph{}}
\newcommand{\Lambdatgs}{\Lambdatg{s}}
\newcommand{\altg}{G}
\newcommand{\altgi}{\indap{\altg}}
\newcommand{\altgacc}{G'}

\newcommand{\altgdacc}{G''}

\newcommand\state{\inMath{s}}

\newcommand{\sslabs}{\lambda}
\newcommand{\slabs}[1]{\sslabs{#1}}
\newcommand{\labs}[2]{\slabs{#1}.\,{#2}}
\newcommand{\sslapp}{@}
\newcommand{\slapp}{\hspace*{1.5pt}}
\newcommand{\lapp}[2]{{#1}\slapp{#2}}

\newcommand\svar{0}

\newcommand{\asig}{\Sigma}

\newcommand{\sarity}{{\mit ar}}
\newcommand{\arity}{\funap{\sarity}}

\theoremstyle{plain}
\newtheorem{theorem}{Theorem}[section]
\newtheorem{lemma}[theorem]{Lemma}

\newtheorem{proposition}[theorem]{Proposition}

\theoremstyle{definition}
\newtheorem{example}[theorem]{Example}
\newtheorem{definition}[theorem]{Definition}

\newtheorem{remark}[theorem]{Remark}
\newcommand{\fig}[1]{\includegraphics[scale=0.81]{figs/#1.pdf}}
\newcommand{\figsmall}[1]{\includegraphics[scale=0.78]{figs/#1.pdf}}
\newcommand{\snlvar}{\inMath{\mathsf{0}}}
\newcommand{\nlvar}{\snlvar}
\newcommand{\snlvarsucc}{\inMath{\mathsf{S}}}
\newcommand{\nlvarsucc}[1]{\snlvarsucc\hspace*{0.5pt}{#1}}
\newcommand{\indir}{\inMath{\text{\small\textbar}}}
\newcommand{\nodef}{\inMath{\text{?}}}

\newcommand{\nllabsfo}{\funap{\sslabs}}
\newcommand{\snllappfo}{@}
\newcommand{\nllappfo}{\bfunap{\snllappfo}}

\newcommand{\snlvarsuccfo}{\snlvarsucc}
\newcommand{\nlvarsuccfo}{\funap{\snlvarsuccfo}}
\newcommand{\snlvarsuccfoi}[1]{\snlvarsucc^{#1}}
\newcommand{\nlvarsuccfoi}[1]{\funap{\snlvarsuccfoi{#1}}}
\newcommand{\sletCRS}[1]{\inMath{\textsf{let}_{#1}}}
\newcommand{\srecinCRS}[1]{\inMath{\textsf{rec-in}_{#1}}}

\newcommand{\letrecinCRS}[3]{\funap{\sletCRS{#1}}{\absCRS{#2}{\funap{\srecinCRS{#1}}{#3}}}}

\newcommand{\allter}{L}
\newcommand{\bllter}{P}
\newcommand{\cllter}{Q}
\newcommand{\allteri}{\indap{\allter}}

\newcommand{\allteracc}{L'}

\newcommand{\allteracci}{\indap{\allteracc}}

\makeatletter
\def\a#1{\reflectbox{$\m@th#1{\lambda}$}}
\def\adbmal{\inMath{\mathpalette{\a}{}}}
\makeatother
\newcommand{\flabs}[2]{\inMath{(#1)\hspace*{1pt}{#2}}}
\newcommand{\femptylabs}[1]{(*{[]})\hspace*{0.5pt}{#1}}

\newcommand{\vvb}[3]{#1 ^{#2}[{#3}]}
\newcommand{\vv}[2]{#1 ^{#2}}
\newcommand{\vvemptyb}[2]{#1 ^{#2}[{\hspace*{1pt}}]}
\newcommand{\vb}[2]{{#1}[{#2}]}

\newcommand\stdprefix{\vvb{x_0}{v_0}{B_0}\;\mcdots\;\vvb{x_n}{v_n}{B_n}}
\newcommand{\spcup}{\vec{\cup}} 
\newcommand{\pcup}{\mathrel{\spcup}}

\newcommand\vertices[2]{\averti{#1}\cdots\averti{#2}}

\newcommand{\avar}{x}
\newcommand{\bvar}{y}
\newcommand{\cvar}{z}

\newcommand{\aavar}{a}
\newcommand{\bbvar}{b}
\newcommand{\ccvar}{c}
\newcommand{\ddvar}{d}

\newcommand\ps{\vec{p}}
\newcommand\psix[1]{\vec{p}_{#1}}

\newcommand{\avari}[1]{\avar_{#1}}

\newcommand{\ater}{M}

\newcommand{\ateri}[1]{\ater_{#1}}

\newcommand{\aiter}{M}

\newcommand{\aiteri}[1]{\aiter_{#1}}

\newcommand{\freevars}{\funap{\textit{FV}}}

\newcommand{\slocfreevars}{\textit{loc-FV}}
\newcommand{\locfreevars}[1]{\funap{\slocfreevars}}

\newcommand{\slocfreerecvars}{\textit{loc-RV}}
\newcommand{\locfreerecvars}[1]{\funap{\slocfreerecvars}}

\newcommand{\sglobfreevars}{\textit{FV}}
\newcommand{\globfreevars}[1]{\funap{\sglobfreevars}}

\newcommand{\sreachfreevars}{\textit{RFV}} 
\newcommand{\reachfreevars}[1]{\funap{\sreachfreevars}}

\newcommand{\sTer}{\mathit{Ter}}
\newcommand{\Ter}{\funap{\sTer}}

\newcommand{\sametavar}{X}
\newcommand{\sbmetavar}{Y}
\newcommand{\scmetavar}{Z}
\newcommand{\sametavari}{\indap{X}}

\newcommand{\scmetavari}{\indap{Z}}

\newcommand{\cmetavar}{\funap{\scmetavar}}
\newcommand{\ametavari}[1]{\funap{\sametavari{#1}}}

\newcommand{\cmetavari}[1]{\funap{\scmetavari{#1}}}
\newcommand{\absCRS}[2]{[{#1}]\hspace*{1pt}{#2}}
\newcommand{\TRS}{TRS}

\newcommand{\CRS}{CRS}
\newcommand{\CRSs}{\CRS{s}}
\newcommand{\iCRS}{\inMath{\text{iCRS}}}

\newcommand{\HRS}{HRS}
\newcommand{\HRSs}{\HRS{s}}

\newcommand{\CRSnotation}{\CRS\nb-no\-ta\-tion}
\newcommand{\subst}[3]{{#1}[{#2}\defdby{#3}]}
\newcommand{\slabsCRS}{\mathsf{abs}}
\newcommand{\labsCRS}[2]{\funap{\slabsCRS}{\absCRS{#1}{#2}}}
\newcommand{\slappCRS}{\mathsf{app}}
\newcommand{\lappCRS}[2]{\funap{\slappCRS}{#1,#2}}
\newcommand{\sred}{{\to}}
\newcommand{\red}{\mathrel{\to}}

\newcommand{\sunfold}{\text{\normalfont unf}}
\newcommand{\sunfoldred}{\indap{\sred}{\sunfold}}
\newcommand{\unfoldred}{\mathrel{\sunfoldred}}

\newcommand{\sunfoldinfred}{\indap{\sinfred}{\sunfold}}
\newcommand{\unfoldinfred}{\mathrel{\sunfoldinfred}}

\newcommand{\unfsem}[1]{\llbracket{#1}\rrbracket_{\hspace*{-0.1pt}{\txtinflambdacal}}}
\newcommand{\sunfsem}{\unfsem\cdot}
\newcommand{\unfsembh}[1]{\llbracket{#1}\rrbracket_{\hspace*{-0.1pt}{\txtinflambdabhcal}}}
\newcommand{\sunfsembh}{\unfsembh{\cdot}}

\newcommand{\graphsemC}[2]{\llbracket{#2}\rrbracket_{#1}}
\newcommand{\graphsem}[1]{\llbracket{#1}\rrbracket}

\newcommand{\sgraphsemC}[1]{\graphsemC{#1}{\cdot}}

\newcommand{\noSsh}{\inMath{\normalfont\textsf{min}}}
\newcommand{\graphsemCmin}[2]{\graphsemC{#1}{#2}^{\noSsh}}
\newcommand{\sgraphsemCmin}[1]{\graphsemCmin{#1}{\cdot}}

\newcommand{\sreadbackC}[1]{\textsf{rb}_{#1}}

\newcommand{\sreadback}{{\normalfont \textsf{rb}}}
\newcommand{\readback}{\funap{\sreadback}}

\newcommand{\sTree}{\mathit{Tree}}
\newcommand{\Tree}{\funap{\sTree}}
\newcommand\thsp{-1.85ex} 
\newcommand\threeheadrightarrow{\twoheadrightarrow\hspace*\thsp\twoheadrightarrow}

\newcommand{\sinfred}{\threeheadrightarrow}

\newcommand{\saequivrel}{\sim}

\definecolor{azure}{rgb}{0.94,1.00,1.00}
\definecolor{blue}{rgb}{0,0,0.5}
\definecolor{brown}{rgb}{.75,.25,.25}
\definecolor{cyan}{rgb}{0.25,0.88,0.82}
\definecolor{chocolate}{rgb}{0.82,0.41,0.12}
\definecolor{darkcyan}{rgb}{0.5,0,1}
\definecolor{darkgreen}{rgb}{0,0.39,0}
\definecolor{darkmagenta}{rgb}{0.5,0,0.5}
\definecolor{firebrick}{RGB}{175,25,25}
\definecolor{forestgreen}{rgb}{0.13,0.55,0.13}
\definecolor{lightcyan}{rgb}{0.88,1.00,1.00}
\definecolor{lightpink}{rgb}{1.00,0.71,0.76}
\definecolor{lightyellow}{rgb}{1.00,1.00,0.88}
\definecolor{lightgoldenrod}{rgb}{0.83,0.97,0.51}
\definecolor{lightgoldenrodyellow}{rgb}{0.98,0.98,0.82}
\definecolor{lightskyblue}{rgb}{0.53,0.81,0.98}
\definecolor{moccasin}{rgb}{1.00,0.89,0.71}
\definecolor{magenta}{rgb}{1,0,1}
\definecolor{navyblue}{rgb}{0,0,0.5}
\definecolor{orange}{rgb}{1.0,0.65,0.0}
\definecolor{orangered}{rgb}{1.0,0.27,0.0}
\definecolor{palegreen}{rgb}{0.60,0.98,0.60}
\definecolor{powderblue}{rgb}{0.69,0.88,0.90}
\definecolor{purple}{rgb}{1,0.5,1}
\definecolor{royalblue}{RGB}{65,105,225}
\definecolor{mediumblue}{RGB}{0,0,205}
\definecolor{cornflowerblue}{RGB}{100,149,237}
\definecolor{springgreen}{rgb}{0.0,1.0,0.5}
\definecolor{turquoise}{rgb}{0.25,0.88,0.82}
\definecolor{snow}{rgb}{1.00,0.98,0.98}
\definecolor{tan}{rgb}{0.82,0.71,0.55}
\definecolor{red}{rgb}{1,0,0}

\newcommand{\sbigO}{O}
\newcommand{\bigO}{\funap{\sbigO}}
\newcommand{\atg}{G}
\newcommand{\atgi}{\indap{\atg}}
\newcommand{\atgacc}{G'}
\newcommand{\atgacci}{\indap{\atgacc}}
\newcommand{\atgdacc}{G''}
\newcommand{\atgdacci}{\indap{\atgdacc}}
  
\newcommand{\descsetexpmid}{\mathrel{\vert}}

\newcommand{\descsetexp}[2]{\left\{{#1}\descsetexpmid{#2}\right\}}

\newcommand{\setexp}[1]{\left\{{#1}\right\}}

\newcommand{\factorset}[2]{{#1}/_{#2}}

\newcommand{\nats}{\mathbb{N}}

\newcommand{\bpos}{q}

\newcommand{\verts}{V}
\newcommand{\vertsof}{\funap{\verts\!}}
\newcommand{\svlab}{\mathit{lab}}
\newcommand{\vlab}{\funap{\svlab}}
\newcommand{\svargs}{\mathit{args}}
\newcommand{\vargs}{\funap{\svargs}}
\newcommand{\sroot}{\mathit{r}}
\newcommand{\vertsi}{\indap{V}}

\newcommand{\svlabi}{\indap{\mathit{lab}}}
\newcommand{\vlabi}[1]{\funap{\svlabi{#1}}}
\newcommand{\svargsi}{\indap{\mathit{args}}}
\newcommand{\vargsi}[1]{\funap{\svargsi{#1}}}
\newcommand{\srooti}{\indap{\mathit{r}}}

\newcommand{\avert}{v}
\newcommand{\bvert}{w}

\newcommand{\averti}{\indap{\avert}}
\newcommand{\bverti}{\indap{\bvert}}

\newcommand{\bvertbp}{\bpap{\bvert}}

\newcommand{\avertacc}{v'}
\newcommand{\bvertacc}{w'}

\newcommand{\bvertacci}{\indap{\bvertacc}}

\newcommand{\blackhole}{\bullet}
\newcommand{\bh}{\blackhole}

\newcommand{\siglambdabh}{\pbap{\asig}{\sslabs}{\blackhole}}
\newcommand{\siglambdaSbh}{\pbap{\asig}{\sslabs}{\snlvarsucc,\blackhole}}

\newcommand{\siglambdacalCRS}{\asig_{\txtlambdacal}}
\newcommand{\siglambdabhcalCRS}{\asig_{\txtlambdabhcal}}
\newcommand{\siglambdaletreccalCRS}{\asig_{\txtlambdaletreccal}}
\newcommand{\siglambdaletrecbhcalCRS}{\asig_{\txtlambdaletrecbhcal}}

\newcommand{\stgsucc}{{\rightarrowtail}}
\newcommand{\tgsucc}{\mathrel{\stgsucc}}
\newcommand{\stgsucci}{\indap{\rightarrowtail}}
\newcommand{\tgsucci}[1]{\mathrel{\stgsucci{#1}}}

\newcommand{\sahom}{h}
\newcommand{\ahom}{\funap{\sahom}}
\newcommand{\sahomext}{{\sahom}^*}

\newcommand{\saiso}{i}

\newcommand{\siso}{{\simeq}}
\newcommand{\iso}{\mathrel{\siso}}

\newcommand{\sbisim}{%
    \setbox0=\hbox{\kern-.1ex{$\leftrightarrow$}\kern-.1ex}
    \setbox1=\vbox{\hbox{\raise .1ex \box0}\hrule}%
    \inMath{\mathrel{\hbox{\kern.1ex\box1\kern.1ex}}}
  }
\newcommand{\bisim}{\mathrel{\sbisim}}

\newcommand{\sbisims}{\supap{\sbisim}}

\newcommand{\sbisimS}{\sbisims{\snlvarsucc}}

\newcommand{\sbisimsubscript}{
    \setbox0=\hbox{\kern-.1ex{$\leftrightarrow$}\kern-.1ex}
    \setbox1=\vbox{\hbox{\raise .1ex \box0}\hrule}%
    \inMath{\mathrel{\hbox{\scalebox{0.75}{\box1}}}}
  }

\newcommand{\sbisimsubscripts}[1]{{\sbisimsubscript^{#1}}}
\newcommand{\sbisimsubscriptS}{\sbisimsubscripts{\snlvarsucc}}

\newcommand{\sinvfunbisim}{%
    \setbox0=\hbox{\kern-.1ex{$\leftarrow$}\kern-.1ex}
    \setbox1=\vbox{\hbox{\raise .1ex \box0}\hrule}%
    {\hbox{\kern.1ex\box1\kern.1ex}}
  }

\newcommand{\sfunbisim}{%
    \setbox0=\hbox{\kern-.1ex{$\rightarrow$}\kern-.1ex}
    \setbox1=\vbox{\hbox{\raise .1ex \box0}\hrule}%
    {\hbox{\kern.1ex\box1\kern.1ex}}
  }
\newcommand{\funbisim}{\mathrel{{\sfunbisim}}}

\newcommand{\sfunbisims}{\supap{\sfunbisim}}

\newcommand{\sfunbisimS}{\sfunbisims{\snlvarsucc}}
\newcommand{\funbisimS}{\mathrel{\sfunbisimS}}

\newcommand{\sconvfunbisim}[1][]{%
    \setbox0=\hbox{\kern-.1ex{$\leftarrow$}\kern-.1ex}
    \setbox1=\vbox{\hbox{\raise .1ex \box0}\hrule}%
    \mathrel{\hbox{\kern.1ex\box1\kern.1ex}}
  }
\newcommand{\convfunbisim}{\mathrel{\sconvfunbisim}}

\newcommand{\sconvfunbisims}{\supap{\sconvfunbisim}}

\newcommand{\sconvfunbisimS}{\sconvfunbisims{\snlvarsucc}}

\newcommand{\sSfunbisim}{\sfunbisim^{\snlvarsucc}}
\newcommand{\Sfunbisim}{\mathrel{\sSfunbisim}}

\newcommand{\Sunsharing}{$\snlvarsucc$\nb-un\-sha\-ring}
\newcommand{\sSunsh}{\textsf{\normalfont unsh}_{\snlvarsucc}}

\newcommand{\abisim}{R}

\newcommand{\classlhotgs}{{\cal H}}
\newcommand{\classlhotreetgs}{\indap{\classlhotgs}{T}} 

\newcommand{\classltgs}{{\cal T}} 
\newcommand{\classeagltgs}{\subap{{\cal T}}{\text{\normalfont eag}}}

\newcommand{\slhotgstoltgs}{{\cal HT}}
\newcommand{\lhotgstoltgs}{\funap{\slhotgstoltgs}}
\newcommand{\sltgstolhotgs}{{\cal TH}}
\newcommand{\ltgstolhotgs}{\funap{\sltgstolhotgs}}

\newcommand{\sabspre}{P}
\newcommand{\abspre}{\funap{\sabspre}}
\newcommand{\sabsprei}{\indap{P}}
\newcommand{\absprei}[1]{\funap{\sabsprei{#1}}}
\newcommand{\sabspreacc}{P'}

\newcommand{\sabspredacc}{P''}

\newcommand{\absprefix}{ab\-strac\-tion-pre\-fix}
\newcommand{\absprefixes}{ab\-strac\-tion-pre\-fixes}

\newcommand{\absprefixfunction}{\absprefix\ function}
\newcommand{\absprefixfunctions}{\absprefix\ functions}

\newcommand{\apre}{p}

\newcommand{\prele}{\le}

\newcommand{\aclass}{{\cal K}}

\newcommand{\aspantree}{T}

\newcommand{\aspantreedacc}{\aspantree''}

\newcommand{\toolfig}[1]{\includegraphics[scale=0.4]{tool-output/{{#1}}}}

\newcommand\tooltip[3]{
\begin{textblock}{15}#1
  \fbox{\parbox{#2}{#3}}
\end{textblock}
}

\newcommand\showcase[5]{
\begin{textblock}{15}#4
  \toolfig{#3/#1_#2.l.dfa}
\end{textblock}
\ifx\\#5\\
\else
\begin{textblock}{15}#5
  \toolfig{#3/#1_#2.l.dfa.min}
\end{textblock}
\fi
\verbfilebox{tool-output/#3/#1_#2.l.log}
\hspace{-0.5cm}
\scalebox{0.73}[0.8]{\theverbbox}
}

\begin{document}

\setlength\itemsep{-0.5ex}

\titlebanner{Maximal Sharing in the Lambda Calculus with \txtletrec} 

\title{Maximal Sharing in the Lambda Calculus with \txtletrec\hspace*{1.5pt}\footnotemark}


\authorinfo{Clemens Grabmayer}
           {Dept.\ of Computer Science, VU University Amsterdam\\
            de Boelelaan 1081a, 1081 HV Amsterdam}
           {c.a.grabmayer@vu.nl}
\authorinfo{Jan Rochel}
           {
            Dept.\ of Computing Sciences, Utrecht University\\
            Princetonplein 5, 
            3584 CC Utrecht, The Netherlands}
           {jan@rochel.info}

\conferenceinfo{ICFP 2014}{September 1--3  2014, Gothenburg, Sweden.}
\copyrightyear{2014}
\copyrightdata{[COPYRIGHTDATA]}

\maketitle

\begin{abstract} 
  Increasing sharing in programs is desirable to compactify the code, and to
  avoid duplication of reduction work at run-time, thereby speeding up
  execution.
  We show how a maximal degree of sharing can be obtained for programs
  expressed as terms in 
  the lambda calculus with $\txtletrec$.  
  We introduce a notion of `maximal compactness' for \lambdaletreccalterms\
  among all terms with the same infinite unfolding.
  Instead of defined purely syntactically, this notion is based 
  on a graph semantics. 
  \lambdaletreccalterms\ are interpreted as first-order \termgraphs\
  so that \unfoldingequivalence\ between terms is preserved and reflected 
  through bisimilarity of the \termgraph\ interpretations.
  %
  Compactness of the \termgraphs\ can then be compared via functional bisimulation.
 
  We describe practical and efficient methods for the following two problems:
  transforming a \lambdaletreccalterm\ into a maximally compact form;
  and deciding whether two
  \lambdaletreccalterms\ are \unfoldingequivalent.  
  The transformation of a \lambdaletreccalterm~$\allter$ into
  maximally compact form $\allteri{0}$ proceeds in three steps:  
  \begin{enumerate*}[label=(\roman*)]
  \item translate $\allter$ into its
        \termgraph\ $\atg = \graphsem{\allter}\,$;
  \item compute the maximally shared form of
  $\atg$ as its bisimulation collapse $\atgi{0}\,$;
  \item read back a \lambdaletreccalterm~$\allteri{0}$ from the \termgraph\
  $\atgi{0}$ with the property $\graphsem{\allteri{0}} = \atgi{0}$.
  \end{enumerate*}
  This guarantees that $\allteri{0}$ and $\allter$ have the same unfolding, 
  and that $\allteri{0}$ exhibits maximal sharing.
  
  The procedure for deciding 
  whether two given \lambdaletreccalterms~$\allteri{1}$ and $\allteri{2}$ are \unfoldingequivalent\ 
  computes their \termgraph\ interpretations $\graphsem{\allteri{1}}$ and $\graphsem{\allteri{2}}$,
  and checks whether these term graphs are bisimilar.
  
  For illustration,  we also provide a readily usable implementation. 
\end{abstract}

\footnotetext{  
  \hspace{-0.4em}%
    This work was supported by NWO in the framework of the
    project \emph{Realising Optimal Sharing (ROS)}, project number 612.000.935.
    }

\category{D.3.3}{Language constructs and features}{Recursion}
\category{F.3.3}{Studies of Programming Constructs}{Functional constructs}

\terms
functional programming, compiler optimisation

\keywords
Lambda Calculus with \txtletrec, unfolding semantics, subterm sharing, maximal sharing,
higher-order \termgraphs\

\section{Introduction}
  \label{sec:intro}

Explicit sharing in pure functional programming languages
is typically expressed by means of the \txtletrec\ construct,
which facilitates cyclic definitions.
The \lambdacalculus\ with \txtletrec, \lambdaletreccal\
forms a syntactic core of these languages, and it can be viewed as their abstraction.
As such \lambdaletreccal\ is well-suited as a test bed for developing
program transformations in functional programming languages. 
This certainly holds for the transformation presented here
that has a strong conceptual motivation, is justified by a form of semantic reasoning,
and is best described first for an expressive, yet minimal language.



\subsection{Expressing sharing and infinite \lambdaterms}
  \label{sec:intro:subsec:sharing:lambdaletreccal}

For the programmer the $\txtletrec$-construct offers the possibility to write a
program more compactly by utilising subterm sharing. \letrecexpression{s} bind
subterms to variables; these variables then denote occurrences of the
respective subterms and can be used anywhere inside of the \letrecexpression\ 
(also recursively).
In this way, instead of repeating a subterm multiple times, 
a single definition can be given which is then referenced from multiple positions.

We will denote the $\txtletrec$-construct here by $\txtlet$ as in Haskell. 



\begin{example}
               \label{ex:cse}
  Consider the \lambdaterm~$\lapp{(\labs{x}{x})}{(\labs{x}{x})}$
  with two occurrences of the subterm $\labs{x}{x}$.
  These occurrences can be shared with as result the \lambdaletreccalterm\ 
  $(\letin{id=\labs{x}{x}}{\lapp{id}{~id}})$.
\end{example}

\noindent
As \letbindings\ permit definitions with cyclic dependencies, terms in \lambdaletreccal\ are able
to finitely denote infinite \lambdaterms\ (for short: \inflambdaterms).
The \inflambdaterm\ $\aiter$ represented by a \lambdaletreccalterm\ $\allter$
can be obtained by a typically infinite process in which the \letbindings\ in $\allter$
are unfolded continually with $\aiter$ as result in the limit.
Then we say that $\aiter$ is the \emph{infinite unfolding} of $\allter$, or
that $\aiter$ is the denotation of $\allter$ in the \emph{unfolding semantics},
indicated symbolically by $\aiter = \unfsem{\allter}$. 

\begin{example}
               \label{ex:fix}
  For the \lambdaletreccalterms~$\allter$ and $\bllter$
  and the \inflambdaterm~$\aiter$:
\begin{align*}
  &
  \begin{aligned}
  \allter &\defdby \labs{f}{\letin{r=\lapp{f}{r}}{r}} \\
  \bllter &\defdby \labs{f}{\letin{r=\lapp{f}{(\lapp{f}{r})}}{r}}
  \end{aligned}
  &
  \aiter  &\defdby  \labs{f}{\lapp{f}{(\lapp{f}{(\dots)})}}
\end{align*}
it holds that both $\allter$ and $\bllter$
(which represent fixed-point combinators)
have $\ater$ as their infinite unfolding:
$\unfsem{\allter} = \unfsem{\bllter}  = \aiter$.
\end{example}

$\allter$ and $\bllter$ in this example are `unfolding equivalent'. 
Note that $\allter$ represents $\aiter$ in a more compact way than $\bllter$. 
It is intuitively clear that there is no \lambdaletreccalterm\ that
represents $\aiter$ more compactly than $\allter$.
So $\allter$ can be called a `maximally shared form' of~$\bllter$ (and of $\ater$).

We address, and efficiently solve, the problems of computing the maximally shared form
of a \lambdaletreccalterm, and of determining whether two
\lambdaletreccalterms\ are \unfoldingequivalent. 
Note that these notions are based on the static unfolding semantics.
We \emph{do not consider}
any dynamic semantics based on evaluation by \betareduction\ or otherwise.

\subsection{Recognising potential for sharing}
  \label{sec:intro:subsec:our:methods}

A general risk for compilers of functional programs is 
``[to construct] multiple instances of the same expression, rather than sharing a single copy of them.
  This wastes space because each instance occupies separate storage, and it wastes time because the
  instances will be reduced separately. This waste can be arbitrarily large, [\ldots]''
  (\cite[p.243]{peyt:jone:1987}).
%
Therefore practical compilers increase sharing, and do so typically
for supercombinator translations of programs (such as fully-lazy \lambdalifting).
Thereby two goals are addressed:
to increase sharing based on a syntactical analysis of the `static' form of the program; 
and to prevent splits into too many supercombinators
   when an anticipation of the program's `dynamic' behaviour
   is able to conclude that no sharing at run-time will be gained.

A well-known method for the `static' part is
common subexpression elimination (CSE) \cite{chit:1997:CS-uncommon}. 
For the `dynamic' part, a predictive syntactic program analysis has been proposed 
for fine-tuning sharing of partial applications in supercombinator translations \cite{gold:huda:1987}.

We focus primarily on the `static' aspect of introducing sharing. 
We provide a conceptual solution that substantially extends CSE. 
But instead of maximising sharing for a supercombinator translation of a program,
we carry out the optimisation on the program itself (the \lambdaletrecterm). 
And instead of applying a purely syntactical program analysis,
we use a \termgraph\ semantics for \lambdaletreccalterms.

\subsection{Approach based on a \termgraph\  semantics}
  \label{sec:intro:subsec:approach}

We develop a combination of techniques 
for realising maximal sharing in \lambdaletreccalterms.
For this we proceed in four steps:
\lambdaletreccalterms\ are interpreted as higher-order
\termgraphs; the higher-order \termgraphs\ are implemented as first-order \termgraphs;
maximally compact versions of such \termgraphs\ can be computed by standard algorithms;
\lambdaletreccalterms\ that represent compacted \termgraphs\ (or in fact arbitrary ones) can be
retrieved by a `readback' operation.   

In more detail, the four essential ingredients are the following: 
  \renewcommand{\descriptionlabel}[1]{\hspace{\labelsep}{\normalfont{#1}}}
\begin{description}
   \item[(1)] 
     A \emph{semantics}~$\sgraphsemC{\classlhotgs}$ for
     interpreting \lambdaletreccalterms\ as \emph{higher-order \termgraphs},
     which are first-order \termgraphs\ enriched with a feature for describing
     binding and scopes.
     We call this specific kind of higher-order \termgraphs\
     `\lambdahotgs'. 
\end{description} 
%
The variable binding structure is recorded in this \termgraph\ concept 
because it must be respected by any addition of sharing. 
The \termgraph\ interpretation 
                              adequately represents 
sharing as expressed by a \lambdaletreccalterm.
It is not injective: a \lambdahotg\ typically is the interpretation
of various \lambdaletrecterms. 
Different degrees of sharing as expressed 
                                          by \lambdaletrecterms\
can be compared via the \lambdahotg\ interpretations by a sharing preorder, 
which is defined as the existence of a homomorphism (functional bisimulation).

While comparing higher-order \termgraphs\
via this preorder
is 
   computable in principle,
standard algorithms do not apply.  
Therefore efficient solvability of the compactification problem and the comparison problem  
is, from the outset, not guaranteed. 
For this reason we devise a first-order 
                                        implementation of \lambdahotgs:
\begin{description}
    \item[(2)]
      An \emph{interpretation}~$\slhotgstoltgs$ of \lambdahotgs\ into a
      specific kind of \emph{first-order \termgraphs}, which we call `\lambdatgs'. 
      It preserves and reflects the sharing preorder.
%
\end{description}    
%
$\slhotgstoltgs$ 
reduces bisimilarity between \lambdahotgs\ (higher-order) to bisimilarity between \lambdatgs\
(first-order), and 
                   facilitates:
\begin{description}
  \item[(3)]\label{ingredients:item:bisimilarity:methods}
    The use of standard methods for \emph{checking} bisimilarity and for
    computing the bisimulation \emph{collapse} of \lambdatgs. 
    Via $\slhotgstoltgs$ also the analogous problems for \lambdahotgs\ can be solved.
\end{description} 
%
\Termgraphs\ can be represented as deterministic process graphs (labelled transition systems),
and even as deterministic finite-state automata (DFAs).
That is why it is possible to apply efficient algorithms for state minimisation and language equivalence of~DFAs.

Finally, an operation to return from \termgraphs\ to \lambdaletreccalterms:      
\begin{description}
  \item[(4)]
    A \emph{readback} function $\sreadback$ from \lambdatgs\ to \lambdaletreccalterms\ that,
    for every \lambdatg~$\altg$, computes a \lambdaletreccalterm\ $\allter$
    from the set of \lambdaletrecterms\ that have $\altg$ as their 
    interpretation via $\sgraphsemC{\classlhotgs}$ and $\slhotgstoltgs$
    (i.e.\ a \lambdaletrecterm\ for which it holds that $\lhotgstoltgs{\graphsemC{\classlhotgs}{\allter}} = \atg$).
\end{description}

\subsection{Methods and their correctness}
  \label{sec:intro:subsec:methods}

On the basis of the concepts above we develop efficient methods for
introducing maximal sharing, and for checking unfolding equivalence, of
\lambdaletreccalterms, as sketched below. 

In describing these methods, we use the following notation:
\begin{description}
  \item[$\classlhotgs\:$:] 
     class of \lambdahotgs, the image of the semantics~$\sgraphsemC{\classlhotgs}\,$; 
  \item[$\classltgs\:$:] 
     class of \lambdatgs, the image of the interpretation~$\slhotgstoltgs\,$;  
 \item[$\sgraphsemC{\classltgs} \defdby \scompfuns{\slhotgstoltgs}{\sgraphsemC{\classlhotgs}}\:$:]
    first-order \termgraph\ semantics for \lambdaletreccalterms; 
  \item[$\scoll\:$:] 
    bisimulation collapse on $\classlhotgs$ and $\classltgs$;
  \item[$\sreadback\:$:]
    readback mapping from \lambdatgs\ to \lambdaletreccalterms.   
\end{description}
%
%
%
We obtain the following methods (for illustrations, see Fig.~\ref{fig:methods}):
\begin{itemize}\setlength{\itemsep}{0.1ex}
  \item[$\smalltriangleright$] 
    \emph{Maximal sharing}:
      for a given \lambdaletreccalterm, a maximally shared form
      can be obtained by collapsing its first-order \termgraph\ interpretation,
      and then reading back the collapse:
      $\scompfuns{\scompfuns{\sreadback}{\scoll}}{\sgraphsemC{\classltgs}}\,$
      %
  \item[$\smalltriangleright$] 
    \emph{Unfolding equivalence}: 
    for given \lambdaletreccalterms~$\allter$ and $\bllter$,
    it can be decided whether $\unfsem{\allter} = \unfsem{\bllter}$ by checking
    whether their \termgraph\ interpretations $\graphsemC{\classltgs}{\allter}$ and $\graphsemC{\classltgs}{\bllter}$ are bisimilar.
\end{itemize}
\begin{figure}[t]
\begin{align*}
\begin{tikzpicture}[>=stealth]
\matrix[row sep=0.5cm,column sep=0.7cm,ampersand replacement=\&]{
  \node(L){$\allter$};\&
  \node(G){$\alhotg$};\&
  \node(g){$\altg$};
  \\
  \node[xshift=3mm](M){$\ater$};
  \\
  \node(L0){$\allteri{0}$};\&
  \node(G0){$\alhotgi{0}$};\&
  \node(g0){$\altgi{0}$};
  \\
};
\draw[|->] (L) to node[left]{$\sunfsem$} (M);
\draw[|->] (L0) to node[left]{$\sunfsem$} (M);
\draw[|->] (L) to node[above]{$\;\;\;\;\sgraphsemC{\classlhotgs}$} (G);
\draw[|->] (G) to node[above]{$\slhotgstoltgs\;\;\,$} (g);
\draw[bend left=40,|->,thick] (L) to node[above]{$\sgraphsemC{\classltgs}$} (g);
\draw[|->] (L0) to node[above]{$\;\;\;\;\sgraphsemC{\classlhotgs}$} (G0);
\draw[|->] (G0) to node[above]{$\slhotgstoltgs\;\;$} (g0);
\draw[bend left=32,|->,thick] (g0) to node[below]{$\sreadback$} (L0);
\draw[bend left=41,|->] (L0) to node[above]{$\sgraphsemC{\classltgs}\hspace*{5ex}\mbox{}$} (g0);
\draw[funbisim] (G) to node[right]{$\scoll$} (G0);
\draw[funbisim,very thick] (g) to node[right]{$\scoll$} (g0);
\end{tikzpicture}
& \hspace*{3ex}
\begin{tikzpicture}[>=stealth]
\matrix[row sep=0.5cm,column sep=0.7cm,ampersand replacement=\&]{
  \node(L1){$\allteri1$};\&
  \node(G1){$\alhotgi1$};\&
  \node(g1){$\altgi1$};
  \\
  \node[xshift=3mm](M){$\ater$};
  \\
  \node(L2){$\allteri{2}$};\&
  \node(G2){$\alhotgi{2}$};\&
  \node(g2){$\altgi{2}$};
  \\
};
\draw[|->] (L1) to node[left]{$\sunfsem$} (M);
\draw[|->] (L2) to node[left]{$\sunfsem$} (M);
\draw[|->] (L1) to node[above]{$\;\;\;\;\sgraphsemC{\classlhotgs}$} (G1);
\draw[|->] (G1) to node[above]{$\slhotgstoltgs\;\;$} (g1);
\draw[|->] (L2) to node[above]{$\;\;\;\;\sgraphsemC{\classlhotgs}$} (G2);
\draw[|->] (G2) to node[above]{$\slhotgstoltgs\;\;$} (g2);
\draw[bisim] (G1) to (G2);
\draw[bisim,very thick] (g1) to (g2);
\draw[bend left=40,|->,thick] (L1) to node[above]{$\sgraphsemC{\classltgs}$} (g1);
\draw[bend left=40,|->,thick] (L2) to node[above]{$\sgraphsemC{\classltgs}\hspace*{6ex}\mbox{}$} (g2);
\draw[color=white,bend left=35,|->] (g2) to node[below]{$\sreadback$} (L2); 
\end{tikzpicture}   
\end{align*}
  \caption{\label{fig:methods}
           Component-step build-up of the methods for
           computing a maximally shared form $\allteri{0}$ of a \lambdaletreccalterm~$\allter$ (left),
           and for deciding unfolding equivalence of \lambdaletreccalterms~$\allteri{1}$ and $\allteri{2}$
           via bisimilarity~$\sbisim$ (right).}     
\end{figure}
\begin{figure}
\hspace{-2ex}
\begin{tikzpicture}
\matrix[row sep=-4mm,column sep=0.3cm]{
  \node(L1){$\labs{f}{\letin{r=\lapp{f}{r}}{r}}$}; &&
  \node(L2){$\labs{f}{\letin{r=\lapp{f}{(\lapp{f}{r})}}{r}}$};\\[5ex]
&\node(M){$\labs{f}{\lapp{f}{(\lapp{f}{(\dots)})}}$}; \\
\node(G1){\figsmall{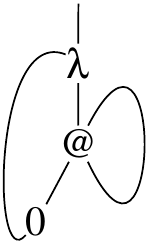}}; && \node(G2){\figsmall{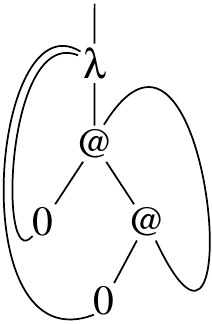}};\\
};
\draw[funbisim,very thick] (G2) to node[above]{$\scoll$} (G1);
\draw[|->] (L1) to node[above,near end]{$\sunfsem$} (M);
\draw[|->,thick] (G1.100) to node[left]{$\sreadback$} (L1.235);
\draw[|->] (L2) to node[above,near end]{$\sunfsem$} (M);
\draw[|->] (L1) to node[right]{$\sgraphsemC{\classltgs}$} (G1);
\draw[|->,thick] (L2) to node[right]{$\sgraphsemC{\classltgs}$} (G2);
\end{tikzpicture}
\caption{\label{fig:ex:compact}
Computing a maximally compact version of the term $\bllter$ from Ex.~\ref{ex:fix} (right)
by using composition of \termgraph\ semantics $\sgraphsemC{\classltgs}$, collapse $\scoll$, and readback $\sreadback$,
yielding the term $\allter$ (left).} 
\end{figure}
See Fig.~\ref{fig:ex:compact} for an illustration of the application of the
maximal sharing method to the \lambdaletreccalterms~$\allter$ and $\bllter$ from Ex.~\ref{ex:fix}. 

The correctness of these methods hinges on the fact that 
the \termgraph\ translation and the readback satisfy the following properties: 
\begin{enumerate}[label=(P\arabic*)]\setlength{\itemsep}{0.1ex}
\item{} \label{methods:properties:correctness}
      \lambdaletreccalterms~$\allter$ and $\bllter$ have the same infinite unfolding 
      if and only if the \termgraphs\ $\graphsemC{\classltgs}{\allter}$ and
      $\graphsemC{\classltgs}{\bllter}$ are bisimilar.
\item{} \label{methods:properties:closedness}
      The class $\classltgs$ of \lambdatgs\ is closed under homomorphism. 
\item{} \label{methods:properties:readback}
      The readback $\sreadback$
      is a right inverse of $\sgraphsemC{\classltgs}$ up to isomorphism~$\siso$,
      that is,
      for all \termgraphs\ $\atg\,{\in}\,\classltgs$ it holds: 
      $\compfuns{(\sgraphsemC{\classltgs}}{\sreadback)}{\atg} \iso \atg$.
\end{enumerate}
\emph{Note}: \ref{methods:properties:closedness} and \ref{methods:properties:readback}
 will be established only for a subclass $\classeagltgs$ of $\classltgs$.
Furthermore, practicality of these methods depends on the property:
\begin{enumerate}[label=(P\arabic*)]\setcounter{enumi}{3}
\item{} \label{methods:properties:efficiency}
      Translation $\sgraphsemC{\classltgs}$ and readback $\sreadback$ are efficiently computable. 
\end{enumerate}

\subsection{Overview of the development}
  \label{sec:intro:subsec:overview}

In the Preliminaries (Section~\ref{sec:prelims}) we fix basic notions and notations for first-order \termgraphs.
\lambdaletrecterms\ and their unfolding semantics are defined in Section~\ref{sec:unfsem}.
In Section~\ref{sec:lhotgs} we develop the concept of `\lambdahotg', which gives rise to the class $\classlhotgs$,
and the higher-order \termgraph\ semantics ~$\sgraphsemC{\classlhotgs}$ 
                                                          for \lambdaletreccalterms.

In Section~\ref{sec:ltgs} 
we develop the concept of first-order `\lambdatg' in the class $\classltgs$,
and define the interpretation~$\slhotgstoltgs$ of \lambdahotgs\ into \lambdatgs\
as a mapping 
             from $\classlhotgs$ to $\classltgs$.
This 
induces the first-order \termgraph\ semantics 
$\sgraphsemC{\classltgs} \defdby \scompfuns{\slhotgstoltgs}{\sgraphsemC{\classlhotgs}}$,
for which we also provide a direct inductive definition.   

In Section~\ref{sec:readback} we define the readback $\sreadback$ with
the desired property as a function from \lambdatgs\ to \lambdaletrecterms. 
Subsequently in Section~\ref{sec:complexity}
we report on the complexity of the described methods, individually,
and in total for the methods described in Subsection~\ref{sec:intro:subsec:methods}.


In Section~\ref{sec:implementation} we link to our implementation of the presented methods. 
Finally in Section~\ref{sec:applications} we explain easy modifications, describe possible extensions, 
and sketch potential practical applications.

\subsection{Applications and scalability }
  \label{sec:intro:subsec:applications}

While our contribution is at first a conceptual one, it holds the promise
for a number of practical applications:
\begin{itemize}
\item 
  Increasing the efficiency of the execution of programs by transforming them
  into their maximally shared form at compile-time.
\item 
  Increasing the efficiency of the execution of programs by repeatedly
  compactifying the program at run time.
\item 
  Improving systems for recognising program equivalence.
\item 
  Providing feedback to the programmer, along the lines:
  `This code has identical fragments and can be written more compactly.' 
\end{itemize}
These and a number of other potential applications are discussed in more detail
in Section~\ref{sec:applications}.

The presented methods scale well to larger inputs,
due to the quadratic bound on their runtime complexity
(see Section~\ref{sec:complexity}).

\subsection{Relationship with other concepts of sharing}
  \label{sec:intro:subsec:related}

The maximal sharing method is targeted at increasing `static' sharing:
in the sense that
a program is transformed at compile time into a version with a
higher degree of sharing. It is not (at least not a priori) a method for `dynamic' 
sharing, i.e.\ for an evaluator that maintains a certain degree of sharing
at run time, such as graph rewrite mechanisms 
                                              for fully-lazy
\cite{wads:1971} or optimal evaluation \cite{
                                             aspe:guer:1998} of the
\lambdacalculus. 
However, we envisage run-time collapsing of the program's
graph 
      interpretation integrated with the evaluator (see Section \ref{sec:applications}).

The term `maximal sharing' stems from work on the ATERM library
\cite{
      bran:klin:2007}.
It describes a technique for minimising memory usage when representing a set of
terms in a first-order term rewrite system (\TRS). The terms are kept in an
aggregate directed acyclic graph by which their syntax trees are shared as much
as possible. 
Thereby terms are created only if they are entirely new; otherwise they are
referenced by pointers to roots of sub-dags.
Our use of the expression `maximal sharing' is inspired by that work, but our
results generalise that approach in the following ways:
\begin{itemize}
  \item Instead of first-order terms we consider terms 
    in a higher-order language
    with the $\txtletrec$-construct for expressing sharing. 
  \item 
    Since \txtletrec\ typically defines cyclic sharing dependencies, 
    we interpret terms as cyclic graphs instead of just dags. 
  \item 
    We are interested in increasing sharing by bisimulation collapse
    instead of by identifying isomorphic sub-dags. 
\end{itemize}
ATERM only checks for equality of subexpressions. Therefore it only introduces 
horizontal sharing and implements a form of
\emph{common subexpression elimination (\CSE)}
\cite[p.\hspace*{1.5pt}241]{peyt:jone:1987}. Our approach is stronger than
\CSE: while Ex.~\ref{ex:cse} can be handled by \CSE, this is not the case for
Ex.~\ref{ex:fix}. In contrast to CSE, our approach increases 
also vertical and twisted sharing
\footnote{For definitions of horizontal, vertical, and twisted sharing we refer to \cite{blom:2001}.}.
(see also \cite{blom:2001}).


\subsection{Contribution of this paper in context}
  \label{sec:intro:subsec:context:our:work}
%
Blom introduces \emph{higher-order \termgraphs} \cite{blom:2001}, which are
extensions of first-order \termgraphs\ by adding a scope function that assigns
a set of vertices, its \emph{scope}, to every abstraction vertex.

As a stepping stone for the methods we develop here,
we use concepts and results that we described in an earlier
paper \cite{grab:roch:2013:a:TERMGRAPH}.
There, for interpreting \lambdaletrecterms, a modification of Blom's higher-order \termgraphs\ 
(the \lambdahotgs of the class $\classlhotgs$)
in which scopes are represented by means of `abstraction prefix functions'.
We also investigated first-order \lambdatgs\ with scope-delimiter vertices (corresponding to the class $\classltgs$ here).
In particular we examined which specific class of first-order
\lambdatgs\ can faithfully represent the higher-order \lambdahotgs\ in such a
way that compactification of the latter can be realised through bisimulation
collapse of the former (this led to the \lambdatgs\ of the class $\classltgs$).

Whereas in the paper \cite{grab:roch:2013:a:TERMGRAPH}
we exclusively focused on the graph formalisms,
         and investigated them in their own right,
here we connect the results obtained there to
the language \lambdaletreccal\ for expressing sharing and
cyclicity. Since the methods presented here are based on the graph formalisms,
and rely on their properties for correctness, 
we recapitulate the concepts and the relevant results in
Sec.~\ref{sec:lhotgs} and \ref{sec:ltgs}.

The translation $\sgraphsemC\classltgs$ of \lambdaletrecterms\ into first-order \termgraphs\ 
was inspired by related representations that use scope delimiters to
indicate end of scopes.
Such representations are generalisations 
of a de~Bruijn index notation for \lambdaterms\ \cite{de-bruijn:72}
in which the de~Bruijn indexes are numerals of the form
$\nlvarsucc{(\ldots(\nlvarsucc{(\snlvar)})\ldots)}$. 
In the generalised form, due to Patterson and Bird~\cite{bird:patt:1999},
the symbol $\snlvarsucc$ can occur anywhere 
between a variable occurrence and its binding abstraction.
The idea to view $\snlvarsucc$ as a scope delimiter was employed by 
Hendriks and van~Oostrom, who defined an end-of-scope symbol $\adbmal$ \cite{hend:oost:2003}.
This approach is also used in the translation of pure \lambdaterms\ (without
\txtletrec) into Lambdascope-graphs (interaction nets) on which van Oostrom
defines an optimal evaluator for the \lambda-calculus
\cite{oost:looi:zwit:2004}.


We have also used these first-order representations of \lambdaterms\ 
with scope delimiters for studying the limits of an optimising program transformation 
that, for a given \lambdaletrecterm, contracts directly visible redexes, 
and, whenever possible, also contracts such redexes that are concealed by 
recursion \cite{roch:grab:2011:TERMGRAPH}. 
The result of the optimisation should again be a \lambdaletrecterm. 
Since this program transformation can best be defined, 
for a given \lambdaletrecterm~$\allter$, on the infinite \lambdaterm~$\ater$ 
that is the unfolding semantics (see Section~\ref{sec:unfsem}) of $\allter$, 
it is crucial to know when the result of contracting a development of redexes in $\ater$ 
that corresponds to a visible or a concealed redex in $\allter$ can again be written as a \lambdaletrecterms.
 
This suggested the question: 
how can those infinite \lambdaterms\ be characterised that are expressible by 
\lambdaletrecterms\ in the sense that they arise as the unfolding semantics of  
a \lambdaletrecterm?
We answered this question for \lambdaletrecterms\ in the report \cite{grab:roch:2012}, 
and, obtaining the same answer,
for \lambdamuterms\ in the article \cite{grab:roch:2013:c:RTA} with accompanying report \cite{grab:roch:2013:d:muexpressibility:report}.
We defined a rewrite system that decomposes \lambdaterms\
by steps that `observe' \lambdaabstractions, applications, variable occurrences, and end of scopes. 
We showed that infinite \lambdaterms\ that are the unfolding semantics of \lambdaletrecterms\ or of \lambdamuterms\
are precisely those that have only finitely many `generated subterms', that is,
reducts in the decomposition rewrite system.





\section{Preliminaries}
  \label{sec:prelims}

By $\nats$ we denote the natural numbers including zero.
For words $w$ over an alphabet $A$, the length of $w$ is denoted by $\length{w}$.

Let $\asig$ be a \TRS\nb-signature \cite{terese:2003} with arity function $\sarity \funin \asig \to \nats$.
A \emph{\termgraph\ over $\asig$} (or a \emph{$\asig$\nb-\termdashgraph})
is a tuple $\tuple{\verts,\svlab,\svargs,\sroot}$ 
where:
$\verts$ is a set of \emph{vertices},
$\svlab \funin \verts \to \asig$ the \emph{(vertex) label function},
$\svargs \funin \verts \to \verts^*$ the \emph{argument function} 
  that maps every vertex $\avert$ to the word $\vargs{\avert}$ consisting of the $\arity{\vlab{\avert}}$ successor vertices of $\avert$
  (hence $\length{\vargs{\avert}} = \arity{\vlab{\avert}}$),
and $\sroot$, the \emph{root}, is a vertex in $\verts$.
\Termgraphs\ may have infinitely many vertices.

Let $\atg$ be a \termgraph\ over signature $\asig$. 
As useful notation for picking out an arbitrary vertex, or the $i$-th vertex, 
from among the ordered successors of a vertex $\avert$ in $\atg$,
we define for each $i\in\nats$ the indexed edge relation \inMath{{\stgsucci{i}} \subseteq {\verts\times\verts}},
and additionally the (not indexed) edge relation ${\stgsucc} \subseteq \verts\times\verts$,
by stipulating for all $\bvert,\bvertacc\in\verts$: 
\begin{center}
$
\begin{aligned}
  \bvert \tgsucci{i} \bvertacc
    \,\funin&\Leftrightarrow\,
  \existsst{\bverti{0},\ldots,\bverti{n}\in\verts\!}
           {\vargs{\bvert} = \bverti{0}\cdots\bverti{n}
              \logand
            \bvertacc = \bverti{i}}
  \\
  \bvert \tgsucc \bvertacc 
    \,\funin& \Leftrightarrow\,
  \existsst{i\in\nats}{\, \bvert \stgsucci{i} \bvertacc} 
\end{aligned}
$
\end{center}
%
A \emph{path} in $\atg$ 
is described by 
$\bverti{0} \tgsucci{k_1} \bverti{1} \tgsucci{k_2} {} \cdots {} \tgsucci{k_{n}} \bverti{n}$,
where $\bverti{0},\bverti{1},\ldots,\bverti{n}\in\verts$ and $n,k_1,k_2,\ldots,k_{n}\in\nats$.
An \emph{access path} of a vertex $\bvert$ of $\atg$ is
a path that starts at the root of $\atg$, ends in $\bvert$, and does not visit any vertex twice. 
Access paths need not be unique.
A \termgraph\ is \emph{root-connected}
if every vertex has an access path. 

\emph{Note:}\label{note:root-connected}
  By a `\termgraph' we will, from now on, always mean a root-connected \termgraph.

Let $\atgi{1} = \tuple{\vertsi{1},\svlabi{1},\svargsi{1},\srooti{1}}$,
    $\atgi{2} = \tuple{\vertsi{2},\svlabi{2},\svargsi{2},\srooti{2}}$
be \termgraphs\ over signature $\asig$, in the sequel.

A \emph{bisimulation} between $\atgi{1}$ and $\atgi{2}$ is
a relation $\abisim \subseteq \vertsi{1}\times\vertsi{2}$ such that
the following conditions hold, for all $\pair{\bvert}{\bvertacc}\in\abisim$:
\begin{equation}\label{eq:def:bisim}
\left.
\begin{array}{cr}
  \pair{\srooti{1}}{\srooti{2}}  \in \abisim   
  &   (\text{roots})
  \\[0.3ex]
  \vlabi{1}{\bvert} = \vlabi{2}{\bvertacc} 
  &    (\text{labels})  
  \\[0.3ex]
  \pair{\vargsi{1}{\bvert}}{\vargsi{2}{\bvertacc}} \in \abisim^*    
  & \hspace*{5ex} (\text{arguments})
\end{array}
\quad\right\}
\end{equation}
where the extension $\abisim^* \subseteq \wordsover{\vertsi{1}}\times\wordsover{\vertsi{2}}$ 
of $\abisim$ to a relation between words over $\vertsi{1}$ and words over $\vertsi{2}$ is defined as:
\begin{equation*}
  \abisim^* \!\defdby\! 
         \parbox[t]{205pt}{\{$\pair{\bverti{1}\cdots\bverti{k}}{\bvertacci{1}\cdots\bvertacci{k}}
                                      \; \!\mid\!$
                                     \\[0.5ex]\hspace*{2.5em} 
                                     $
                                     \bverti{1},\ldots,\bverti{k}\in\vertsi{1},
                                     \bvertacci{1},\ldots,\bvertacci{k}\in\vertsi{2}$,
                                     \\[0.5ex]
                                     \hspace*{\fill}%
                                     $
                                     \text{for $k\in\nats$ such that~}
                                     \pair{\bverti{i}}{\bvertacci{i}}\in\abisim
                                     \text{ for all $1\le i\le k$}\}$.}
\end{equation*}                                     
We write
$\atgi{1} \bisim \atgi{2}$
if there is a bisimulation between $\atgi{1}$ and $\atgi{2}$\rule{0pt}{2.2ex},
and we say, in this case, that $\atgi{1}$ and $\atgi{2}$ are \emph{bisimilar}.
Bisimilarity $\sbisim$ is an equivalence relation on \termgraphs.

A \emph{functional bisimulation} from $\atgi{1}$ to $\atgi{2}$ is
a bisimulation that is the graph of a function from $\vertsi{1}$ to $\vertsi{2}$. 
An alternative characterisation of this concept is that of 
\emph{homomorphism} from $\atgi{1}$ to $\atgi{2}$:
a morphism from the structure $\atgi{1}$ to the structure $\atgi{2}$, 
that is, a function 
$ \sahom \funin \vertsi{1} \to \vertsi{2}$ 
such that, for all $\avert\in\vertsi{1}$ it holds:
\begin{equation}\label{eq:def:homom}
\left.
\begin{array}{cr}
  \ahom{\srooti{1}} = \srooti{2}
  &   (\text{roots})
  \\[0.3ex]
  \vlabi{1}{\avert} = \vlabi{2}{\ahom{\avert}} 
  &    (\text{labels})  
  \\[0.3ex]
  \funap{\sahomext}{\vargsi{1}{\avert}}
    = 
  \vargsi{2}{\ahom{\avert}}
  & \hspace*{5ex} (\text{arguments})
\end{array}
\quad\right\}
\end{equation}
%
where $\sahomext$ is the homomorphic extension 
$\sahomext \funin \vertsi{1}^* \to \vertsi{2}^*$, 
$ \averti{1}\scdots\averti{n}   \mapsto \ahom{\averti{1}}\scdots\ahom{\averti{n}} $
of $\sahom$ to words over $\vertsi{1}$.
We write $\atgi{1} \funbisim \atgi{2}$
if there is a functional bisimulation (a homomorphism) from $\atgi{1}$ to $\atgi{2}$.
An \emph{isomorphism} between $\atgi{1}$ and $\atgi{2}$ is a bijective homomorphism 
$\saiso \funin \vertsi{1} \to \vertsi{2}$ from $\atgi{1}$ to $\atgi{2}$.
If there is an isomorphism between $\atgi{1}$ and $\atgi{2}$,
we write $\atgi{1} \iso \atgi{2}$, and say that $\atgi{1}$ and $\atgi{2}$ are \emph{isomorph}. 

Let $\atg = \tuple{\verts,\svlab,\svargs,\sroot}$ be a term graph.
A \emph{bisimulation collapse} of $\atg$ 
is a maximal element in the class $\descsetexp{\atgacc}{\atg \funbisim \atgacc}$ up to $\siso$,
that is,
a term graph $\atgacci{0}$ with \mbox{$\atg \funbisim \atgacci{0}$}
such that if \mbox{$\atgacci{0} \funbisim \atgdacci{0}$} for some term graph $\atgdacci{0}$,
then \mbox{$\atgdacci{0} \iso \atgacci{0}$}. 
The \emph{canonical bisimulation collapse}~$\coll{\atg}$ of $\atg$
is defined as the root-connected part of 
the `factor term graph' $\factor{\atg}{\abisim}$ of $\atg$
with respect to the largest bisimulation~$\abisim$ 
between $\atg$ and $\atg$ 
(the largest `self-bisimulation' on $\atg$), which is an equivalence relation on $\verts$.
The \emph{factor term graph}~$\factor{\atg}{\saequivrel}$ of $\atg$ with respect to an equivalence relation $\saequivrel$ on $\verts$ is defined as
$\factor{\atg}{\saequivrel} \defdby \tuple{ \factor{\verts}{\saequivrel}, \factor{\svlab}{\saequivrel}, \factor{\svargs}{\saequivrel}, \eqclofwrt{\sroot}{\saequivrel} }$
where $\factor{\verts}{\saequivrel}$ is the set of $\saequivrel$\nb-equi\-va\-lence classes of vertices in $\verts$,
$\eqclofwrt{\sroot}{\saequivrel}$ is the $\saequivrel$\nb-equi\-va\-lence class of $\sroot$,
and $\factor{\svlab}{\saequivrel}$ and $\factor{\svargs}{\saequivrel}$ are the mappings 
on $\factor{\verts}{\saequivrel}$
that are induced by $\svlab$ and $\svargs$, respectively.
Every two bisimulation collapses of $\atg$ are isomorphic.
This justifies the common abbreviation of saying that 
`the bisimulation collapse' of $\atg$ is unique up to isomorphism.


\section{Unfolding Semantics of \protect\lambdaletreccalterms}
  \label{sec:unfsem}

Informally, we regard \emph{\lambdaletreccalterms} as being defined defined by the following grammar:
\begin{equation*}
\begin{array}{lllll}
  \allter    \sep{::=} \labs{\avar}{\allter}           & (\textit{abstraction}) \\
  \sep{ | } \lapp{\allter}{\allter}         & (\textit{application}) \\ 
  \sep{ | } \avar                           & (\textit{variable})    \\
  \sep{ | } \letin{\abindgroup}{\allter}   & (\textit{letrec})      \\[0.75ex]
  \mathit{\abindgroup} \sep{::=} \tfundef{\arecvari{1}=\allter}{\dots}{\arecvari{n}=\allter}   & (\textit{equations})   \\
  \sep{   } (\arecvari{1},\dots,\arecvari{n} \in \llrecvars~\text{all distinct})
\end{array}
\end{equation*}
Formally, we consider \lambdaletreccalterms\ to be defined correspondingly as
higher-order terms in the formalism of Combinatory Reduction Systems (CRS) \cite{terese:2003}.
CRSs are a higher-order term rewriting framework tailor-made for formalising and
manipulating expressions in higher-order languages (i.e.\ languages with
binding constructs like \lambda-abstractions and \txtlet-bindings). They provide a
sound basis for defining our language and for reasoning with \letrecexpression{s}.  
By formalising a system of unfolding rules 
as a CRS we conveniently externalise issues
like name capturing and $\alpha$\nb-re\-na\-ming, 
which otherwise would have to be handled by 
a calculus of explicit substitution. Also, we can lean on the
rewriting theory of CRSs for the proofs. 

As \CRS\nb-signature we use
$\siglambdaletreccalCRS = \siglambdacalCRS \cup \descsetexp{\sletCRS{n},\srecinCRS{n}}{n\in\nats}$
with $\siglambdacalCRS = \setexp{ \slabsCRS, \slappCRS } $,
where the unary symbol $\slabsCRS$ and the binary symbol $\slappCRS$ represent
\lambdaabstraction\ and application, respectively;
the symbols $\sletCRS{n}$ of arity one, and $\srecinCRS{n}$ of arity $n+1$ 
together formalise \letexpressions\ with $n$ bindings.
By $\termsize{\allter}$ we denote the size (number of symbols) of a \lambdaletrecterm~$\allter$.
By $\Ter{\txtlambdaletreccal}$ we denote the set of CRS\nb-terms over $\siglambdaletreccalCRS$.
For readability, we will rely on the informal first-order notation.

\emph{Infinite \lambdaterms} are formalised as \iCRS-terms 
(terms in an infinitary CRS \cite{kete:simo:2011}) over $\siglambdacalCRS$, forming the set $\Ter{\txtinflambdacal}$.
Informally, infinite \lambdaterms\ are generated co-inductively 
by the alternatives (\textit{abstraction}), (\textit{application}),
and (\textit{variable}) of the grammar above. 

In order to formally define the infinite unfolding of \lambdaletrecterms\ we utilise
a CRS whose rewrite rules formalise unfolding steps \cite{grab:roch:2012}. 
Every \lambdaletrecterm~$\allter$ that represents an infinite \lambdaterm~$\aiter$ 
can be rewritten by a typically infinite rewrite sequence 
that converges to $\aiter$ in the limit. 
However, not every \lambdaletreccalterm\ represents an \inflambdaterm. For
instance the \lambdaletreccalterm\ 
$\cllter = \labs{\avar}{\letin{\arecvar = \arecvar}{\lapp{\arecvar}{\avar}}}$ 
with a meaningless
\letbinding\ for $\arecvar$ does not unfold to a \inflambdaterm.
Therefore we introduce a constant symbol $\bh$,
called `black hole', for expressing meaningless bindings,
in order to define the unfolding operation as a total function.
The unfolding semantics of $\cllter$ will then be $\labs{\avar}{\lapp{\bh}{\avar}}$.
So we extend the signature $\siglambdacalCRS$ to $\siglambdabhcalCRS$ including $\blackhole$,
and denote the set of infinite \lambdaterms\ over $\siglambdacalCRS$ by $\Ter{\txtinflambdabhcal}$.
Similarly, the rules below are defined for terms in $\Ter{\txtlambdaletrecbhcal}$
based on signature $\siglambdaletrecbhcalCRS$ that extends $\siglambdaletreccalCRS$ by the blackhole constant.



\begin{definition}[unfolding CRS for \lambdaletreccal-terms]
The rules
\begin{align*}
  & (\unfrulelapp)            && \hspace*{-1.5ex}
    \letin{\abindgroup}{\lapp{\allteri{0}}{\allteri{1}}}
      \;\red\; 
    \lapp{(\letin{\abindgroup}{\allteri{0}})}{(\letin{\abindgroup}{\allteri{1}})}
    \displaybreak[0]\\
  & (\unfrulelabs)            && \hspace*{-1.5ex}
    \letin{\abindgroup}{\labs{\avar}{\allteri{0}}}
      \;\red\; 
    \labs{\avar}{\letin{\abindgroup}{\allteri{0}}}
    \displaybreak[0]\\
  & (\unfruleletin)       && \hspace*{-1.5ex} 
    \letin{\abindgroupi{0}}{\letin{\abindgroupi{1}}{\allter}}
      \;\red\; 
     \letin{\abindgroupi{0},\abindgroupi{1}}{\allter} 
    \displaybreak[0]\\
  & (\unfruleletrec)   && \hspace*{-1.5ex} 
    \letin{\tfundef{\abindgroupi{1}}{\arecvar = \allter}{\abindgroupi{2}}}{\arecvar}
      \;\red\; 
    \letin{\tfundef{\abindgroupi{1}}{\arecvar = \allter}{\abindgroupi{2}}}{\allter} 
    \displaybreak[0]\\
  & (\unfrulegarbage)          && \hspace*{-1.5ex}  
      \letin{\tfundef{\arecvari{1} = \allteri{1}}{\ldots}{\arecvari{n} = \allteri{n}}}{\bllter}
        \;\red\; 
      \bllter
      \\[-0.5ex]
      &&& \hspace*{12ex}\text{(if $\arecvari{1},\ldots,\arecvari{n}$ do not occur in $\bllter$)}
    \displaybreak[0]\\
  & (\unfruletighten)     && \hspace*{-1.5ex}
      \letin{\tfundef{\abindgroupi{1}}{\arecvar = \brecvar}{\abindgroupi{2}}}{\allter}
      \\[-0.5ex]
      &&& \hspace*{1.5ex}
        \red\;
      \letin{\bfundef{\subst{\abindgroupi{1}}{\arecvar}{\brecvar}}{\subst{\abindgroupi{2}}{\arecvar}{\brecvar}}}{\subst{\allter}{\arecvar}{\brecvar}} 
      \\[-0.5ex]
      &&& \hspace*{2ex}\text{(where $\brecvar$ with $\brecvar \neq \arecvar$ a recursion variable in $\abindgroupi{1}$ or $\abindgroupi{2}$)}
    \displaybreak[0]\\
  & (\unfruleblackhole)        && \hspace*{-1.5ex}
    \letin{\tfundef{\abindgroupi{1}}{\arecvar = \arecvar}{\abindgroupi{2}}}{\allter}
      \;\red\;
    \letin{\tfundef{\abindgroupi{1}}{\arecvar = \blackhole}{\abindgroupi{2}}}{\allter}
\end{align*}
define, in informal notation, the \emph{unfolding CRS} for \lambdaletreccalterms\ with rewrite relation~$\sunfoldred$. 
See Fig.~\ref{fig:unfrules:CRSnotation} for the formulation of these rules in explicit \CRSnotation. 
\begin{figure*}
\begin{align*}
 &
 (\unfrulelapp) && \hspace*{-1.5ex}
       \letrecinCRS{n}{\vec{\arecvar}}{\ametavari{1}{\vec{\arecvar}}, \ldots, \ametavari{n}{\vec{\arecvar}}, \lappCRS{\cmetavari{0}{\vec{\arecvar}}}{\cmetavari{1}{\vec{\arecvar}}}}
       \\
       &&&{} \hspace*{3ex} \red\;
       \slappCRS{((\letrecinCRS{n}{\vec{\arecvar}}{
                                                                              \ldots,\ametavari{n}{\vec{\arecvar}},\cmetavari{0}{\vec{\arecvar}}}),
                 (\letrecinCRS{n}{\vec{\arecvar}}{
                                                                                \ldots, \ametavari{n}{\vec{\arecvar}},\cmetavari{1}{\vec{\arecvar}}}))}                        
   \displaybreak[0]\\[0.75ex]
  &
  \scalebox{1}{(\unfrulelabs)} && \hspace*{-1.5ex}
        \scalebox{1}{$
        \letrecinCRS{n}{\vec{\arecvar}}{\ametavari{1}{\vec{\arecvar}}, \ldots, \ametavari{n}{\vec{\arecvar}}, \labsCRS{\avar}{\cmetavar{\vec{\arecvar},\avar}}}
                        $}
        \\
        &&& \hspace*{3ex} 
        {} \scalebox{1}{$\sred$}\;
        \scalebox{1}{$
        \labsCRS{\avar}{\letrecinCRS{n}{\vec{\arecvar}}{\ametavari{1}{\vec{\arecvar}}, \ldots, \ametavari{n}{\vec{\arecvar}}, \cmetavar{\vec{\arecvar},\avar}}}
                        $}
    \displaybreak[0]\\[0.75ex]
  &  
  \scalebox{1}{(\unfruleletin)} && \hspace*{-1.5ex}
        \scalebox{1}{$
        \letrecinCRS{n}{\vec{\arecvar}}{\funap{\vec{\sametavar}}{\vec{\arecvar}}, 
                                        \letrecinCRS{m}{\vec{\brecvar}}{\funap{\vec{\sbmetavar}}{\vec{\arecvar},\vec{\brecvar}}}, 
                                                                        \cmetavar{\vec{\arecvar},\vec{\brecvar}}}      
                       $}                                            
        \\
        &&& \hspace*{3ex} 
        {} \scalebox{1}{$\sred$}\;
        \scalebox{1}{$
        \letrecinCRS{n+m}{\vec{\arecvar}\vec{\brecvar}}{\funap{\vec{\sametavar}}{\vec{\arecvar}}, 
                                                        \funap{\vec{\sbmetavar}}{\vec{\arecvar},\vec{\brecvar}}, 
                                                        \cmetavar{\vec{\arecvar},\vec{\brecvar}}}
                        $}
   \displaybreak[0]\\[0.75ex]
 &  
 (\unfruleletrec) && \hspace*{-1.5ex}
     \letrecinCRS{n}{\vec{\arecvar}}{\ametavari{1}{\vec{\arecvar}}, \ldots, \ametavari{n}{\vec{\arecvar}}, \arecvari{i}}
     \;\red\;
     \letrecinCRS{n}{\vec{\arecvar}}{\ametavari{1}{\vec{\arecvar}}, \ldots, \ametavari{n}{\vec{\arecvar}}, \ametavari{i}{\vec{\arecvar}}}   
   \displaybreak[0]\\[0.75ex]
 &  
 (\unfrulegarbage)          && \hspace*{-1.5ex}  
   \letrecinCRS{n}{\vec{\arecvar}}{\ametavari{1}{\vec{\arecvar}}, \ldots, \ametavari{n}{\vec{\arecvar}},\scmetavar}
     \;\red\; 
   \scmetavar 
     \displaybreak[0]\\[0.75ex]
 &    
 (\unfruletighten) && \hspace*{-1.5ex}
     \letrecinCRS{n}{\vec{\arecvar}}{\ametavari{1}{\vec{\arecvar}}, \ldots, \ametavari{i-1}{\vec{\arecvar}},
                                  \arecvari{j},                             
                                  \ametavari{i+1}{\vec{\arecvar}}, \ldots, \ametavari{n}{\vec{\arecvar}},
                                  \cmetavar{\vec{\arecvar}}} 
     \\
     &&& \hspace*{3ex} \red\; 
     \letrecinCRS{n-1}{\vec{\brecvar}}{\ametavari{1}{\vec{\brecvar}'}, \ldots, \ametavari{i-1}{\vec{\brecvar}'},
                                  \ametavari{i+1}{\vec{\brecvar}'}, \ldots, \ametavari{n}{\vec{\brecvar}'},
                                  \cmetavar{\vec{\brecvar}'}} 
     \\
     &&& \phantom{\;\hspace*{3ex}\red\;\;} \hspace*{3ex}
       \text{where: }
       \parbox[t]{300pt}{$i,j\in\setexp{1,\ldots,n}$ with $i\neq j$, and
                         $\vec{\brecvar}' = \tuple{\brecvari{1},\ldots,\brecvari{i-1},\brecvari{j},\brecvari{i+1},\ldots,\brecvari{n}}$}   
   \displaybreak[0]\\[0.75ex]
 &  
 (\unfruleblackhole) && \hspace*{-1.5ex}
      \letrecinCRS{n}{\vec{\arecvar}}{\ametavari{1}{\vec{\arecvar}}, \ldots, \ametavari{i-1}{\vec{\arecvar}},
                                      \arecvari{i},                             
                                      \ametavari{i+1}{\vec{\arecvar}}, \ldots, \ametavari{n}{\vec{\arecvar}},
                                      \arecvari{i}}
      \\
      &&& \hspace*{3ex}                                
      \red\;
      \letrecinCRS{n}{\vec{\arecvar}}{\ametavari{1}{\vec{\arecvar}}, \ldots, \ametavari{i-1}{\vec{\arecvar}},
                                      \blackhole,                             
                                      \ametavari{i+1}{\vec{\arecvar}}, \ldots, \ametavari{n}{\vec{\arecvar}},
                                      \arecvari{i}}
\end{align*}
  \caption{\label{fig:unfrules:CRSnotation}
           The rules for unfolding \lambdaletrecterms\ in explicit \CRSnotation.}
\end{figure*}
\end{definition}

\begin{example}[Unfolding derivation of $\allter$ from Ex.~\ref{ex:fix}]~\\
$\labs{f}{\letin{r=\lapp{f}{r}}{r}} ~\sred_{\sunfold}^{(\unfruleletrec)}~
\labs{f}{\letin{r=\lapp{f}{r}}{\lapp{f}{r}}} ~\sred_{\sunfold}^{(\unfrulelapp)}$\\
$\labs{f}{\lapp{(\letin{r=\lapp{f}{r}}{f})}{(\letin{r=\lapp{f}{r}}{r})}} ~\sred_{\sunfold}^{(\unfrulegarbage)}$\\
$\labs{f}{\lapp{f}{(\letin{r=\lapp{f}{r}}{r}})} ~\sred_{\sunfold}^{(\unfruleletrec)}~ \dots$
\end{example}

We say that a \lambdaletreccalterm~$\allter$ \emph{unfolds to} an \inflambdabhterm~$\aiter$,
or that $\allter$ \emph{expresses} $\aiter$, 
if there is a (typically) infinite $\sunfoldred$-rewrite sequence 
                                                                  from $\allter$ that
converges to $\aiter$, symbolically $\allter \unfoldinfred \aiter$. 
Note that any such rewrite sequence is strongly convergent 
(the depth of the contracted redexes tends to infinity),
because the resulting term does not contain any \letexpressions. 


\begin{lemma}
  Every \lambdaletrecterm\ unfolds to precisely one \inflambdabhterm.
\end{lemma}

\begin{proof}[Proof (Outline)]
  Infinite normal forms of $\sunfoldred$ are \inflambdabhterms\
  since: every occurrence of a \letexpression\ in a
  \txtlambdaletrecbhcal\ gives rise to a redex; and 
  infinite terms over $\siglambdaletrecbhcalCRS$ 
  without \letexpressions\ are 
                               \inflambdabhterms. 
  Also, outermost-fair rewrite sequences 
  in which the rules ($\text{tighten}$) and ($\blackhole$) are applied eagerly
  are (strongly) convergent. 

  Unique infinite normalisation of $\sunfoldred$ follows from finitary confluence of $\sunfoldred$. 
  In previous work \cite{grab:roch:2012} we proved confluence for the slightly simpler CRS 
  that does not contain the final two rules, which together introduce black holes
  in terms with meaningless bindings. 
  That confluence proof can be adapted
  by extending the argumentation to deal with the additional critical pairs. 
\end{proof}

\begin{definition}
                  \label{def:unfsem}
  The \emph{unfolding semantics} for \lambdaletrecterms\ is defined by the function
  $\sunfsembh \funin \Ter{\txtlambdaletreccal} \to \Ter{\txtinflambdabhcal}$,
  where 
  $ \allter \mapsto \unfsembh{\allter} \defdby\,$
  the infinite unfolding of $\allter$. 
\end{definition}

\begin{remark}[Regular and strongly regular \inflambdaterms]
  \label{rem:regular:strongly:regular}
\inflambdaterms\ that arise as infinite unfoldings of \lambdaletreccalterms\
form a proper subclass of those \inflambdaterms\ that have a regular term structure
\cite{grab:roch:2012}. \inflambdaterms\ that belong to this subclass are called
`strongly regular', and can be characterised by means of a decomposition rewrite system,
and as those that contain only finite `\bindcaptchains'
\cite{grab:roch:2012,grab:roch:2013:c:RTA}.  
\end{remark}

\section{\Lambdahotgs}
  \label{sec:lhotgs}

In this section we motivate the use of higher-order \termgraphs\ 
as a semantics for 
     \lambdaletreccalterms;
we introduce the class $\classlhotgs$ of `\lambdahotgs' and
define the semantics 
                     $\sgraphsemC{\classlhotgs}$
for interpreting \lambdaletreccalterms\ as \lambdahotgs. 
Finally, we sketch a proof of the correctness of $\sgraphsemC{\classlhotgs}$
with respect to unfolding equivalence (the property \ref{methods:properties:correctness}).

We start out from a natural interpretation of \lambdaletrecterms\ as
first-order \termgraphs:
occurrences of abstraction variables are resolved as edges 
pointing to the corresponding abstraction;
occurrences of recursion variables as edges to the subgraph belonging to the
respective binding. 
We therefore consider \termgraphs\ over the signature $\siglambdabh =
\setexp{\sslapp,\sslabs,\snlvar,\bh}$ with arities $\arity{\sslapp} = 2$,
$\arity{\sslabs} = 1$, $\arity{\snlvar} = 1$, and $\arity\bh = 0$.
These function symbols represent 
applications, \lambdaabstractions, abstraction variables, and black holes.

We will later define a subclass of these term graphs that excludes meaningless graphs.
In line with the choice to
regard all terms as higher-order terms (thus modulo \alphaconversion), we
consider a nameless graph representation, 
so that \alphaequivalence\ of two terms can be recognised 
as their graph interpretations being isomorphic.                  

For a \termgraph\ $\atg$ over $\siglambdabh$ with set $\verts$ of vertices
we will henceforth denote by
$\vertsof{\sslapp}$, $\vertsof{\sslabs}$, $\vertsof{\snlvar}$, and $\vertsof{\blackhole}$
the sets of \emph{application vertices}, \emph{abstraction vertices}, \emph{variable vertices}, and \emph{blackhole vertices},
that is, those with label $\sslapp$, $\sslabs$, $\snlvar$, $\blackhole$, respectively.

\begin{example}[Natural first-order interpretation]%
    \label{ex:naive:representation}
  The \lambdaletreccalterms\ $\allter$ and $\bllter$ in Ex.~\ref{ex:fix}
  can be represented as the \termgraphs\ in Fig.~\ref{fig:ex:compact}.
\end{example}

These two graphs are bisimilar, which suggests that $\allter$ and $\bllter$ are
unfolding equivalent. Moreover, there is a functional bisimulation from the
larger \termgraph\ to the smaller one, indicating that $\allter$ expresses more sharing than $\bllter$,
or in other words: $\allter$ is more compact. Also, there is no smaller \termgraph\
that is bisimilar to $\allter$ and $\bllter$.  We conclude that $\allter$ is
a maximally shared form of $\bllter$.

However, this translation is incorrect in the sense that bisimilarity does not
in general guarantee unfolding equivalence, the desired property~\ref{methods:properties:correctness}.
This is witnessed by the following counterexample.

\begin{example}[Incorrectness of the natural first-order interpretation\vspace{-1.5ex}]\label{ex:counternat}
\begin{equation*}
  \begin{aligned}
    \allteri{1} & \;\;=\;\;
      \letin{\arecvar = \labs{\avar}{\lapp{(\labs{\bvar}{\lapp{\arecvar}{\bvar}})}{\avar}}}
            {\arecvar}
    \\
    \allter & \;\;=\;\;
      \letin{\arecvar = \labs{\avar}{\lapp{\arecvar}{\avar}}}
            {\arecvar}            
    \\
    \allteri{2} & \;\;=\;\;
      \letin{\arecvar = \labs{\avar}{\lapp{(\labs{\bvar}{\lapp{\arecvar}{\avar}})}{\avar}}}
            {\arecvar}
  \end{aligned}
\end{equation*}

  \noindent
  While $\unfsem{\allteri{1}} \hspace{-0.6ex} = \hspace{-0.6ex} \unfsem{\allter}$
  and
  $\unfsem{\allter} \hspace{-0.6ex} \neq \hspace{-0.6ex} \unfsem{\allteri{2}}$,
  all of their \termgraphs\ $G_1$, $G$, $G$ are bisimilar (please ignore the shading for now):
  %
\mbox{}\hfill
\vcentered{\figsmall{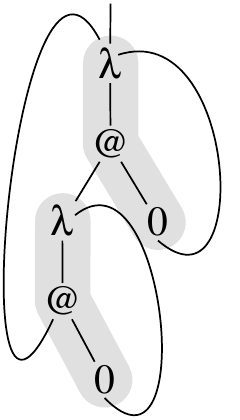}}
\hfill\vcentered{\sfunbisim}\hfill
\vcentered{\figsmall{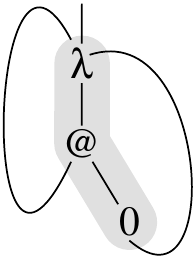}}
\hfill\vcentered{\sinvfunbisim}\hfill
\vcentered{\figsmall{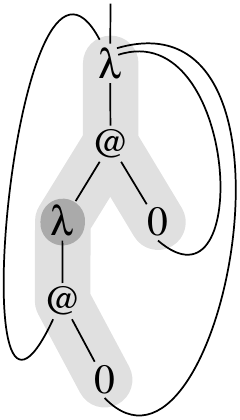}}
\hfill\mbox{}
\\[-2ex]
\hspace*{10.25ex}
$G_1$
\hspace*{12.75ex}
$G$
\hspace*{24ex}
$G_2$
\end{example}

Consequently this interpretation 
lacks the necessary structure for correctly
modelling compactification via bisimulation collapse.

We therefore impose additional structure on the \termgraphs. This
is indicated by the shading in the picture above, and in the graphs throughout
this paper. A shaded area depicts the \emph{scope} of an abstraction: it
comprises all positions between the abstraction and its bound variable occurrences as
well as the scope of any abstraction on these positions. 
By this stipulation, scopes are properly nested.

%

Now note that the functional bisimulation on the right in the picture in Ex.~\ref{ex:counternat} 
does \emph{not} respect the scopes:
The scope of the topmost abstraction vertex 
in the \termgraph\ $G_2$ interpreting $\allteri{2}$
contains another \lambdaabstraction; hence the image of this scope 
under the functional bisimulation cannot fit into, and is not contained in,
the single scope in 
                    the \termgraph\ $G$ of $\allter$.
Also, the trivial scope
of the vacuous abstraction in 
                              $G_2$ is not mapped to
a scope in $G$. 
Thus the natural first-order interpretation is incorrect, in the sense that
functional bisimulation does not preserve scopes on the first-order \termgraphs\ 
that are interpretations of \lambdaletreccalterms.



To prevent that interpretations of not unfolding-equivalent terms like
$\allteri{1}$ and $\allteri{2}$ in Ex.~\ref{ex:counternat} become bisimilar, we
enrich first-order \termgraphs\ by a formal concept of scope.
More precisely, \emph{abstraction prefixes} are added as
vertex labels. They also serve the purpose of defining the
subclass of meaningful term graphs over $\siglambdabh$ that
sensibly represent cyclic \lambdaterms.
In the enriched \termgraphs, each vertex $\avert$ is annotated with a label $\abspre{\avert}$, the
\emph{abstraction prefix} of $\avert$, which is a list of vertex names
that identifies the abstraction vertices in whose scope $\avert$ resides. 
Alternatively scopes can be represented by a scope function (as in
\cite{blom:2001}) that assigns to every abstraction vertex the set of vertices in its scope. 
In the article~\cite{grab:roch:2013:a:TERMGRAPH} 
we show that higher-order \termgraphs\  with scope functions correspond bijectively to those with abstraction~prefix~functions.

Abstraction prefixes can be determined by traversing over the graph and
recording every binding encountered. When passing an abstraction vertex $\avert$
while descending into the subgraph representing the body of the abstraction,
one enters or opens the scope of $\avert$. This is recorded by appending
$\avert$ to the abstraction prefix of $\avert$'s successor. 
$\avert$ is removed from the prefix at positions under which the
abstraction variable is no longer used,
but not before any other variable that was added to the prefix in the meantime has itself been removed.  
In other words, the abstraction prefix behaves like a stack. 
We call \termgraphs\  for representing
\lambdaletreccalterms\ that are equipped with \absprefix{es}
`\lambda-higher-order \termgraphs' 
                                   (\lambdahotgs).

\begin{example}[The \lambdahotgs\ of the terms in Ex.~\ref{ex:counternat}]\mbox{}
\label{ex:counternat_lhotgs}
  \begin{center}
    $
    \begin{aligned}[c] 
      \fig{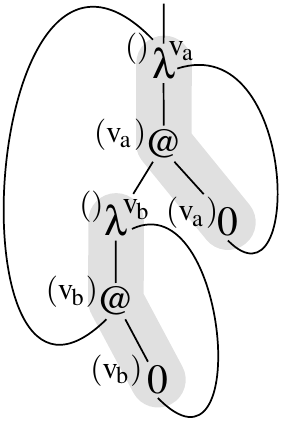}
    \end{aligned} 
       \hfill\sfunbisim\hfill
    \begin{aligned}[c]
      \fig{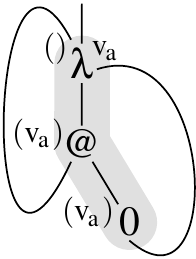}
    \end{aligned}
       \hfill\not\sinvfunbisim\hfill
    \begin{aligned}[c]  
      \fig{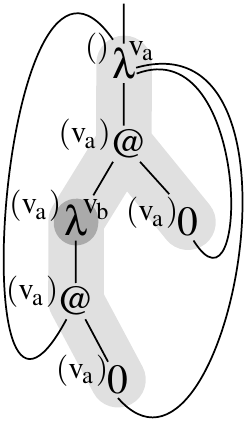}
    \end{aligned}
    $   
  \end{center}
The superscripts of abstraction vertices indicate their names. The abstraction
prefix of a vertex is annotated to its top left. Note that abstraction vertices
themselves are not included in their own prefix.
\end{example}

We define \lambdahotgs{} as term graphs over $\siglambdabh$
together with an abstraction-prefix function that assigns to each vertex an
abstraction prefix.
It has to respect certain correctness conditions
restricting the \lambdahotgs{} to exclude meaningless term graphs.

%


\begin{definition}
  [correct abstraction-prefix function for \termgraphs\  over $\siglambdabh$]%
    \label{def:abspre:function}\normalfont
  Let $ \atg = \tuple{\verts,\svlab,\svargs,\sroot}$ be
  a $\siglambdabh$\nb-term-graph.

  An \emph{abstraction-prefix function} for $\atg$ 
  is a function $\sabspre \funin \verts \to \verts^*$
  from vertices of $\atg$ to words of vertices.
  Such a function is called \emph{correct} if
  for all $\bvert,\bverti{0},\bverti{1}\in\verts$ and $k\in\setexp{0,1}$
  it holds:
  \begin{align*}
    & \abspre{\sroot} = \emptyword 
    & (\text{root})
    \\
    & \abspre{\bh} = \emptyword 
    & \hspace*{-5ex} (\text{black hole})
    \\
    \bvert\in\vertsof{\sslabs} 
      \;\logand\;
    \bvert \tgsucci{0} \bverti{0}
      \;\; & \Rightarrow \;\;
    \abspre{\bverti{0}} \prele \abspre{\bvert} \bvert
    & (\sslabs)
    \displaybreak[0]\\
    \bvert\in\vertsof{\sslapp}
      \;\logand\;
    \bvert \tgsucci{k} \bverti{k}
      \;\; & \Rightarrow \;\;
    \abspre{\bverti{k}} \prele \abspre{\bvert} 
    & (\sslapp)
    \displaybreak[0]\\
    \bvert\in\vertsof{\snlvar}
      \;\logand\;
    \bvert \tgsucci{0} \bverti{0}
      \;\; & \Rightarrow \;\;
      \left\{\hspace*{1pt}
      \begin{aligned}[c]
        & \bverti{0}\in\vertsof{\sslabs}
        \\[-0.5ex]
        & \;\logand\;
        \abspre{\bverti{0}}\bverti{0} = \abspre{\bvert}
      \end{aligned}
      \right.
    & (\snlvar)  
  \end{align*}
  Here and later we denote by $\prele$ the `is-prefix-of' relation.
\end{definition}  

\begin{definition}[\lambdahotg]\label{def:lhotg}\normalfont
  A \emph{\lambdahotg}
  over $\siglambdabh$
  is a five-tuple $\alhotg = \tuple{\verts,\svlab,\svargs,\sroot,\sabspre}$
  where $\atgi{\alhotg} = \tuple{\verts,\svlab,\svargs,\sroot}$ is a \termgraph\  over $\siglambdabh$,
  called the \termgraph\  \emph{underlying} $\alhotg$,
  and $\sabspre$ is a correct \absprefix\ function for $\atgi{\alhotg}$.
  The class of \lambdahotg{s} over $\siglambdabh$ 
  is denoted by $\classlhotgs$.
\end{definition}

\begin{definition}[homomorphism, bisimulation for \lambdahotgs]\label{def:homom:aplambdahotg}\normalfont
  Let $\alhotgi{1} = \tuple{\vertsi{1},\svlabi{1},\svargsi{1},\srooti{1},\sabsprei{1}}$
  and $\alhotgi{2} = \tuple{\vertsi{2},\svlabi{2},\svargsi{2},\srooti{2},\sabsprei{2}}$
  be \lambdahotg{s} over $\siglambdabh$. 

  A \emph{bisimulation} between $\alaphotgi{1}$ and $\alaphotgi{2}$ 
  is a relation $\abisim \subseteq \vertsi{1}\times\vertsi{2}$ 
  such that for all $\pair{\bvert}{\bvertacc}\in\abisim$ 
  the conditions \eqref{eq:def:bisim}, and additionally:
  A \emph{bisimulation} between $\alaphotgi{1}$ and $\alaphotgi{2}$ is a
  relation $\abisim \subseteq \vertsi{1}\times\vertsi{2}$ that is a
  bisimulation between the term graphs $\atgi{\alaphotgi{1}}$ and $\atgi{\alaphotgi{2}}$
  underlying $\alaphotgi{1}$ and $\alaphotgi{2}$, respectively, and for which also the
  following condition:
  \begin{equation}\label{eq:def:bisim:lambdahotg}
    \begin{aligned}
      \pair{\absprei{1}{\bvert}}
          {\absprei{2}{\bvertacc}}
       & \in \abisim^*
       & & & (\text{abstraction-prefix functions})
    \end{aligned}
  \end{equation}
  (for $\abisim^*$ see p.~\pageref{eq:def:bisim} below \eqref{eq:def:bisim})
  is satisfied for all $\bvert\in\vertsi{1}$ and all $\bvertacc\in\vertsi{2}$.
  If there is a such bisimulation, 
  then $\alhotgi{1}$ and $\alhotgi{2}$ are \emph{bisimilar},
  denoted by $\alhotgi{1} \bisim \alhotgi{2}$. 

  A \emph{homomorphism} (a \emph{functional bisimulation})
  from $\alaphotgi{1}$ to $\alaphotgi{2}$ 
  is a morphism from the structure~$\alhotgi{1}$ to the structure $\alhotgi{2}$,
  or more explicitly, it is a homomorphism $\sahom \funin \vertsi{1} \to \vertsi{2}$
  from $\atgi{\alaphotgi{1}}$ to $\atgi{\alaphotgi{2}}$ that additionally satisfies, for all $\bvert\in\vertsi{1}$, the following condition:
  \begin{equation}\label{eq:def:homom:lambdahotg}
    \begin{aligned}
      \funap{\bar{\sahom}}{\absprei{1}{\bvert}}
        & = \absprei{2}{\ahom{\bvert}}
        & & & (\text{abstraction-prefix functions})
    \end{aligned}
  \end{equation}
  for all $\bvert\in\vertsi{1}$,
  where ${\bar{\sahom}}$ is the homomorphic extension of $\sahom$ to words over $\vertsi{1}$.
  We write 
  $\alhotgi{1} \funbisim \alhotgi{2}$
  if there is a homomorphism between $\alhotgi{1}$ and~$\alhotgi{2}$.
\end{definition}


\subsection{Interpretion of \lambdaletrecterms\ as \lambdahotgs} 
  \label{sec:lhotgs:subsec:trans}

\begin{figure*}[htb]
\input{trans-lambdaletreccal-lhotgs.tex}
\caption{\label{fig:def:graphsem:lhotgs}
  Translation rules $\rulestranslambdaletreccaltolhotgs$
  for interpreting \lambdaletreccalterms\ as \lambdahotgs.
  See Section~\ref{sec:lhotgs:subsec:trans} for explanations.
}
\end{figure*}

In order to interpret a \lambdaletreccalterm\ $\allter$ as \lambdahotg,
the translation rules $\rulestranslambdaletreccaltolhotgs$ 
from Fig.~\ref{fig:def:graphsem:lhotgs} are applied to a `translation box'
\adjustbox{fbox={\fboxrule} 1pt 0.5pt}{$\femptylabs{\allter}$}.
It contains $\allter$ furnished with a prefix 
consisting of a dummy variable $*$, and an empty set $[]$ of binding equations.
The translation process proceeds by 
induction on the syntactical structure of the prefixed \lambdaletreccalexpression's body.
Ultimately, 
a \termgraph\ $\atg$ over $\siglambdabh$ is produced, together with a correct~\absprefixfunction~for~$\atg$. 

For reading the rules $\rulestranslambdaletreccaltolhotgs$ in Fig.~\ref{fig:def:graphsem:lhotgs} correctly,
observe the details as described here below.                   
For illustration of their application,
please refer to Appendix~\ref{app:translation} where 
                                                     several \lambdaletrecterms\ 
are translated into \lambdahotgs. 



\begin{itemize}
 \item 
      A translation box \adjustbox{fbox={\fboxrule} 1pt 0.5pt}{$\flabs{\vec{p}}{\allter}$} 
      contains a prefixed, partially decomposed \lambdaletrecterm~$\allter$.  
      The prefix contains a vector $\vec{p}$ of annotated \lambdaabstractions\ that have already
      been translated and whose scope typically extends into $\allter$. 
      Every prefix abstraction is annotated with
      a set of binding equations that are defined at its level. 
      There is special dummy variable denoted by $*$ at the left of the prefix
      that carries top-level function
      bindings, i.e.\ binding equations that are not defined under any enclosing \lambdaabstraction.
      The $\sslabs$\nb-rule strips off an abstraction from the body of the expression, and pushes the abstraction variable into the prefix,  
      which initially contains an empty set of function bindings.
\item Names of abstraction vertices are indicated to the right,
      and \absprefixes\ to the left of the created vertices.
      In order to refer to the vertices in the prefix we use the following notation:
      $\vs{\vec{\apre}} = \averti{1}\,\cdots\,\averti{n}$
      \mbox{} if \mbox{} $\vec{\apre} = *[B_0]\;x_1^{\averti{1}}[B_1]\;\dots\;x_n^{\averti{n}}[B_n]$.
\item Vertices drawn with dashed lines have been created in earlier translation steps,
      and in the current step are referenced by edges in the current step.
\item In the $\snlvarsucc$\nb-rule, which takes care of closing scopes,
      $FV(\allter)$ stands for the set of free variables in $\allter$.
\item The $\sslet$\nb-rule for translating \letexpressions\
      creates a box for the $\txtin$\nb-part as well as
      for each binding equation. 
      The translation of each of the bindings starts with an \emph{indirection vertex}.
      These vertices guarantee the well-definedness of the process when it translates 
      meaningless bindings such as $\arecvar = \arecvar$, or $\brecvar = \crecvar,\, \crecvar = \brecvar$,
      which would otherwise give rise to loops without vertices.
      The $\sslet$\nb-rule pushes the function bindings into the abstraction
      prefix, associating each function binding with one of the variables in
      the abstraction prefix. There is some freedom as to which variable a
      function binding is assigned to. This freedom is limited by scoping
      conditions that ensure that the prefixed term is a valid \CRS-term:
      function bindings may only depend on variables and functions that occur
      further to the left in the prefix.
      The chosen association also directly determines the prefix lengths used
      in the translation boxes for the function bindings.
\item Indirection vertices 
      are eliminated by an erasure process at the end:
      Every indirection vertex that does not point to itself
      is removed, redirecting all incoming edges to the successor vertex.
      Finally every loop on a single indirection vertex  
      is replaced by a \emph{black hole} vertex that represents a meaningless binding.
      Abstraction prefixes for such black holes are defined to be empty.
\end{itemize}
\begin{figure}[htb]
\fig{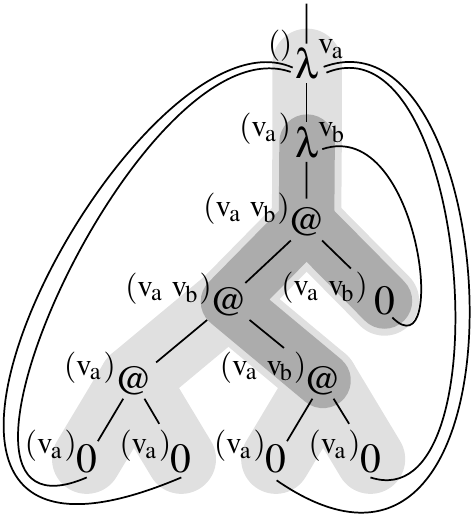}
\hfill
\fig{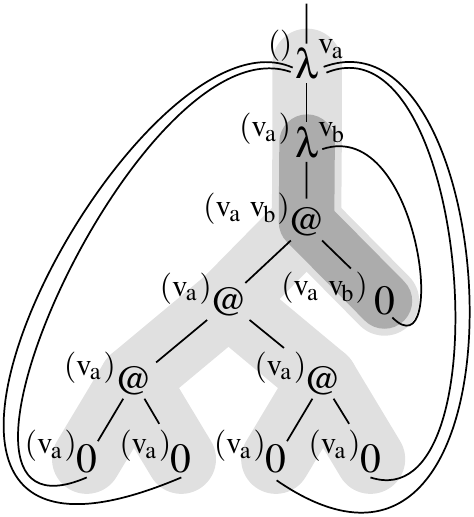}
\caption{\label{fig:rigid_let_rule}
Translation of 
$\labs{\aavar}{\labs{\bbvar}{\letin{\arecvar=\aavar}{\lapp{\lapp{\lapp{\aavar}{\aavar}}{(\lapp{\arecvar}{\aavar})}}{\bbvar}}}}$ 
with equal (left) and with minimal prefix lengths (right) in the \txtlet-rule.}
\end{figure} 
\begin{definition}\label{def:rulestranslambdaletreccaltolhotgsgenerated}
We say that a \termgraph\ $\atg$ over $\siglambdabh$ and an \absprefixfunction~$\sabspre$
is \emph{\rulestranslambdaletreccaltolhotgsgenerated\ from} a
\lambdaletrecterm~$\allter$ if $\atg$ and $\sabspre$ are obtained by applying
the rules~$\rulestranslambdaletreccaltolhotgs$ from
Fig.~\ref{fig:def:graphsem:lhotgs} to
\adjustbox{fbox={\fboxrule} 1pt 0.5pt}{$\femptylabs{\allter}$}.
\end{definition}

\begin{remark}[Inference rule formulation of $\rulestranslambdaletreccaltolhotgs$]
  See also Fig.~\ref{fig:trans-lambdaletreccal-lhotgs-proof-system} for inference
  rules that correspond to the deconstruction of prefixed terms in
  $\rulestranslambdaletreccaltolhotgs$.
\end{remark}

\begin{figure*}[tbh]
  \input{trans-lambdaletreccal-lhotgs-proof-system.tex} 
\caption{\label{fig:trans-lambdaletreccal-lhotgs-proof-system}
  Alternative formulation as inference rules
  of the translation rules in Fig.~\ref{fig:def:graphsem:lhotgs}
  for the interpretation of \lambdaletreccalterms\ as \lambdahotgs.
  }
\end{figure*}

\begin{proposition}
\protect
  \label{prop:trans_lhotgs_correct} \normalfont Let $\allter$ be a
  \lambdaletreccalterm. Suppose that a \termgraph\  $\atg$ over
  $\siglambdabh$, and an \absprefix\ function~$\sabspre$ are
  \rulestranslambdaletreccaltolhotgsgenerated\ from $\allter$.
Then $\sabspre$ is a correct \absprefix\ function for $\atg$, and consequently,
$\atg$ and $\sabspre$ together form a \lambdahotg\ in $\classlhotgs$.
\end{proposition}

There are two sources of \nondeterminism\ in this translation: The
$\snlvarsucc$\nb-rule for shortening prefixes can be applicable at the same
time as other rules. And the \txtlet-rule does not fix the lengths
$l_1,\dots,l_k$ of the abstraction prefixes for the translations of the binding
equations, but admits various choices of prefixes that are shorter than the prefix of the left-hand side. 
Neither kind of \nondeterminism\ affects 
the term graph that is produced, but
in general several \absprefixfunctions,
and thus different \lambdahotg{s}, can be obtained.%

\begin{figure} \fig{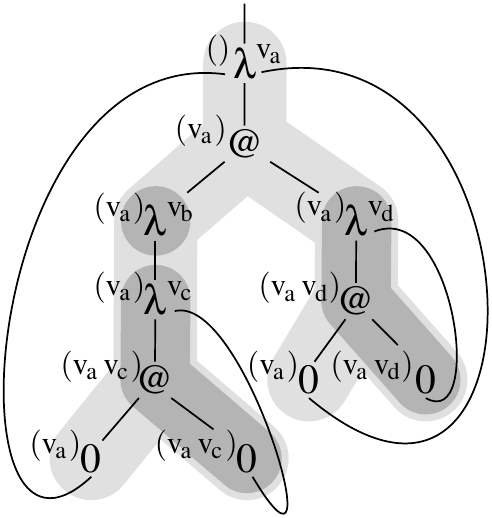} \hfill
  \fig{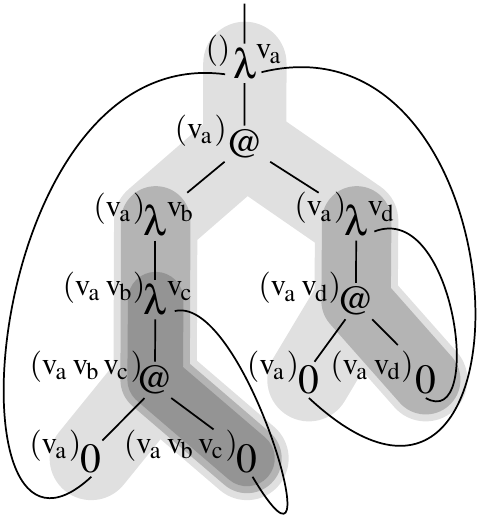}
\caption{\label{fig:eager_more_sharing} Translation of
$\labs{\aavar}{\lapp{(\labs{\bbvar}{\labs{\ccvar}{\lapp{\aavar}{\ccvar}}})}{(\labs{\ddvar}{\lapp{\aavar}{\ddvar}})}}$
with eager scope-closure (left), and with lazy scope-closure (right). While in
the left \termgraph\  four vertices can be shared, 
with as result the translation of the term
$\labs{\aavar}{\letin{\arecvar=\labs{\ccvar}{\lapp{\aavar}{\ccvar}}}{\lapp{(\labs{\bbvar}{\arecvar})}{\arecvar}}}$,
in the right \termgraph\ only a single variable occurrence can be shared.}
\end{figure}

\subsection{Interpretation as \eagscope\ \lambdahotgs}
\label{sec:lhotgs:subsec:eagscope}

  Of the different translations of a \lambdaletrecterm\ into \lambdahotgs\
  we are most interested in the one with the shortest possible abstraction prefixes.
We say that such a term graph has `eager scope-closure`, or that it is `\eagscope'.%
The reason for this choice is illustrated in
Fig.~\ref{fig:eager_more_sharing}: eager-scope closure allows for more sharing.

\begin{definition}[eager scope]\label{def:eagscope}                
  Let
  $\alhotg = \tuple{\verts,\svlab,\svargs,\sroot,\sabspre}$ be a \lambdahotg.
  $\alhotg$ is called \emph{eager-scope} if for every $\bvert\in\verts$ with
  $\abspre{\bvert} = \apre\avert$ for $\apre\in\verts^*$ and
  $\avert\in\verts$, there is a path $\bvert = \bverti{0} \tgsucc \bverti{1}
  \tgsucc \cdots \tgsucc \bverti{m} \tgsucci{0} \avert$ in $\alhotg$ from
  $\bvert$ to $\avert$ with $\abspre{\bvert} \prele \abspre{\bverti{i}}$ for
  all $i\in\setexp{1,\ldots,m}$, and (this follows)
  $\bverti{m}\in\vertsof{\snlvar}$ and $\avert \in \vertsof\sslabs$.
\end{definition}

Hence if a \lambdahotg\ is not eager-scope, then it contains a vertex $\bvert$ with \absprefix\
$\averti{1} \dots \averti{n}$ from which $\averti{n}$ 
is only reachable, if at all,
by leaving the scope of $\averti{n}$.
It can be shown that in this case another \absprefixfunction\ with shorter prefixes exists,
and in which $\averti{n}$ has been removed from the prefix of $\bvert$.  


\begin{proposition}[\eagscope\ = minimal scope; uniqueness of \eagscope\
  \lambdahotgs]\normalfont Let $\alhotgi{i} =
  \tuple{\verts,\svlab,\svargs,\sroot,\sabsprei{i}}$ for $i\in\setexp{1,2}$ be
  \lambdahotgs\ with the same underlying \termgraph. If  $\alhotgi{1}$ is
  \eagscope, then $\length{\absprei{1}{\bvert}} \prele
  \length{\absprei{2}{\bvert)}}$ for all $\bvert\in\verts$. If, in addition,
  also $\alhotgi{2}$ is \eagscope, then $\sabsprei{1} = \sabsprei{2}$. Hence
  \eagscope\ \lambdahotgs\ over the same underlying \termgraph\ are unique.
\end{proposition}


Also, we will call a translation process `\eagscope' if it resolves the \nondeterminism\ in
$\rulestranslambdaletreccaltolhotgs$ in such a way that
it always yields eager-scope \lambdahotgs. 
In order to obtain an \eagscope\
translation we have to consider the following aspects.

\myparagraphbf{Garbage removal}
  In the presence of \emph{garbage}, unused function bindings, a translation cannot be \eagscope.
Consider the term
$\labs{\avar}{\labs{\bvar}{\letin{\arecvar=\avar}{\bvar}}}$. The expendable
binding $\arecvar=\avar$ prevents the application of the $\snlvarsucc$\nb-rule,
and hence the closure of the scope of $\slabs{\avar}$, directly below
$\slabs{\avar}$. Therefore we henceforth assume that \emph{all unused function
bindings are removed} prior to applying the rules
$\rulestranslambdaletreccaltolhotgs$.
A \lambdaletrecterm\ without garbage will be called \emph{garbage-free}.

\myparagraphbf{Short enough prefix lengths in the $\sslet$-rule}
For obtaining an \eagscope\ translation, we will usually stipulate
that the $\snlvarsucc$\nb-rule is applied eagerly, i.e.\ it is given
precedence over the other rules.  This is clearly necessary for keeping the
abstraction prefixes minimal. But how do we choose the prefix lengths $l_1,\dots,l_k$
in the \txtlet-rule? The prefix lengths $l_i$ determine at which position a
binding $\arecvari{i}=\allteri{i}$ is inserted into the abstraction prefixes.
Therefore $l_i$ may not be chosen too short; otherwise a function $\arecvar$
depending on a function $\brecvar$ may end up to the right of $\brecvar$, and
hence may be removed from the prefix by the $\snlvarsucc$-rule prematurely.
preventing completion of the translation. 
Yet simply choosing $l_i = n$ may prevent scopes from being minimal.
For example, when translating the
term~$\labs{\aavar}{\labs{\bbvar}{\letin{\arecvar=\aavar}{\lapp{\lapp{\lapp{\aavar}{\aavar}}{(\lapp{\arecvar}{\aavar})}}{\bbvar}}}}$,
it is crucial to allow shorter prefixes for the binding than for the \inpart.
As shown in Fig.~\ref{fig:rigid_let_rule} the graph on the left does not have
eager scope-closure even if the $\snlvarsucc$-rule is applied eagerly.
Consequently the opportunity for sharing the lower application vertices is
lost.



\myparagraphbf{Required variable analysis}
For choosing the prefixes in the \txtlet-rule correctly,
the translation process must know for each function binding which \lambdavariables\
are `required' on the right-hand side of its definition.
For this we use an analysis obtaining the required variables for positions in a \lambdaletrecterm\
as employed by algorithms for \lambdalifting\
\cite{john:1985,danv:schu:2004}. The term `required variables' was coined
by Moraz\'an and Schultz \cite{mora:schu:2008}.
A \lambdavariable~$\avar$ is called \emph{required at a position $p$} in a \lambdaletrecterm~$\allter$
if $\avar$ is bound by an abstraction above $p$,
and has a free occurrence in the complete unfolding of $\allter$ below $p$
(also recursion variables from above $p$ are unfolded).  
The required variables at position $p$ in $\allter$ can be computed
as those \lambdavariables\ with free occurrences that are reachable from $p$ by a
downwards traversal with the stipulations: on encountering a
\txtlet-binding the $\txtin$-part is entered; when encountering a recursion
variable the traversal continues at the right-hand side of the corresponding
function binding (even if it is defined above $p$).
\vspace{0.5ex}

With the result of the required variable analysis at hand,
  we now define properties of the translation process
  that can guarantee that the resulting \lambdahotg\ is \eagscope.%


\begin{definition}[eager-scope and minimal-prefix generated]\label{def:eagscope:minprefix:generated}
  Let $\allter$ be a \lambdaletrecterm, and let $\alhotg$ be a \lambdahotg.
 
  We say that $\alhotg$ 
  is \emph{\eagscope} $\rulestranslambdaletreccaltolhotgs$\nb-ge\-ne\-ra\-ted 
  from $\allter$
  if $\alhotg$ is $\rulestranslambdaletreccaltolhotgs$\nb-ge\-ne\-ra\-ted from $\allter$ 
  by a  translation process with the following property:
  for every translation box reached during the process with label     
  $\flabs{\ps\;\vvb{\avar}{\avert}{\abindgroup}}{\bllter}$, 
  where $\bllter$ is a subterm of $\allter$ at position $\bpos$,
  it holds that if $\avar$ is not a required variable at $\bpos$ in $\allter$,
  then in the next translation step performed to this box
  either one of the rules $\arecvar$ or $\sslet$ is applied,
  or the prefix is shortened by the $\snlvarsucc$\nb-rule.%
  
  We say that $\alhotg$ 
  is $\rulestranslambdaletreccaltolhotgs$\nb-ge\-ne\-ra\-ted \emph{with minimal prefixes}
  from $\allter$
  if $\alhotg$ is $\rulestranslambdaletreccaltolhotgs$\nb-ge\-ne\-ra\-ted from $\allter$ 
  by a translation process 
  in which minimal prefix lengths are achieved 
  by giving applications of the $\snlvarsucc$\nb-rule precedence 
  over applications of all other rules, 
  and by 
  always choosing prefixes minimally in applications of the $\sslet$-rule.
\end{definition}

\begin{proposition}\normalfont\label{prop:eagscope:generated:minimal:prefixes}
  Let $\alhotg$ be a \lambdahotg\ that is $\rulestranslambdaletreccaltolhotgs$\nb-ge\-ne\-ra\-ted 
  from a garbage-free \lambdaletrecterm~$\allter$. The following statements hold:\vspace{-0.5ex}
  \begin{enumerate}[label=(\roman*)]\setlength{\itemsep}{0.25ex}
    \item{}\label{prop:eagscope:generated:minimal:prefixes:item:i}
      If $\alhotg$ is \eagscope~$\rulestranslambdaletreccaltolhotgs$\nb-ge\-ne\-ra\-ted from $\allter$,
      then $\alhotg$ is \eagscope.
    \item{}\label{prop:eagscope:generated:minimal:prefixes:item:ii}
      If $\alhotg$ is $\rulestranslambdaletreccaltolhotgs$\nb-ge\-ne\-ra\-ted with minimal prefixes from $\allter$,
      then $\alhotg$ is \eagscope~$\rulestranslambdaletreccaltolhotgs$\nb-ge\-ne\-ra\-ted from $\allter$,
      hence by~\ref{prop:eagscope:generated:minimal:prefixes:item:i} $\alhotg$ is \eagscope.
  \end{enumerate}
\end{proposition}


\begin{proposition}\label{prop:trans_eager}\normalfont For every
  \lambdaletreccalterm~$\allter$, $\graphsemC{\classlhotgs}{\allter}$ is
  \eagscope. \end{proposition}

\subsection{Correctness of $\sgraphsemC{\classlhotgs}$ with respect to unfolding semantics}
  \label{subsec:correctness:translation:into:lambdahotgs} 

In preparation of establishing the desired property~\ref{methods:properties:correctness} 
in Sect.~\ref{sec:ltgs}, we formulate, and outline the proof of, 
the fact that the semantics~$\sgraphsemC{\classlhotgs}$ is correct with respect
to the unfolding semantics on \lambdaletrecterms.

\begin{theorem}\label{thm:graphrep:classlhotgs}
  $\unfsem{\allteri{1}} = \unfsem{\allteri{2}}$ 
    if and only if
  $\graphsemC{{\classlhotgs}}{\allteri{1}} \bisim \graphsemC{{\classlhotgs}}{\allteri{2}}$,
  for all \lambdaletreccalterms~$\allteri{1}$ and $\allteri{2}$. 
\end{theorem}

\begin{proof}[Sketch of Proof]
  Central for the proof are \lambdahotgs\ that have tree form and only contain variable backlinks, 
  but no recursive backlinks.
  They form the class $\classlhotreetgs\subsetneqq\classlhotgs$.
  Every $\alhotg\in\classlhotgs$ has a unique `tree unfolding' $\Tree{\alhotg}\in\classlhotreetgs$.
  We make use of the following statements. 
  For all $\allter,\allteri{1},\allteri{2}\in\Ter{\txtlambdaletreccal}$,
  $\aiter,\aiteri{1},\aiteri{2}\in\Ter{\txtinflambdabhcal}$,
  $\alhotg,\alhotgi{1},\alhotgi{2}\in\classlhotgs$, and
  $\alhotreetg,\alhotreetgi{1},\alhotreetgi{2}\in\classlhotreetgs$ it can be shown that:
\begin{gather}
  \allteri{1} \unfoldred \allteri{2}
    \;\;\Rightarrow\;\;
  \graphsemC{\classlhotgs}{\allteri{1}} \convfunbisim \graphsemC{\classlhotgs}{\allteri{2}}
  \label{eq:correctness:llter:unfoldred:llter:then:graphsem:funbisim:graphsem}
  \displaybreak[0]\\[-0.25ex]
  %
  %
  %
  \allter \unfoldinfred \aiter \;\;
    (\text{hence}\; \unfsem{\allter} = \aiter )
    \;\;\Rightarrow\;\;
  \graphsemC{\classlhotgs}{\allter} \convfunbisim \graphsemC{\classlhotgs}{\aiter}
  \label{eq:correctness:graphsem:unfsem:llter:funbisim:graphsem:llter}
  \displaybreak[0]\\[-0.25ex]
  \graphsemC{\classlhotgs}{\aiter} \in \classlhotreetgs 
  \label{eq:correctness:graphsem:aiter:in:lhotreetgs}
  \displaybreak[0]\\[-0.25ex]
  \graphsemC{\classlhotgs}{\aiteri{1}} 
    \iso
  \graphsemC{\classlhotgs}{\aiteri{2}}
    \;\;\Rightarrow\;\;
  \aiteri{1} 
    =
  \aiteri{2}      
  \label{eq:correctness:graphsem:iter:equality}
  \displaybreak[0]\\[-0.25ex]
  \alhotg \convfunbisim \Tree{\alhotg}
  \label{eq:correctness:treeunfolding:funbisim}
  \displaybreak[0]\\[-0.25ex]
  %
  \alhotreetgi{1} \bisim \alhotreetgi{2}
    \;\;\Rightarrow\;\;
  \alhotreetgi{1} \iso \alhotreetgi{2}
  \label{eq:correctness:tree:bisim:tree:then:tree:iso:tree}
  \displaybreak[0]\\[-0.25ex]
  \alhotgi{1} \bisim \alhotgi{2}
    \;\;\Rightarrow\;\;
  \Tree{\alhotgi{1}} \iso \Tree{\alhotgi{2}} 
  \label{eq:correctness:lhotg:bisim:lhotg:then:treeunfolding:iso:treeunfolding}
\end{gather}
%
%
Hereby \eqref{eq:correctness:llter:unfoldred:llter:then:graphsem:funbisim:graphsem}
is used for proving \eqref{eq:correctness:graphsem:unfsem:llter:funbisim:graphsem:llter}, 
and \eqref{eq:correctness:treeunfolding:funbisim} with
    \eqref{eq:correctness:tree:bisim:tree:then:tree:iso:tree}
for \eqref{eq:correctness:lhotg:bisim:lhotg:then:treeunfolding:iso:treeunfolding}.
Now for proving the theorem, let $\allteri{1}$ and $\allteri{2}$ be arbitrary \lambdaletreccalterms.
\begin{description}\setlength{\itemsep}{0.5ex}
  \item[``$\Rightarrow$'':] 
    Suppose 
      $\unfsem{\allteri{1}} = \unfsem{\allteri{2}}$.
    Let $\aiter$ be the infinite unfolding of $\allteri{1}$ and $\allteri{2}$,
    i.e., 
      $ \graphsemC{\classlhotgs}{\allteri{1}} 
          = \aiter 
          = \graphsemC{\classlhotgs}{\allteri{2}}$.
    Then by \eqref{eq:correctness:graphsem:unfsem:llter:funbisim:graphsem:llter} it follows
       $\graphsemC{\classlhotgs}{\allteri{1}} 
          \convfunbisim
        \graphsemC{\classlhotgs}{\aiter}
          \funbisim
        \graphsemC{\classlhotgs}{\allteri{2}}$,
    and hence
    $\graphsemC{\classlhotgs}{\allteri{1}}
       \bisim
     \graphsemC{\classlhotgs}{\allteri{2}}$. 
  \item[``$\Leftarrow$'':]
    Suppose
      $\graphsemC{\classlhotgs}{\allteri{1}}
         \bisim
       \graphsemC{\classlhotgs}{\allteri{2}}$.  
    Then by \eqref{eq:correctness:lhotg:bisim:lhotg:then:treeunfolding:iso:treeunfolding} it follows that
      $\Tree{\graphsemC{\classlhotgs}{\allteri{1}}} 
         \iso
       \Tree{\graphsemC{\classlhotgs}{\allteri{2}}}$.
    Let $\aiteri{1},\aiteri{2}\in\Ter{\txtinflambdabhcal}$
    be the infinite unfoldings of $\allteri{1}$ and $\allteri{2}$,
    i.e.\
      $\aiteri{1} = \unfsem{\allteri{1}}$, and
      $\aiteri{2} = \unfsem{\allteri{2}}$. 
    Then \eqref{eq:correctness:graphsem:unfsem:llter:funbisim:graphsem:llter} 
    together with the assumption entails
      $\graphsemC{\classlhotgs}{\aiteri{1}}
         \bisim
       \graphsemC{\classlhotgs}{\aiteri{2}}$.  
    Since 
       $\graphsemC{\classlhotgs}{\aiteri{1}},
                   \graphsemC{\classlhotgs}{\aiteri{2}}\in\classlhotreetgs$
    by \eqref{eq:correctness:graphsem:aiter:in:lhotreetgs},
    it follows by \eqref{eq:correctness:tree:bisim:tree:then:tree:iso:tree} that
      $\graphsemC{\classlhotgs}{\aiteri{1}}
         \iso
       \graphsemC{\classlhotgs}{\aiteri{2}}$.
    Finally, by using \eqref{eq:correctness:graphsem:iter:equality} we get
      $\aiteri{1}
         = 
       \aiteri{2}$, and hence
      $\unfsem{\allteri{1}}
         = 
       \aiteri{1}
         = 
       \aiteri{2}
         = 
       \unfsem{\allteri{2}}$.\hspace*{\fill}\qed
\end{description} 
\renewcommand{\qed}{} 
\end{proof}


\section{\Lambdatgs}
  \label{sec:ltgs} 

While modelling sharing expressed by \lambdaletrecterms\ through
\lambdahotgs\ 
              facilitates comparisons via bisimilarity,  
it is not immediately clear how 
the compactification of \lambdahotgs\ via the bisimulation collapse~$\scoll$
for \lambdahotgs\ (which has to respect scopes in the form of the \absprefix\ functions)
can be computed efficiently.
We therefore develop an implementation as first-order \termgraphs, 
for which standard methods are available.

Due to Ex.~\ref{ex:counternat}, the
scoping information cannot just be discarded,
as functional bisimilarity on the underlying \termgraphs\ 
does not faithfully implement functional bisimilarity on \lambdahotgs.
Therefore the scoping information has to be incorporated 
in the first-order interpretation
in some way. 
We accomplish this by extending $\siglambdabh$ with
\snlvarsuccvertices, scope delimiters, that signify the end of scopes.
When translating a \lambdahotg{} into 
a first-order \termgraph, 
\snlvarsuccvertices\ are placed along those edges 
in the underlying \termgraph\  at which the abstraction prefix decreases in the \lambdahotg.

\begin{example}[Adding \snlvarsuccvertices]
Consider the terms in Ex.~\ref{ex:counternat} and their \lambdahotgs\ in
Ex.~\ref{ex:counternat_lhotgs}. In the first-order interpretation below, the
shading is just for illustration purposes; it is \emph{not} part of the
structure, and does \emph{not directly} impair functional bisimulation.
  \begin{center}
    $
    \begin{aligned}[c] 
      \figsmall{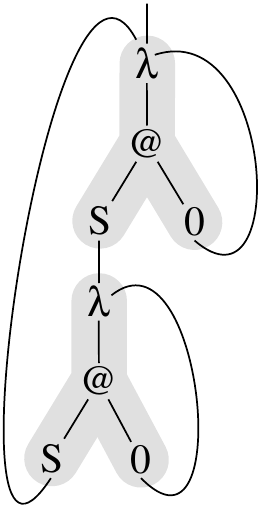}
    \end{aligned} 
       \hspace*{2ex}\sfunbisim\hspace*{2ex}
    \begin{aligned}[c]
      \figsmall{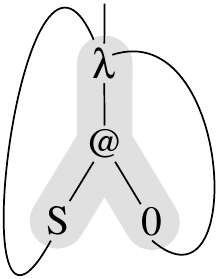}
    \end{aligned}
       \hspace*{2ex}\not\sinvfunbisim\hspace*{2ex}
    \begin{aligned}[c]  
      \figsmall{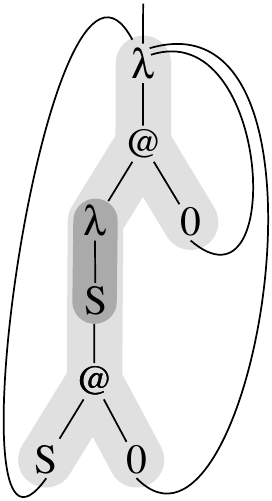}
    \end{aligned}
    $   
  \end{center}
\end{example}
\noindent
The addition of scope delimiters resolves the problem of
Ex.\ref{ex:counternat}. They adequately represent the scoping information.

As for \lambdahotgs, we will define correctness conditions by means
of an \absprefixfunction. However, the current approach with unary delimiter vertices
leads to a problem.

\begin{example}[$\snlvarsucc$-backlinks]\label{ex:S-backlinks}
The \termgraph\  with scope delimiters on the left admits
a functional bisimulation that fuses two \snlvarsuccvertices\ that
close different scopes. We cannot hope to find a unique abstraction
prefix for the resulting fused $\snlvarsucc$-vertex. This is remedied on the right
by using a variant representation that requires backlinks from each
\snlvarsuccvertex\ to the abstraction vertex whose scope it closes.
Then \snlvarsuccvertices\ can only be fused 
if the corresponding abstractions have already been merged.
Hence in the presence of $\snlvarsucc$\nb-backlinks,   
as in the right illustration below, 
only the variable vertex can be shared. 
  \begin{center}
    $
    \begin{aligned}[c] 
      \figsmall{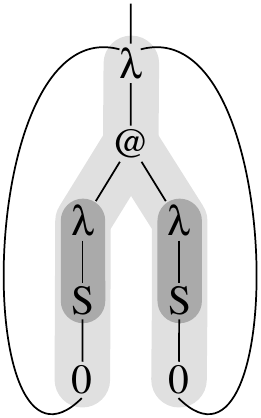}
    \end{aligned} 
       \hspace*{0.5ex}\sfunbisim\hspace*{-2.5ex}
    \begin{aligned}[c]
      \figsmall{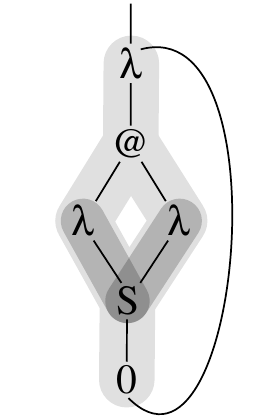}
    \end{aligned}
    \hfill
    \begin{aligned}[c] 
      \figsmall{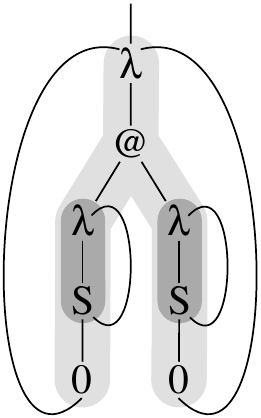}
    \end{aligned} 
       \hspace*{0.5ex}\sfunbisim\hspace*{-2.5ex}
    \begin{aligned}[c]
      \figsmall{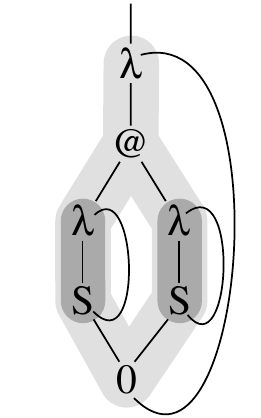}
    \end{aligned}
    $   
  \end{center}
\end{example}
\noindent
Therefore we consider \termgraphs\  over the extension $\siglambdaSbh$
of $\siglambdabh$ with a symbol $\snlvarsucc$ of arity $2$; one
edge targets the successor vertex, the other is a backlink.
We give correctness conditions, similar as for \lambdahotgs,
and define the arising class of `\lambdatgs'. 

\begin{definition}[correct abstraction-prefix function for \termgraphs\ over $\siglambdaSbh$]%
    \label{def:abspre:function:siglambdaij}\normalfont
  Let $\atg = \tuple{\verts,\svlab,\svargs,\sroot}$ be a $\siglambdaSbh$\nb-term-graph.


  An \emph{\absprefixfunction} $\sabspre \funin \verts \to \verts^*$ on $\atg$ 
  is called \emph{correct} if
  for all $\bvert,\bverti{0},\bverti{1}\in\verts$ and $k\in\setexp{0,1}$ it holds:
  \begin{align*}
    & \abspre{\sroot} = \emptyword 
    & 
    (\text{root})
    \displaybreak[0]\\
    & \abspre{\bh} = \emptyword 
    &
    \hspace*{-5ex}
    (\text{black hole})
    \\
      \bvert\in\vertsof{\sslabs}
        \,\logand\,
      \bvert \tgsucci{0} \bvertbp{0}{}
      \; & \Rightarrow \;
    \abspre{\bvertbp{0}{}} = \abspre{\bvert} \bvert 
    & 
    (\sslabs)
  \displaybreak[0]\\
    \bvert\in\vertsof{\sslapp}
      \,\logand\,
    \bvert \tgsucci{k} \bvertbp{k}{}
      \; & \Rightarrow \;\;
    \abspre{\bvertbp{k}{}} = \abspre{\bvert} 
    & 
    (\sslapp)  
  \displaybreak[0]\\[0.4ex]
    \bvert\in\vertsof{\snlvar}
      \,\logand\,
    \bvert \tgsucci{0} \bvertbp{0}{}
      \; & \Rightarrow \;\;
    \left\{\hspace*{1pt}   
    \begin{aligned}[c]  
      & 
        \bvertbp{0}{}\in\vertsof{\sslabs}
      \\[-0.5ex]
      & 
      \,\logand\,
      \abspre{\bvertbp{0}{}} \bvertbp{0}{} = \abspre{\bvert} 
    \end{aligned}
    \right.      
    & 
    (\snlvar)_1 
    \displaybreak[0]\\
    \bvert\in\vertsof{\snlvarsucc}
      \,\logand\,
    \bvert \tgsucci{0} \bvertbp{0}{}
      \; & \Rightarrow \;\;
    \left\{\hspace*{1pt} 
    \begin{aligned}[c]
      &  
      \abspre{\bvertbp{0}{}} \avert
        =
      \abspre{\bvert}
      \\[-0.5ex]
      & \hspace*{6ex}  
      \text{for some $\avert\in\verts$} 
    \end{aligned}
    \right.
    &
    (\snlvarsucc)_1 
  \displaybreak[0]\\
    \bvert\in\vertsof{\snlvarsucc}
        \,\logand\,
    \bvert \tgsucci{1} \bvertbp{1}{}  
        \; & \Rightarrow \;\;
    \left\{\hspace*{1pt}    
    \begin{aligned}[c]    
      &
      \bvertbp{1}{} \in \vertsof{\sslabs}
      \\[-0.5ex]
      &
        \,\logand\,        
      \abspre{\bvertbp{1}{}} \bvertbp{1}{}
        =
      \abspre{\bvert}
    \end{aligned}
    \right.
    &
    (\snlvarsucc)_2  
   \end{align*}   
\end{definition}
\noindent
While in \lambdahotgs\ the abstraction prefix can shrink by several vertices along
an edge (cf.\ Def.~\ref{def:abspre:function}), here the situation is strictly regulated:
the prefix can only shrink by one variable, and only along the outgoing edge of a delimiter vertex.

\begin{proposition}[uniqueness of the abstraction prefix function]
\label{prop:unique_abspre}
\normalfont
  Let $\atg$ be a \termgraph\  over the signature $\siglambdaSbh$. If
  $\sabspre_1$ and $\sabspre_2$ are correct abstraction prefix functions of
  $\atg$, then $\sabspre_1 = \sabspre_2$.
\end{proposition}


\begin{definition}[\lambdatg]\label{def:ltg}\normalfont
  A \emph{\lambdatg} is a \termgraph\  $\altg =
  \tuple{\verts,\svlab,\svargs,\sroot}$ over $\siglambdaSbh$ that has a
  correct \absprefix\ function 
  (which is not a part of $\altg$). 
  The class of \lambdatg{s} is $\classltgs$.
\end{definition}

\begin{definition}[eager scope] 
  A \lambdatg\ $\altg$ is called \emph{\eagscope}
  if together with its \absprefixfunction\ it meets the condition 
                                                                  in Def.~\ref{def:eagscope}.
  $\classeagltgs$ denotes the class of \eagscope\ graphs.

\end{definition}

\subsection{Correspondence between $\lambda$-ho- and \lambdatgs} 
  \label{sec:ltgs:subsec:trans}

The correspondences between \lambdahotgs\ and \lambdatgs:
\begin{align*}
  \slhotgstoltgs \funin \classlhotgs \to \classltgs &
  &  \hspace{3ex}
  \sltgstolhotgs \funin \classltgs \to \classlhotgs &
\end{align*}
are defined as follows: For obtaining $\lhotgstoltgs{\alhotg}$ for a $\alhotg\in\classlhotgs$,
insert scope-delimiters wherever the prefix decreases, as illustrated in Fig.~\ref{fig:insertS}.
For obtaining $\ltgstolhotgs{\altg}$ for a $\altg\in\classltgs$,
retain the \absprefixfunction, and
remove every delimiter vertex from $\altg$,
thereby connecting its incoming edge with its outgoing edge. 
For formal definitions and well-definedness of $\sltgstolhotgs$ and
$\slhotgstoltgs$, see \cite{grab:roch:2013:a:TERMGRAPH}.

\begin{figure}[th]
\transpicture{
  \ltgnode{a}{$a$};
  \addPrefix{a}{v_1 \dots v_n};
  \ltgnode[node distance=9mm,below=of a]{b}{$b$};
  \addPrefix{b}{v_1 \dots v_m};
  \draw[->](a) to (b);
}
\vcentered{$\overset{m<n}{
                          \implies}$}
\hspace{-4.5ex}
\transpicture{
  \ltgnode{a}{$a$};
  \addPrefix{a}{v_1 \dots v_n};
  \ltgnode[below=of a]{s1}{\snlvarsucc}; \addPrefix{s1}{v_1 \dots v_{n-1}}; \draw[->](a) to (s1);
  \node[node distance=6mm,below=of s1](si){$\vdots$}; \draw[->](s1) to (si);
  \ltgnode[node distance=6mm,below=of si]{sn}{\snlvarsucc}; \addPrefix{sn}{v_1 \dots v_{m+1}}; \draw[->](si) to (sn);
  \ltgnode[below=of sn]{b}{$b$};
  \addPrefix{b}{v_1 \dots v_m};
  \draw[->](sn) to (b);
}
\hfill
\vcentered{\figsmall{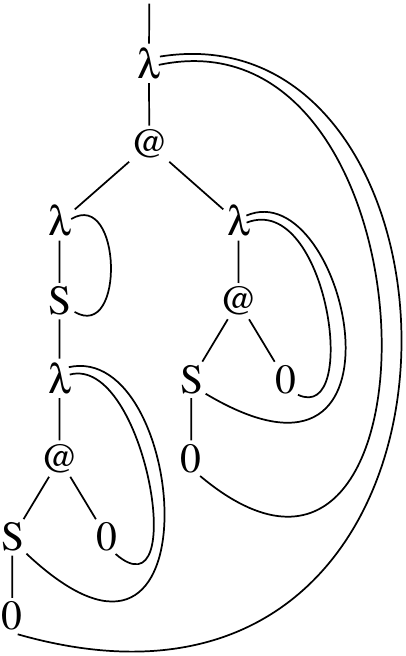}}
\caption{\label{fig:insertS}
  Left: definition of $\slhotgstoltgs$ by inserting \snlvarsuccvertices,
        between edge-connected vertices of a \lambdahotg.
  Right: interpretation $\lhotgstoltgs{\alhotg}$ of the \eagscope\ \lambdahotg~$\alhotg$
         in Fig.~\ref{fig:eager_more_sharing}.}
\end{figure}

Note that a \lambdahotg\ may have multiple corresponding \lambdatgs\ that
differ only with respect to their `degree' of $\snlvarsucc$\nb-sha\-ring
(the extent to which $\snlvarsucc$\nb-ver\-ti\-ces occur shared). 
$\slhotgstoltgs$ maps to a \lambdatg\ with no sharing of $\snlvarsucc$\nb-ver\-ti\-ces at all.

The proposition below guarantees the usefulness of the translation~$\slhotgstoltgs$
for implementing functional bisimulation on \lambdahotgs.
In particular, this is due to items~\ref{prop:corr:laphotgs:ltgs:item:preserve} and \ref{prop:corr:laphotgs:ltgs:item:eagerness}.
As formulated by item~\ref{prop:corr:laphotgs:ltgs:item:retraction:section},
$\sltgstolhotgs$ is a retraction of $\slhotgstoltgs$
(and $\slhotgstoltgs$ a section of $\sltgstolhotgs$).
The converse is not the case, 
yet it holds up to $\snlvarsucc$\nb-sharing by item~\ref{prop:corr:laphotgs:ltgs:item:retraction:section:up:to}.
For the proof, we refer to our article \cite{grab:roch:2013:a:TERMGRAPH}.%

\begin{samepage}
\begin{proposition}[correspondence with \lambdahotg{s}]\label{prop:corr:laphotgs:ltgs}\normalfont
  \mbox{}\vspace*{-0.5ex}\nopagebreak
\begin{enumerate}[label=(\roman*)]\setlength{\itemsep}{0.25ex} 
  \item{}\label{prop:corr:laphotgs:ltgs:item:retraction:section} 
    $\scompfuns{\sltgstolhotgs}{\slhotgstoltgs} = \sidfunon{\classlhotgs}$.
  \item{}\label{prop:corr:laphotgs:ltgs:item:retraction:section:up:to}.
    $(\compfuns{\slhotgstoltgs}{\sltgstolhotgs)}{\altg}
           \Sfunbisim
         \altg$
    holds for all $\altg\in\classltgs\,$.
  \item{}\label{prop:corr:laphotgs:ltgs:item:preserve}
    $\sltgstolhotgs$ and $\slhotgstoltgs$ 
    preserve and reflect functional bisimulation~$\sfunbisim$ and bisimulation~$\sbisim$
    on $\classlhotgs$ and $\classltgs$.
  \item{}\label{prop:corr:laphotgs:ltgs:item:eagerness}
    $\sltgstolhotgs$ and $\slhotgstoltgs$ preserve and reflect the property eager-scope.
  \item{}\label{prop:corr:laphotgs:ltgs:item:closed}
    $\classltgs$ is closed under $\sfunbisimS$, $\sconvfunbisimS$, and $\sbisimS$.
  \item{}\label{prop:corr:laphotgs:ltgs:item:isomorphism}  
    $\slhotgstoltgs$ and $\sltgstolhotgs$
    induce isomorphisms between $\classlhotgs$ and $\factorset{\classltgs}{\sbisimsubscriptS}$. 
\end{enumerate}
\end{proposition}
\end{samepage}


\subsection{Closedness of $\classltgs$ under functional bisimulation}

While preservation of $\sfunbisim$ by $\slhotgstoltgs$ is necessary
for its implementation via $\sfunbisim$ on $\classltgs$,
the practicality of the interpretation $\slhotgstoltgs$ 
also depends on the closedness of $\classltgs$ under $\sfunbisim$. 
Namely, if the bisimulation collapse $\altg = \coll{\lhotgstoltgs{\alhotg}}$ 
of the interpretation of some $\alhotg\in\classlhotgs$
were not contained in $\classltgs$, then the converse interpretation $\sltgstolhotgs$ 
could not be applied to $\altg$ in order to obtain the bisimulation collapse of $\alhotg$.  

A subclass $\aclass$ of the \termgraphs\  over a signature $\asig$
is called \emph{closed under functional bisimulation} 
if, for all \termgraphs\  $\atg$, $\atgacc$ over $\asig$,
whenever $\atg\in\aclass$ and $\atg \funbisim \atgacc$,
then also $\atgacc\in\aclass$.

Note that for obtaining this property 
the use of variable backlinks, and backlinks for delimiter vertices 
is crucial (cf.\ Ex.~\ref{ex:S-backlinks}).


Yet the class $\classltgs$ is actually not closed under $\sfunbisim\hspace*{1pt}$:
See Fig.~\ref{fig:ltgs_not_closed_under_funbisim} at the top for 
a homomorphism from a non-eager-scope \lambdatg\ to a \termgraph\ 
over $\siglambdaSbh$ that is not a \lambdatg\ (as suggested by the
overlapping scopes). The use of eager scope-closure remedies the situation,
see at the bottom: then the bisimulation collapse is a \lambdatg.
%
This motivates the following theorem,
which is proved in the extended report of \cite{grab:roch:2013:a:TERMGRAPH}.
It justifies property~\ref{methods:properties:closedness} with $\classeagltgs$ for $\classltgs$.

\begin{figure}[ht]
\hspace{1cm}
\vcentered{\figsmall{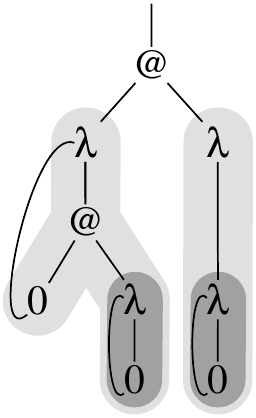}}
\hfill\vcentered{$\funbisim$}\hfill
\vcentered{\figsmall{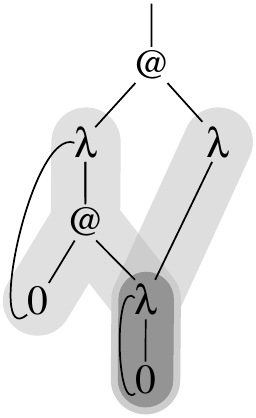}}
\hspace{1cm}
\\\mbox{}
\hspace{1cm}
\vcentered{\figsmall{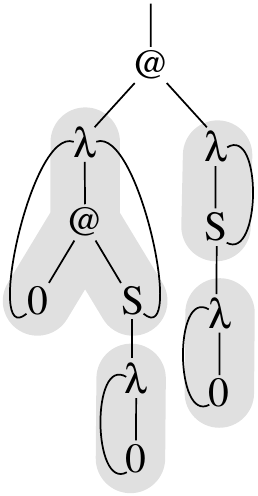}}
\hfill\vcentered{$\funbisim$}\hfill
\vcentered{\figsmall{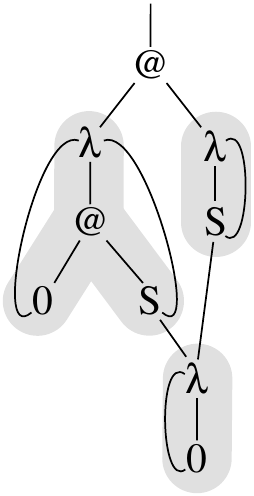}}
\hspace{1cm}
\caption{\label{fig:ltgs_not_closed_under_funbisim}
$\classltgs$ is not closed under functional bisimulation, yet~$\classeagltgs$~is.}
\end{figure}

\begin{theorem}\label{thm:classeagltgs:closed:under:funbisim}
  The class $\classeagltgs$ of eager-scope \lambdatgs\ is closed under
  functional bisimulation $\sfunbisim$.   
\end{theorem}


\subsection
           {\lambdatg\ semantics for \lambdaletrecterms}

We will consider in fact two interpretations of \lambdaletrecterms\ as \lambdatgs:
first we define $\sgraphsemCmin{\classltgs}$ 
as the composition of $\sgraphsemC{\classlhotgs}$ and $\slhotgstoltgs$; 
then we define 
       the semantics $\sgraphsemC{\classltgs}$ with more fine-grained
\snlvarsucc-sharing, which is necessary for defining a readback with the property~\ref{methods:properties:readback}.

By composing the interpretation $\slhotgstoltgs$ of \lambdahotgs\ as \lambdatgs\
with the \lambdahotg\ semantics $\sgraphsemC{\classlhotgs}$,
a semantics of \lambdaletrecterms\ as \lambdatgs\ is obtained. 
There is, however, a more direct way to define this semantics:
by using an adaptation of the translation rules $\rulestranslambdaletreccaltolhotgs$
in Fig.~\ref{fig:def:graphsem:lhotgs}, on which $\sgraphsemC\classlhotgs$ is based.
For this, let $\rulestranslambdaletreccaltoltgs$ be the result of replacing
the rule~$\snlvarsucc$ in $\rulestranslambdaletreccaltolhotgs$ by the version
in Fig.~\ref{fig:def:graphsem:ltgs}.
While applications of this variant of the $\snlvarsucc$\nb-rule also shorten the \absprefix, 
they additionally produce a delimiter vertex. 

\begin{figure}[bth]
  \input{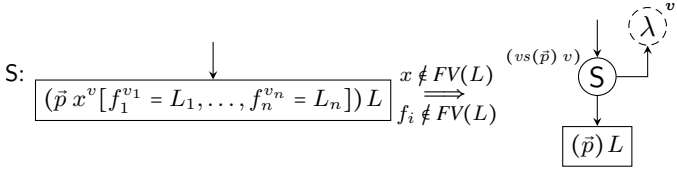}
  \caption{\label{fig:def:graphsem:ltgs}
           Delimiter-vertex producing version of the $\snlvarsucc$-rule in Fig.~\ref{fig:def:graphsem:lhotgs}
   }
\end{figure}
  Here, at the end of the translation process, 
  every loop on an indirection vertex with a prefix of length $n$ 
  is replaced by a chain of $n$ $\snlvarsucc$\nb-vertices 
  followed by a black hole vertex.%
Note that, while the system $\rulestranslambdaletreccaltoltgs$ inherits all of the non-determinism of $\rulestranslambdaletreccaltolhotgs$,
the possible degrees of freedom have additional impact on the result, 
because now they also determine the precise degree of $\snlvarsucc$\nb-vertex sharing.

By analogous stipulations as in Def.~\ref{def:eagscope:minprefix:generated}
we define the conditions under which 
a \lambdatg\ is called \emph{\eagscope} $\rulestranslambdaletreccaltoltgs$\nb-ge\-ne\-ra\-ted,
or $\rulestranslambdaletreccaltoltgs$\nb-ge\-ne\-ra\-ted \emph{with minimal prefixes}, 
from a \lambdaletrecterm.
For these notions, statements entirely analogous to Prop.~\ref{prop:eagscope:generated:minimal:prefixes} hold.%

\begin{definition}{}\normalfont
    \label{def:graphsem-min:ltgs}
  The \emph{semantics} $\sgraphsemCmin{\classltgs}$ 
  for \lambdaletrecterms\ as \lambdatgs\ is defined as
  $\sgraphsemCmin{\classltgs} \funin \Ter{\txtlambdaletreccal} \to \classeagltgs$,
  $ \allter \mapsto \graphsemCmin{\classltgs}{\allter} \defdby\,$
  the \eagscope\ \termgraph\  that is
  $\rulestranslambdaletreccaltoltgs$\nb-ge\-ne\-ra\-ted with minimal prefixes
  from a garbage-free version $\allteracc$ of $\allter$.
\end{definition}
For an example, see Ex.~\ref{ex:translations} below.
In $\sgraphsemCmin{\classltgs}$, 
`\noSsh' also indicates 
that 
\lambdatgs\ obtained via this semantics 
exhibit minimal (in fact no) sharing (two or more incoming edges) of \snlvarsuccvertices.
This is substantiated by the next proposition, in the light of the fact that
$\slhotgstoltgs$ does not
create any shared $\snlvarsucc$\nb-vertices.
\begin{proposition}\label{prop:graphsem:ltgs:min:graphsem:ltgs}
  $\sgraphsemCmin{\classltgs} = \scompfuns{\slhotgstoltgs}{\sgraphsemC{\classlhotgs}}$. 
\end{proposition}

Hence $\sgraphsemCmin{\classltgs}$ only yields \lambdatgs\ without sharing of \snlvarsuccvertices, 
and therefore its image cannot be all of $\classeagltgs$.
As a consequence,
we cannot hope to define a readback function $\sreadback$ with respect to $\sgraphsemCmin{\classltgs}$
that has the desired property \ref{methods:properties:readback}, 
because that requires 
that the image of the semantics is $\classeagltgs$ in its entirety. 

Therefore we modify the definition of $\sgraphsemCmin{\classltgs}$
to obtain another \lambdatg\ semantics~$\sgraphsemC{\classltgs}$ 
with image $\image{\sgraphsemC{\classltgs}} = \classeagltgs$. 
This is achieved by
letting the \txtlet-bin\-ding-struc\-ture of the
  \lambdaletrecterm\ influence the degree of $\snlvarsucc$-sharing as much as possible,
  while staying \eagscope.  

We say that a \lambdahotg~$\alhotg$ 
is \emph{\eagscope\ $\rulestranslambdaletreccaltolhotgs$\nb-ge\-ne\-ra\-ted with maximal prefixes}
from a \lambdaletrecterm~$\allter$
if $\alhotg$ is $\rulestranslambdaletreccaltolhotgs$\nb-ge\-ne\-ra\-ted from $\allter$ 
by a translation process
in which in applications of the $\sslet$-rule 
the prefixes are chosen maximally, but so that the \eagscope\ property of the process is not compromised.
It can be shown that this condition fixes the prefix lengths per application of the $\sslet$\nb-rule.%

\begin{definition}{}\normalfont
    \label{def:graphsem:ltgs}
  The \emph{semantics} $\sgraphsemC{\classltgs}$ 
  for \lambdaletrecterms\ as \lambdatgs\ is defined as
  $\sgraphsemC{\classltgs} \funin \Ter{\txtlambdaletreccal} \to \classeagltgs$,
  $ \allter \mapsto \graphsemC{\classltgs}{\allter} \defdby\,$
  the \lambdatg\ 
  that is \eagscope\ $\rulestranslambdaletreccaltoltgs$\nb-ge\-ne\-ra\-ted 
           with maximal prefixes from a garbage-free version $\allteracc$ of $\allter$.
\end{definition}

\begin{figure}
\begin{minipage}[b]{2.6cm}
$  \graphsemC\classltgs{\allteri1}$\\
$= \graphsemCmin{\classltgs}{\allteri1}
 = \graphsemCmin{\classltgs}{\allteri2}$\\
$= \graphsemCmin{\classltgs}{\allteri3}
 = \graphsemCmin{\classltgs}{\allteracci3}$
\\[0.7cm]
  \hspace*{\fill}
$\graphsemC\classltgs{\allteri{2}}$
\\[1.2cm]
  \hspace*{\fill}
$\graphsemC\classltgs{\allteri{3}} = \graphsemC\classltgs{\allteracci3}$
\\[6mm]
\end{minipage}
\hfill
\figsmall{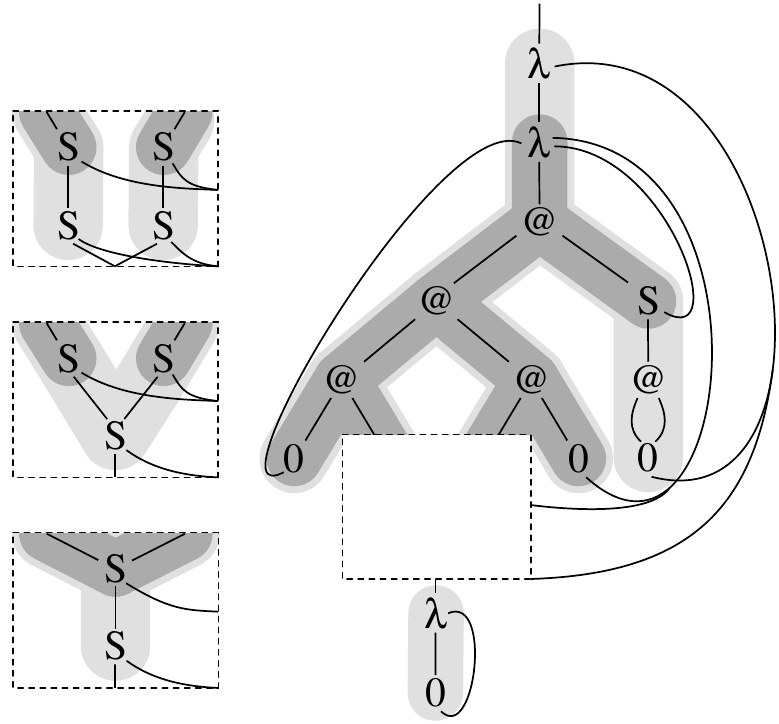}
\caption{\label{fig:translations}
Translation of the \protect\lambdaletrecterms\ from Ex.~\ref{ex:translations}
with the semantics $\sgraphsemCmin{\classltgs}$ and $\sgraphsemC{\classltgs}$. 
For legibility some backlinks are merged.
}
\end{figure}

\begin{proposition}\label{prop:graphsem:ltgs:min:graphsem:ltgs:S}\normalfont
  $\graphsemCmin{\classltgs}{\allter} \funbisimS \graphsemC{\classltgs}{\allter}$
  holds for all \lambdaletrecterms~$\allter\,$.
\end{proposition}

Now due to this, and due to Prop.~\ref{prop:corr:laphotgs:ltgs},~\ref{prop:corr:laphotgs:ltgs:item:preserve},
the statement of Thm.~\ref{thm:graphrep:classlhotgs} can be transferred to $\classltgs$,
yielding property~\ref{methods:properties:correctness} for $\sgraphsemC{\classltgs}$. 

\begin{theorem}\label{thm:graphrep:classltgs}
  For all \lambdaletreccalterms~$\allteri{1}$ and $\allteri{2}$
  the following holds:
  $\unfsem{\allteri{1}} = \unfsem{\allteri{2}}$ 
    if and only if
  $\graphsemC{{\classltgs}}{\allteri{1}} \bisim \graphsemC{{\classltgs}}{\allteri{2}}$.
\end{theorem}

\begin{example}\label{ex:translations}
Consider the following four \lambdaletrecterms:
\[
\begin{aligned}
  \allteri1 &= \letin{I=\labs{z}{z}}{\labs{x}{\labs{y}{\letin{f=x}{\lapp{(\lapp{(\lapp{y}{I})}{(\lapp{I}{y})})}{(\lapp{f}{f})}}}}}\\
  \allteri2 &= \labs{x}{\letin{I=\labs{z}{z}}{\labs{y}{\letin{f=x}{\lapp{(\lapp{(\lapp{y}{I})}{(\lapp{I}{y})})}{(\lapp{f}{f})}}}}}\\
  \allteri{3} &= \labs{x}{\labs{y}{\letin{I=\labs{z}{z},\,f=x}{\lapp{(\lapp{(\lapp{y}{I})}{(\lapp{I}{y})})}{(\lapp{f}{f})}}}}\\
  \allteracci{3} &= \labs{x}{\letin{I=\labs{z}{z}}{\labs{y}{\letin{f=x,\ g=I}{\lapp{(\lapp{(\lapp{y}{g})}{(\lapp{g}{y})})}{(\lapp{f}{f})}}}}}\\
\end{aligned}
\]
The three possible fillings of the dashed area in Fig.~\ref{fig:translations}
depict the translations
$\graphsemC\classltgs{\allteri1}$,
$\graphsemC\classltgs{\allteri2}$, and
$\graphsemC\classltgs{\allteri3} = \graphsemC\classltgs{\allteracci3}$.
The translations of the four terms with
$\inMath{\sgraphsemCmin}$ are identical:\\
$\graphsemCmin{\classltgs}{\allteri1} =
 \graphsemCmin{\classltgs}{\allteri2} =
 \graphsemCmin{\classltgs}{\allteri3} =
 \graphsemCmin{\classltgs}{\allteracci3} =
 \graphsemC{\classltgs}{\allteri1}$.
\end{example}

\section{Readback of \lambdatgs}
  \label{sec:readback}

In this section we describe how from a given \lambdatg~$\altg$ a \lambdaletrecterm~$\allter$
that represents $\altg$ (i.e.\ for which $\graphsemC{\classltgs}{\allter} = \altg$ holds) can be `read back'.
For this purpose we define a process based on term synthesis rules. 
It defines a readback function from \lambdatgs\ to \lambdaletreccalterms. 
We illustrate this process by an example, formulate its most important properties,
and sketch the proof of~\ref{methods:properties:readback}.

The idea underlying the definition of the readback procedure is the following:
For a given \lambdatg~$\altg$, a spanning tree $\aspantree$ for $\altg$
(augmented with a dedicated root node) is constructed
that severs cycles of $\altg$ at (some) recursive bindings or variable $\snlvarsucc$\nb-back\-links.
Now the spanning tree $\aspantree$
facilitates an inductive bottom-up (from the leafs upwards) synthesis process
along $\aspantree$, which labels the edges of $\atg$ (except for variable backlinks) with
prefixed \lambdaletrecterms. For this process we use local rules (see
Fig.~\ref{fig:readback:rules}) that synthesise labels for incoming edges of a
vertex from the labels of
its outgoing edges. Eventually the readback of $\altg$ is obtained as
the label for the edge that singles out the root of \termgraph. 

The design of the readback rules is based on a decision about where 
\letbindings\ are placed in the synthesised term. 
Namely there exists some freedom for
these placements, as certain kind of shifts of \letexpressions\ (\letfloating\ steps
\cite{grab:roch:2013:e:IWC}) preserve the \lambdatg\ interpretation.
Here, \letbindings\ will always be declared in a \letexpression\ that is
placed as high up in the term as possible: a binding arising from the term synthesised for a
shared vertex $\bvert$ is placed in a \letexpression\ that is created at the
enclosing \lambdaabstraction\ of $\bvert$ (the leftmost vertex in the
\absprefix\ $\abspre{\bvert}$ of~$\bvert$).  


\begin{figure}[tb]
\mbox{}\hfill
\rbpicture{
\node(rootrb){$\femptylabs{\labs{x}{\letin{f=\labs{y}{\lapp{\lapp{f}{x}}{y}}}{\lapp{f}{f}}}}$};
\ltgnode[node distance=4mm,below=of rootrb]{root}{$\top$}; \draw[->] (rootrb) to (root);
\ltgnode[below=of root]{absx}{$\sslabs$};
\addPos{absx}{x};
\rbedge{$\flabs{*[]}{\labs{x}{\letin{f=\labs{y}{\lapp{\lapp{f}{x}}{y}}}{\lapp{f}{f}}}}$}{root}{absx};
\addPrefix[node distance=2mm]{root}{};
\ltgnode[below=of absx]{app1}{$\sslapp$};
\rbedge{$\flabs{*[]\;x[f=\labs{y}{\lapp{\lapp{f}{x}}{y}}]}{\lapp{f}{f}}$}{absx}{app1};
\addPrefix[node distance=2mm]{absx}{};
\ltgnode[below=of app1]{indir}{\indir};
\rbedge[bend right,left]{$\flabs{*[]\;x[f=\labs{y}{\lapp{\lapp{f}{x}}{y}}]}{f}$}{app1}{indir};
\addPrefix[node distance=2mm]{app1}{x};
\addPos{indir}{\;\;f};
\rbedge[bend left,right,dotted]{$\flabs{*[]\;x[f=\nodef]}{f}$}{app1}{indir};
\ltgnode[below=of indir]{absy}{$\sslabs$};
\addPos{absy}{y};
\rbedge{$\flabs{*[]\;x[f=\nodef]}{\labs{y}{\lapp{\lapp{f}{x}}{y}}}$}{indir}{absy};
\addPrefix[node distance=2mm]{indir}{x};
\ltgnode[below=of absy]{app2}{$\sslapp$};
\rbedge{$\flabs{*[]\;x[f=\nodef]\;y[]}{\lapp{\lapp{f}{x}}{y}}$}{absy}{app2};
\addPrefix[node distance=2mm]{absy}{x};
\node[below=of app2](bapp2){};
\ltgnode[left=of bapp2]{s}{$\snlvarsucc$};
\rbedge[left]{$\flabs{*[]\;x[f=\nodef]\;y[]}{\lapp{f}{x}}$}{app2}{s};
\addPrefix[node distance=2mm]{app2}{x\;y};
\ltgnode[below=of s]{app3}{$\sslapp$};
\rbedge{$\flabs{*[]\;x[f=\nodef]}{\lapp{f}{x}}$}{s}{app3};
\addPrefix[node distance=2mm]{s}{x\;y};
\draw[dotted,->,bend right=80](s) to (absy);
\node[below=of app3](bapp3){};
\ltgnode[right=of bapp3]{x}{$\snlvar$};
\addPrefix[node distance=3mm]{x}{x};
\draw[dotted,->,bend right=90](x) to (absx);
\rbedge{$\flabs{*[]\;x[]}{x}$}{app3}{x};
\addPrefix[node distance=2mm]{app3}{x};
\ltgnode[right=of bapp2]{y}{$\snlvar$};
\addPrefix[node distance=3mm]{y}{x\;y};
\draw[dotted,->,bend right=80](y) to (absy);
\rbedge{$\flabs{*[]\;x[]\;y[]}{y}$}{app2}{y};
\node[left=of app3,yshift=-1mm,xshift=6mm](lapp3){$\flabs{*[]\;x[f=\nodef]}{f}$};
\draw[dotted,bend left](app3) to (lapp3);
\draw[dotted,->,bend left=80](lapp3) to (indir);
}
\mbox{}\hfill
\caption{\label{fig:ex:readback}%
         Example of the readback synthesis from a \lambdatg. 
          }
\end{figure}

\begin{definition}[readback of \lambdatgs]%
  \label{def:readback}
  Let $\altg\in\classltgs$ be a \lambdatg.
  The process of computing the readback of $\altg$ (a \lambdaletreccalterm)
  consists of the following five steps, starting on $\altg\,$: 
  \begin{description}\setlength{\itemsep}{0.2ex}
    \item[\namedlabel{readback:step:absprefix}{(Rb-1)}]
      Determine the \absprefixfunction~$\sabspre$ for $\altg$
      by performing a traversal over $\altg$, 
      and associate with every vertex~$\bvert$ of $\altg$ its \absprefix~$\abspre{\bvert}$.
    \item[\namedlabel{readback:step:top}{(Rb-2)}]
      Add a new vertex on top with label $\top$, arity~1, and empty abstraction prefix.
      Let $\altgacc$ be the resulting \termgraph, and $\sabspreacc$ its \absprefixfunction. 
    \item[\namedlabel{readback:step:indirection}{(Rb-3)}] 
      Introduce indirection vertices to organise sharing: 
      For every vertex $\bvert$ of $\altgacc$ with two or more incoming non-variable-backlink edges,
      add an in\-di\-rec\-tion vertex $\bverti{0}$, redirect the incoming edges of $\bvert$
      that are not variable backlinks to $\bverti{0}$,
      and direct the outgoing edge from $\bverti{0}$ to $\bvert$.
      In the resulting \termgraph\  $\altgdacc$ only indirection vertices are shared%
        \footnote{Incoming variable backlinks are not counted as sharing here.};
      their names will be used. 
      Extend $\sabspreacc$ to an \absprefixfunction~$\sabspredacc$ for $\altgdacc$
      so that every indirection vertex $\bverti{0}$ gets the prefix of its successor~$\bvert$.
    \item[\namedlabel{readback:step:spantree}{(Rb-4)}] 
      Construct a spanning tree $\aspantreedacc$ of $\altgdacc$
      by using a depth-first search (DFS) on $\altgdacc$.
      Note that all variable backlinks, and
      $\snlvarsucc$\nb-back\-links, and some of the recursive
      back-bindings, of $\altgdacc$, are not contained in $\aspantreedacc$,
      because they are back-edges of the DFS.  
    \item[\namedlabel{readback:step:synthesis}{(Rb-5)}] 
      Apply the readback synthesis rules from Fig.~\ref{fig:readback:rules}
      to $\altgdacc$ with respect to $\aspantreedacc$.
      By this a complete labelling of the edges of $\altgdacc$ by prefixed \lambdaletreccalterms\ is constructed.
      The rules define how the labelling for an incoming edge (on top) of a vertex~$\bvert$ is synthesised
      under the assumption of an already determined labelling of an outgoing edge of (and below) $\bvert$.
      If the outgoing edge in the rule does not carry a label, then the labelling of the incoming edge can
      happen regardless.
      Note that in these rules: 
      \begin{itemize}\setlength{\itemsep}{-0.1ex}
        \item 
          full line (dotted line) edges indicate spanning tree (non-spanning tree) edges,
          broken line edges either of these sorts;
        \item 
          abstraction prefixes of vertices are crucial for the \snlvarvertex, 
          and the second indirection vertex rule, where the prefixes in the synthesised terms are created;  
          in the other rules the prefix of the assumed term is used;
          for indicating a correspondence between a term's and a
          vertex's abstraction prefix we denote by $\vszero{\ps}$ the word
          of vertices occurring in a term's prefix $\ps$;
       \item
         the rule for indirection vertices with incoming non-spanning tree edge
         introduces an unfinished binding $\arecvar \mathrel = {?}$ for $\arecvar$;     
         unfinished bindings are completed in the course of the~process;   
       \item
         the \slappvertex\ rule applies only if
         $\vszero{\psix{0}} = \vszero{\psix{1}}$; the operation $\spcup$ used
         in the synthesised term's prefix builds the union per prefix variable of the
         pertaining bindings; 
         if the prefixed terms $\flabs{\vec{\apre}_0}{\allteri{0}}$ and $\flabs{\vec{\apre}_1}{\allteri{1}}$ assumed in this rule contain 
         both a yet unfinished binding equation $\arecvar = {?}$ and a completed equation $\arecvar = \bllter$ 
         at a \lambdavariable~$\cvar$, 
         then the synthesised term contains the completed binding $\arecvar = \bllter$ for $\arecvar$ at $\cvar\,$;
       \item
         not depicted in Fig.~\ref{fig:readback:rules} are variants
         of $\top$- and \slabsvertices\ rules for the cases with empty $\abindgroup$: 
         then no \letbinding\ is introduced in the synthesised term,
         but the term from the \inpart\ is used.
      \end{itemize}
  \end{description}
  \begin{figure}[tb]
    \input{readback}
    \caption{\label{fig:readback:rules}
      Readback synthesis rules for computing a representing \lambdaletreccalterm\
      from a \lambdatg.
      The rules for $\top$- and $\sslabs$-vertices have variants for the case that $B$ is empty.
      For explanations, see Def.~\ref{def:readback},~\ref{readback:step:synthesis}.}
  \end{figure}
  If this process yields the label $\femptylabs{\allter}$ for the (root-)edge pointing
  to the new top vertex of $\altgdacc$, where $\allter$ is a \lambdaletrecterm, then we 
  call $\allter$ the \emph{readback of $\altg$}. 
\end{definition}

  Note that firing of the rules in step (Rb-5) of the readback process 
  proceeds in bottom-up direction in the spanning tree, starting from 
  the back-edges, with some room for parallelism concerning work in 
  different subtrees.
  Furthermore observe that on all directed edges~$e$ (spanning tree edges or back edges)
  the rule applied to derive the edge label is uniquely determined by (is tied to) the label of
  the target vertex $\avert$ of $e$, with the single exception of $\avert$ being an 
  indirection vertex. In that case one of the two indirection vertex rules
  applies, depending on whether $e$ is a spanning-tree edge or a back-edge.

\begin{proposition}\protect\label{prop:readback}\normalfont
  Let $\altg$ be a \lambdatg.
  The process described in Def.~\ref{def:readback} produces a complete edge labelling
  of the (modified) \termgraph, with label $\femptylabs{\allter}$ for the topmost edge,
  where $\allter$ is a \lambdaletrecterm. Hence it yields $\allter$ as the readback of $\altg$. 
  Thus Def.~\ref{def:readback} defines a function 
  $\sreadback \funin \classltgs \to \Ter{\txtlambdaletreccal}$,
  the \emph{readback function}.
\end{proposition}

\begin{example}{}\label{ex:readback}
  See Fig.~\ref{fig:ex:readback} for the illustration
  of the synthesis of the readback from an exemplary \lambdatg. 
  Full line edges are in the spanning tree, dotted line edges are not. 
  Note that at the top vertex, no empty \letbinding\ is created
  since the variant of the $\top$\nb-ver\-tex rule for empty binding groups is applied.  
\end{example}

The following theorem validates property~\ref{methods:properties:readback}, with $\classeagltgs$ for $\classltgs$.

\begin{theorem}{}\label{thm:readback}
  For all $\altg\in\classeagltgs$: 
    $\compfuns{(\sgraphsemC{\classltgs}}{\sreadback)}{\altg} = \graphsemC{\classltgs}{\readback{\altg}} 
                                                      \iso \atg$,
  i.e., $\sreadback$ is a right-inverse of $\sgraphsemC\classltgs$, and
  $\sgraphsemC{\classltgs}$ a left-inverse of $\sreadback$, up~to~$\siso$.                                             
  Hence $\sreadback$ is injective, and $\sgraphsemC{\classltgs}$ is surjective,
  thus $\image{\sgraphsemC{\classltgs}} = \classeagltgs$.
\end{theorem}

\begin{figure*}
\input{trans-local-2}
\caption{\label{fig:trans-local}%
         Modification of (two of) the translation rules in Fig.~\ref{fig:def:graphsem:lhotgs}
         for a variant definition of the \lambdatg\ interpretation of \lambdaletreccalterms.
         Here the translation of a \letexpression\ does not directly spawn translations
         for the binding equations, but the \inpart\ has to be translated first.}
\end{figure*}

\begin{figure*}
  \input{trans-local-1-proof-system.tex} 
\caption{\label{fig:trans-local-proof-system}
  Formulation of the local translation rules in Fig.~\ref{fig:trans-local}
  in the form of inference rules. 
  }
\end{figure*}

\begin{proof}[Sketch of the Proof]
  Graph translation steps can be linked with corresponding readback steps 
  in order to establish that the former roughly reverse the latter.
  Roughly, because e.g.\ reversing a $\lambda$\nb-read\-back step necessitates
  both a $\lambda$- and a $\txtlet$\nb-trans\-la\-tion step. 
  However, this holds only for a modification of the translation rules~$\rulestranslambdaletreccaltoltgs$
  from Fig.~\ref{fig:def:graphsem:lhotgs}, Fig.~\ref{fig:def:graphsem:ltgs}
  where the rules $\sslet$ (for \letexpressions) and $\arecvar$ (for occurrences of recursion variables) 
  are replaced by the locally-operating versions in Fig.~\ref{fig:trans-local},
  and a initiating rule:
\begin{center}
\translation
{$\top$}
{\node(empty){$\rule{2ex}{0pt}$};}
{
  \ltgnode{root}{$\top$};
  \addPos{root}{*};
  \node[below=of root,shape=transbox,draw](body){$\femptylabs{\allter}$}; 
  \draw[->](root) to (body.north);
  \draw[<-](root.north) to +(0mm,4mm); 
}
\hspace*{3ex}
(start of translation of \lambdaletrecterm~$\allter$)
\end{center}

  for creating a top vertex is added.
  Now the translation of a \letexpression\ does no longer directly spawn translations of the bindings,
  but the bindings will only be translated later once their calls have been reached during
  the translation process of the \inpart, or of the definitions of other already translated bindings.
  Note that in the $\sslet$\nb-rule in Fig.~\ref{fig:trans-local}
  function bindings are associated with the rightmost variable in the prefix,
  which corresponds to choosing $l_i = n$ in the $\sslet$\nb-rule in Fig.~\ref{fig:def:graphsem:lhotgs}.
  While such a stipulation does not guarantee the \eagscope\ translation of every term,
  it actually does so for all \lambdaletrecterms\ that are obtained by the readback

Please find 
in Fig.~\ref{fig:1:proof:thm:readback} on page~\pageref{fig:1:proof:thm:readback}
and in Fig.~\ref{fig:2:proof:thm:readback}) on page~\pageref{fig:2:proof:thm:readback}
graphical arguments for the stepwise
reversal of readback steps through translation steps.
This establishes that graph translations steps reverse readback steps, 
which is the crucial step in the proof of the theorem.
The proof uses induction on access paths, and an invariant that relates
the \eagscope\ property localised for a vertex $\avert$ 
with the applicability of the $\snlvarsucc$\nb-rule to the readback term synthesised at $\avert$.
\end{proof}  

\begin{figure*}
\vspace*{-2.5ex}  
 \input{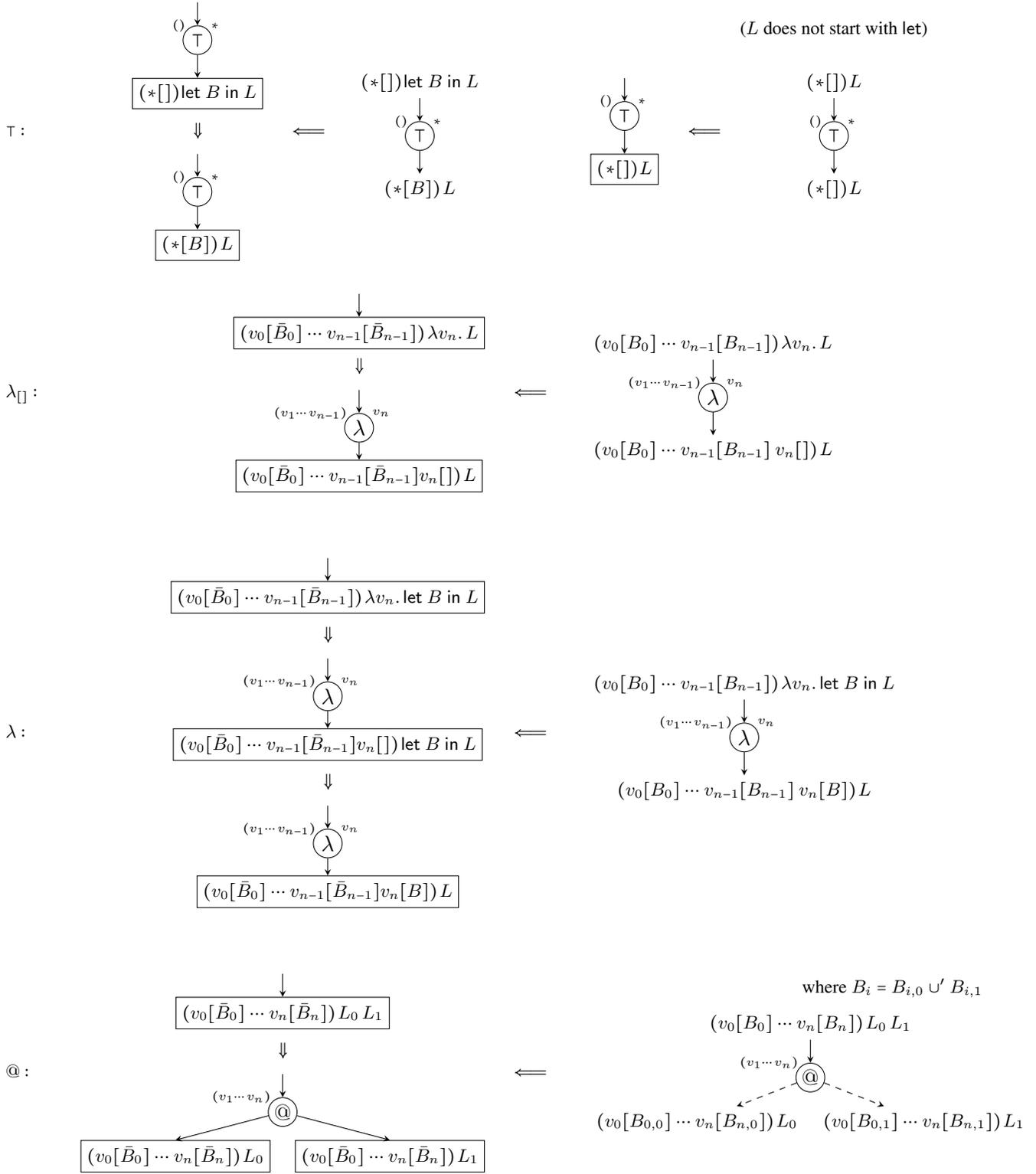}
\caption{\label{fig:1:proof:thm:readback}
 Reversal of the readback steps for top vertices, abstraction vertices, and application vertices
 by corresponding translation steps.
}
\end{figure*}

\begin{figure*}
\vspace*{-2.5ex}  
 \input{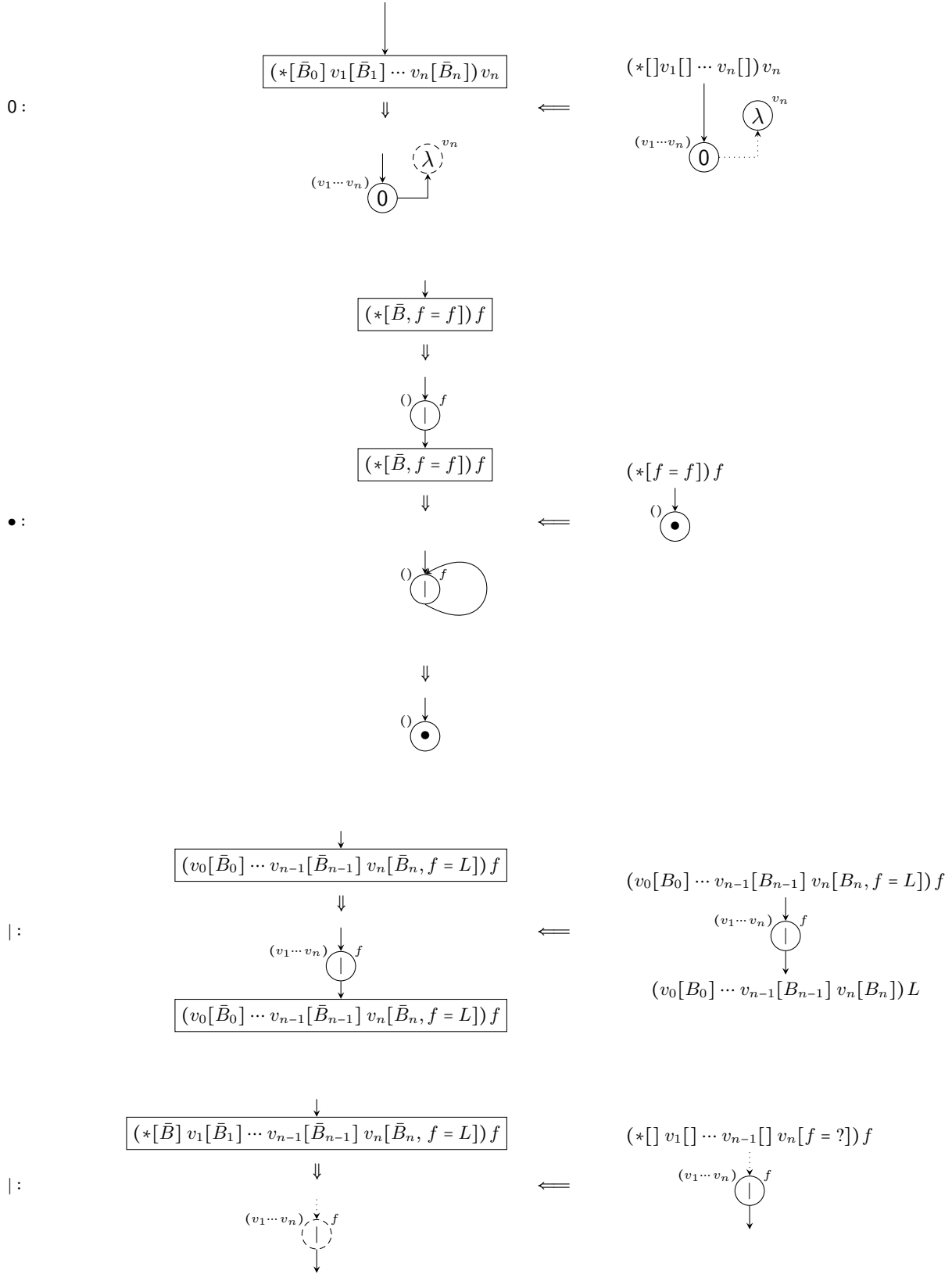}
\caption{\label{fig:2:proof:thm:readback}
 Reversal of readback steps for variable vertices, black-hole vertices, and indirection vertices
 by corresponding translation steps.
}
\end{figure*}


%

\section{Complexity analysis}
  \label{sec:complexity} 

Here we report on a complexity analysis for the individual operations
from the previous sections, 
for the used standard algorithms,
and overall, for compactification and unfolding-equivalence.

In the lemma below,
\ref{lem:complexity:item:graphsem} and \ref{lem:complexity:item:readback}
justify the property~\ref{methods:properties:efficiency} of our methods.
Items~\ref{lem:complexity:item:collapse} and \ref{lem:complexity:item:bisimilarity}
detail the complexity of standard methods when used for computing 
bisimulation collapse and bisimilarity of \lambdatgs.   
For this
note that first-order \termgraphs\  can be modelled by deterministic process graphs,
and hence by DFAs. Therefore
bisimilarity of term graphs
can be computed via language equivalence
of corresponding
   DFAs \cite{hopc:karp:1971}
(in time $\bigO{n \invAck{n}}$ \cite{nort:2009}, where $\sinvAck$ is the quasi-constant \emph{inverse Ackermann function})
and bisimulation collapse 
through state minimisation of DFAs (in time $\bigO{n \log n}$) \cite{hopc:1971}.

\begin{lemma}{}\label{lem:complexity}
  \begin{enumerate}[label=(\roman*)]\setlength{\itemsep}{0.25ex}
    \item{}\label{lem:complexity:item:length:graphsem} 
      $\graphsize{\graphsemC{\classltgs}{\allter}\!} \in\bigO{\termsize{\allter}^2}$ for $\allter\in\Ter{\txtlambdaletreccal}$. 
    \item{}\label{lem:complexity:item:graphsem} 
      Translating $\allter\in\Ter{\txtlambdaletreccal}$ into $\graphsemC{\classltgs}{\allter}\in\classltgs$ takes time $\bigO{\termsize{\allter}^2}$. 
    \item{}\label{lem:complexity:item:collapse}
      Collapsing $\altg\in\classltgs$ to $\coll{\altg}$ is in $\bigO{ \graphsize{\altg} \log \graphsize{\altg} }$.
    \item{}\label{lem:complexity:item:bisimilarity} 
      Deciding bisimilarity of $\altgi{1},\altgi{2}\in\classltgs$ 
      requires time $\bigO{n \invAck{n}}$ for $n = \max \setexp{\graphsize{\altgi{1}}, \graphsize{\altgi{2}}}$. 
      %
    \item\protect\label{lem:complexity:item:readback}
      Computing the readback 
                $\readback{\altg}$ for a given $\altg\in\classltgs$ requires time
      $\bigO{n \log n}$, for $n = \graphsize{\altg}$.
  \end{enumerate}
\end{lemma}

See Fig.~\ref{fig:ex:translation:quadratic} for an example that
the size bound in item~\ref{lem:complexity:item:length:graphsem} of the lemma is tight. 

 \begin{proposition}\normalfont
   $\graphsize{\graphsemC{\classltgs}{\allter}\!} \in\bigO{\termsize{\allter}^2}$
   for \lambdaletreccalterms~$\allter$.
 \end{proposition}
 
\begin{figure*}[t!]                         
Consider the finite \lambdaterms~$\ateri{n}$ with $n$ occurrences of bindings $\slabs{\avari{2}}$:
\begin{equation*}
  \labs{\avari{0}\avari{1}}{\lapp{\lapp{\avari{0}}{\avari{1}}}
                                 {\labs{\avari{2}}{\lapp{\lapp{\avari{0}}{\avari{1}}}
                                                  {\labs{\avari{1}}{\lapp{\lapp{\avari{0}}{\avari{2}}}
                                                                   {\labs{\avari{2}}{\lapp{\lapp{\avari{0}}{\avari{1}}}
                                                                         {\ldots
                                                                          {\labs{\avari{2}}{\lapp{\lapp{\avari{0}}{\avari{1}}}{\avari{2}}}}}}}}}}}}
\end{equation*} 
Then $\termsize{\ateri{n}}\in\bigO{n}$.                                                
But both the transformation of $\ateri{n}$ into a de-Bruijn index representation:
\begin{equation*}
  \nllabsfo{\nllabsfo{\nllappfo{\nllappfo{\nlvarsuccfo{\nlvar}}
                                         {\nlvar}}
                               {\nllabsfo{\nllappfo{\nllappfo{\nlvarsuccfoi{2}{\nlvar}}
                                                             {\nlvarsuccfo{\nlvar}}}
                                                   {\nllabsfo{\nllappfo{\nllappfo{\nlvarsuccfoi{3}{\nlvar}}
                                                                                 {\nlvarsuccfo{\nlvar}}}
                                                                       {\nllabsfo{\nllappfo{\nllappfo{\nlvarsuccfoi{4}{\nlvar}}
                                                                                                     {\nlvarsuccfo{\nlvar}}}
                                                                                 {\ldots
                                                                                  \nllabsfo{\nllappfo{\nllappfo{\nlvarsuccfoi{2n}{\nlvar}}
                                                                                                               {\nlvarsuccfo{\nlvar}}}
                                                                                           {\nlvar}}}}}}}}}}}                                                                    
\end{equation*}
and the rendering of $\ateri{n}$ with respect to the
eager scope-delimiting strategy:
\begin{equation*}
  \nllabsfo{\nllabsfo{\nllappfo{\nllappfo{\nlvarsuccfo{\nlvar}}
                                         {\nlvar}}
                               {\nllabsfo{\nllappfo{\nlvarsuccfo{\nllappfo{\nlvarsuccfoi{1}{\nlvar}}
                                                                          {\nlvar}}}
                                                   {\nllabsfo{\nllappfo{\nlvarsuccfo{\nllappfo{\nlvarsuccfoi{2}{\nlvar}}
                                                                                               {\nlvar}}}
                                                                       {\nllabsfo{\nllappfo{\nlvarsuccfo{\nllappfo{\nlvarsuccfoi{3}{\nlvar}}
                                                                                                                  {\nlvar}}}                       
                                                                                           {\ldots
                                                                                            \nllabsfo{\nllappfo{\nlvarsuccfo{\nllappfo{\nlvarsuccfoi{2n-1}{\nlvar}}
                                                                                                                                      {\nlvar}}}
                                                                                                     {\nlvar}}}}}}}}}}}                                                                   
\end{equation*}
have size $\bigO{n^2}$.
\caption{\label{fig:ex:translation:quadratic}
         Example of a sequence $\{\aiteri{n}\}_{n}$ of finite \lambdaterms\ $\ateri{n}$
         whose translation into \lambdatgs\ grows quadratically in the size of $\ateri{n}$.} 
\end{figure*}
    
Based on this lemma, and on further considerations, we obtain the following
complexity statements for our methods.

\begin{theorem}{}\label{thm:methods:complexity}
  \begin{enumerate}
      \setlength{\itemsep}{0.25ex}
        \renewcommand{\labelenumi}{(\roman{enumi})}
    \item
      The computation for a \lambdaletreccalterm~$\allter$ with $\termsize{\allter} = n$, of a maximally compactified form  
      $\funap{(\sreadback \sfuncomp \scoll \sfuncomp \sgraphsemC{\classltgs})}{\allter}$
      of a \lambdaletreccalterm~$\allter$
      requires time $\bigO{n^2 \log n}$.
      By using an \Sunsharing\ operation $\sSunsh$,
      a (typically smaller) \lambdaletrecterm\ $\compfuns{(\scompfuns{\sreadback}{\scompfuns{\sSunsh}{\scoll}}}{\sgraphsemC{\classltgs})}{\allter}$
      of size $\bigO{n \log n}$ can be obtained,
      with the same time complexity.%
    \item
      The decision of whether two \lambdaletreccalterms~$\allteri{1}$ and $\allteri{2}$ are unfolding equivalent
      requires time $\bigO{n^2 \invAck{n}}$ for $n = \max \setexp{\termsize{\allteri{1}},\termsize{\allteri{2}}}$. 
  \end{enumerate}
\end{theorem}

\section{Implementation}
  \label{sec:implementation} 
We have implemented our methods in Haskell using the
\href{http://www.cs.uu.nl/wiki/bin/view/HUT/AttributeGrammarSystem}{Utrecht
University Attribute Grammar System}.
The implementation is available at
\url{http://hackage.haskell.org/package/maxsharing/}.
Output produced for three examples from this paper, and explanations for it,
can be found in Appendix~\ref{app:impl:showcase};
for all examples in \cite{grab:roch:2014:maxsharing-arxiv}.

\section{
         Modifications, 
         extensions and 
                        applications}
  \label{sec:applications}

We have described an adaptation of the bisimulation proof method 
for 
\lambdaletrecterms.    
Recognising unfolding equivalence and
increasing sharing 
are reduced to problems involving first-order \termgraphs. 
The principal idea is to use the nested scope structure
of higher-or\-der terms
for an interpretation by \termgraphs\  with scope delimiters. 

We conclude by describing easy modifications, rather direct extensions,
and finally, promising areas of application for~our~methods.  


\subsection{Modifications}
  \label{sec:conclusion:subsec:modifications}

\myparagraph{Implicit sharing of \lambdavariables}
Multiple occurrences of the same \lambdavariable\ in a \lambdaletrecterm~$\allter$
are not shared (represented by a shared variable vertex) in the graph interpretation~$\graphsemC{\classlhotgs}{\allter}$.   
Consequently, 
our method compactifies the term $\labs{x}{\lapp{x}{x}}$ 
into $\labs{x}{\letin{f=x}{\lapp ff}}$.
Such explicit sharing of variables is excessive
for many applications. 
It can be remedied easily, namely 
by unsharing variable vertices before applying the readback,
or by preventing the readback
from introducing \letbindings\ when only a variable vertex is shared. 

\myparagraph{Avoiding aliases produced by the readback}
The readback function in Section~\ref{sec:readback}
is sensitive to the degree of sharing of $\snlvarsucc$\nb-vertices in the given \lambdatg: 
it maps two \lambdatgs\ that only differ in what concerns sharing of $\snlvarsucc$\nb-vertices
to different \lambdaletreccalterms. 
Typically, for \lambdatgs\ with maximal sharing of $\snlvarsucc$\nb-vertices
this can produce \letbindings\ that are just `aliases', such as $g$ is alias for $I$ in
$\allteracci{3}$ from Ex.~\ref{ex:translations}.
This can be avoided in two ways: by slightly adapting the readback function, or
by performing maximal unsharing of $\snlvarsucc$\nb-vertices before applying
the readback as defined.  

\myparagraph{Preventing disadvantageous sharing}
  Introducing sharing at compile-time can cause `space leaks', 
  i.e.\ a needlessly high memory footprint, at run-time, because
  `a large data structure becomes shared [\ldots], and therefore its space
  which before was reclaimed by garbage collection now cannot be reclaimed
  until its last reference is used'~\cite{sant:1995}. 
  For this reason, realisations of \CSE~\cite{chit:1997:CS-uncommon} 
  restrict the locally operating rewrite rules employed for introducing sharing 
  by suitable conditions that account for the type of potentially shared subexpressions, and their strictness in the program.  
  For our global method of introducing sharing via the bisimulation collapse, a different approach is needed.%

  Here the bisimulation collapse can be restricted so that sharing is not introduced 
  at vertices that should not be shared. 
  More precisely, it can be prevented that  
  any unshared vertex (in-degree one) from a pre-de\-ter\-mined set of `sharing-unfit' vertices 
  would have a shared vertex (in-degree greater than one) 
  as its image in the bisimulation collapse. 
  This can be achieved by modifying the graph interpretation $\graphsemC{\classltgs}{\cdot}$.
  Any set of sharing-unfit positions in $\allter$ gives rise to 
  a set of sharing-unfit vertices in $\graphsemC{\classltgs}{\allter}$.
  In the modification of $\graphsemC{\classltgs}{\allter}\,$,
  special back-links are added from every 
  sharing-unfit vertex with in-degree one to its immediate successor. 
  These back-links prevent that such a sharing-unfit vertex $\avert$ can collapse with 
  another vertex $\avertacc$ without that also the predecessors of $\avert$ and $\avertacc$ 
  would collapse as well. 

\myparagraph{A more general notion of readback}
Condition \ref{methods:properties:readback} is rather rigorous in that it
imposes sharing structure on \lambdaletreccal\ that is specific to
\lambdatgs\ (degrees of $\snlvarsucc$-sharing). For a weaker version of
\ref{methods:properties:readback} with $\sbisimsubscriptS$ in
place of isomorphism, a readback does not have to be injective,
and, independently of how much $\snlvarsucc$\nb-sharing a translation into \lambdatgs\ introduces, 
a readback function always exists.


\subsection{Extensions}
  \label{sec:conclusion:subsec:extensions}

\myparagraph{Full functional languages}  
In order to support programming languages that are based on \lambdaletreccal\
like Haskell, additional language constructs need to be supported. Such
languages can typically be desugared into a core language, which comprises
only a small subset of language constructs such as constructors, case
statements, and primitives. These constructs can be represented in an
extension of \lambdaletreccal\ by additional function symbols. 
In conjunction with a desugarer our methods are applicable to full programming
languages.

\myparagraph{
             Other programming languages, 
             and calculi with binding constructs}
  Most programming languages feature constructs for grouping definitions that are similar to \txtletrec.
  We therefore expect that our methods can be adapted to many imperative languages in particular, 
  and may turn out to be fruitful for optimising compilers.
  Our methods for achieving maximal sharing certainly generalise to
  theoretical frameworks, and calculi with binding constructs,
  such as the \picalculus\ \cite{miln:1999}, and higher-order rewrite systems 
  (e.g.\ \CRSs\ and \HRSs, \cite{terese:2003}) as used here for the formalisation of \lambdaletreccal.

\myparagraph{Fully-lazy \lambdalifting}
    There is a close connection between our methods and fully-lazy \lambdalifting~\cite{hugh:1982,peyt:jone:1987}. 
    In particular, the required-variable and scope analysis of a \lambdaletrecterm~$\allter$ 
    on which the \lambdatg-trans\-la\-tion~$\graphsemC{\classltgs}{\allter}$ is based 
    is closely analogous 
    to the one needed for extracting from $\allter$ the supercombinators in the result $\Hat{\allter}$ of fully-lazy \lambdalifting\ $\allter$.
    Moreover, the fully-lazy \lambdalifting\ transformation can even be implemented
    in a natural way on the basis of our methods.
    Namely as the composition $\scompfuns{\sreadbackC{\text{\it LL}}}{\graphsemC{\classltgs}{\cdot}}$ 
    of the translation $\graphsemC{\classltgs}{\cdot}$ into \lambdatgs,
    where $\sreadbackC{\text{\it LL}}$ is a variant readback function that,
    for a given \lambdatg, 
    synthesises the system $\Hat{\allter}$ of supercombinators,
    instead of the \lambdaletrecterm~$\readback{\allter}$.%

\myparagraph{Maximal sharing on supercombinator translations of \lambdaletreccalterms}
  \lambdaletrecterms~$\allter$ correspond to supercombinator systems $\Hat{\allter}$, 
                                                                     the result of fully-lazy \lambdalifting\ $\allter$:
  the combinators in $\Hat{\allter}$ 
                     correspond to `extended scopes' \cite{grab:roch:2012} (or `skeletons' \cite{bala:2012})
  in $\allter$, and supercombinator reduction steps on $\Hat{\allter}$ 
                                                                       correspond to weak \betareduction\ steps $\allter$.
  In the case of \lambdacalculus\ this has been established by Balabonski \cite{bala:2012}.
  Via this correspondence
  the maximal-sharing method for \lambdaletrecterms\ can be lifted to obtain
  a maximal-sharing method systems of supercombinators obtained by fully-lazy \lambdalifting. 

\myparagraph{\mbox{Non-eager scope-closure strategies}}
We focused on eager-scope trans\-la\-tions, because they facilitate maximal sharing,
and guarantee that interpretations of unfolding-equivalent \lambdaletrecterms\ are bisimilar. 
Yet every scope-clo\-sure strategy \cite{grab:roch:2012} induces a translation
and its own notion of maximal sharing.
For adapting our maximal sharing method it is 
necessary to modify the translation into first-order term graphs in such a way
that the image class obtained is closed under homomorphism         
($\classltgs$ is not closed under $\sfunbisim$, unlike its subclass $\classeagltgs$).
This can be achieved by using delimiter vertices also below variable vertices 
to close scopes that are still open~\cite[report]{grab:roch:2013:a:TERMGRAPH}.%

\myparagraph{Weaker notions of sharing} 
The presented methods deal with sharing as expressed by \txtletrec\ 
that is horizontal, vertical, or twisted \cite{blom:2001}. 
By contrast, the construct $\mu$ \cite{blom:2001,grab:roch:2013:c:RTA} expresses only vertical, and the
non-recursive \txtlet\ only horizontal, sharing. By restricting
bisimulation, our methods can be adapted to the \lambdacalculus\ with $\mu$, or with $\txtlet$.

\myparagraph{Nested term graphs}
  The nested scope structure of a \lambdaletrecterm\ 
  can also be represented by a nested structure of term graphs. 
  The representation of a \lambdaletrecterm\ as a `nested term graph'~\cite{grab:oost:2014}
  starts with an ordinary term graph
  in which some of the vertices are labelled by `nested' symbols
  that designate outermost bindings together with their scope.
  Any such vertex is additionally associated with a usual term graph 
  that specifies the subterm context describing the scope, 
  where any inner scopes are again expressed by nested symbols.
  The association between nested symbols and their term graph specifications
  is required to be tree-like.
  The implementation result developed here can be generalised to show 
  that nested term graphs can be implemented faithfully by first-order term graphs~\cite{grab:oost:2014}.%

\subsection{Applications}
  \label{sec:conclusion:subsec:applications}

\myparagraph{Maximal sharing at run-time}
Maximal sharing can be applied repeatedly at run-time 
in order to regain a maximally shared form,
thereby speeding up evaluation.
This is reminiscent of `collapsed tree rewriting' \cite{plum:1993}
for evaluating first-order \termgraphs\ represented as maximally shared dags. 
Since the state of a program in the memory at run-time is 
typically represented as a supercombinator graph, compactification by bisimulation collapse
can take place directly on that graph (see Sec.~\ref{sec:conclusion:subsec:extensions}), no translation is needed. 
Compactification can be coupled with garbage collection as bisimulation
collapse subsumes some of the work required for a mark and sweep garbage collector.
However, a compromise needs to be found between
the costs for the optimisation and the gained efficiency.

\myparagraph{Compile-time optimisation phase}
Increasing sharing facilitates potential gains in efficiency. Our method generalises
common subexpression elimination, but therefore it also inherits its
shortcomings: the cost of sharing (e.g.\ of very small functions) might exceed
the gain. In non-strict functional languages, sharing can cause `memory
leaks' \cite{chit:1997:CS-uncommon}. Therefore, similar as for \CSE, additional
dynamic analyses like binding-time analysis \cite{pals:schw:1994}, and
heuristics to restrict sharing in cases when it is disadvantageous
\cite{peyt:jone:1987,gold:huda:1987} are advisable.

\myparagraph{Additional prevention of disadvantageous sharing}
While static analysis methods for preventing sharing that may be disadvantageous at run-time
can be adapted from \CSE\ to the maximal-sharing method (see Sec.~\ref{sec:conclusion:subsec:modifications}),
this has yet to be investigated for binding-time analysis \cite{pals:schw:1994}
and a sharing analysis of partial applications~\cite{gold:huda:1987}.%

\myparagraph{Code improvement} 
In programming it is generally desirable to avoid duplication of code. 
As extension of \CSE, our method is able to detect code duplication. 
The bisimulation collapse of the \termgraph\  interpretation of a program can,
together with the readback, provide guidance on how code can be
refactored into a more compact form.
This application requires some fine-tuning to avoid excessive
behaviour like the explicit sharing of variable occurrences (see Sec.~\ref{sec:conclusion:subsec:modifications}). 
Yet for this only lightweight additional machinery is needed, such as
size constraints or annotations to restrict the bisimulation collapse.

\myparagraph{Function equivalence}  
Recognising whether two programs 
implement the same function is undecidable. Still, this problem is tackled by
proof assistants, and by automated theorem provers used in type-checkers of
compilers for dependently-typed programming languages such as Agda. For such
systems co-inductive proofs are more difficult to find than
inductive ones, and require more effort by the user. Our method for
deciding unfolding-equivalence could help to develop new approaches to finding
co-inductive proofs.

\vspace{0.5ex}
\paragraph{Acknowledgment}
    We want to thank Vincent van Oostrom for extensive feedback on a draft,
    Doaitse Swierstra and Dimitri Hendriks for helpful comments,
    and Jeroen Keiren for a suggestion concerning restricting the bisimulation collapse.
    We also thank the anonymous reviewers for their comments,
    and a number of stimulating questions.

\bibliographystyle{abbrvnat}
\bibliography{maxsharing}

\appendix

\onecolumn
\section{Example for the translation \lambdaletrecterms\ into \lambdahotgs\ and into \lambdatgs}\label{app:translation}
\begin{multicols}{2}
  For two terms from the paper 
  we provide the stepwise translation
  of \lambdaletrecterms\ into \lambdahotgs. 
  
  We start off by a simple example, namely the translation
  of the term $\bllter$ from Example~\ref{ex:fix} on page \pageref{ex:fix}. 
  There is no application of the $\snlvarsucc$-rule, thus it yields the same sequence of
  graphs regardless of whether the rules $\rulestranslambdaletreccaltolhotgs$
  from Fig.~\ref{fig:def:graphsem:lhotgs} are used,
  by which a \lambdahotg\ is produced,
  or the modified rules
  $\rulestranslambdaletreccaltoltgs$ (the result of dropping the
  $\snlvarsucc$\nb-rule from $\rulestranslambdaletreccaltolhotgs$, and using
  the $\snlvarsucc$\nb-rule in Fig.~\ref{fig:def:graphsem:ltgs} instead),
  by which \lambdatgs\ with an \absprefixfunction\ are produced.
\end{multicols}

\vspace{1cm}

\transinit{
	\node[draw,shape=transbox](t){$\femptylabs{\labs{f}{\letin{r=\lapp{f}{(\lapp fr)}}{r}}}$};
	\draw[<-](t.north) to +(0mm,6mm);
}
\transstep{\sslabs\hspace{6mm}}{
	\node[draw,shape=transbox](t){$\flabs{*[]\;f^v[]}{\letin{r=\lapp{f}{(\lapp fr)}}{r}}$};
	\ltgnode[above=of t.north]{f}{$\sslabs$}; \addPos{f}{v}; \addPrefix{f}{};
	\draw[<-](f.north) to +(0mm,4mm);
	\draw[->](f) to (t.north);
}
\transstep{\sslet\hspace{6mm}}{
	\node[draw,shape=transbox](in){$\flabs{*[]\;f^v[r^w=\lapp{f}{(\lapp fr)}]}{r}$};
	\ltgnode[above=of in.north]{f}{$\sslabs$}; \addPos{f}{v}; \addPrefix{f}{};
	\draw[->](f) to (in.north);
	\draw[<-](f.north) to +(0mm,4mm);
	\node[node distance=2mm,draw,shape=transbox,left=of in](ffr){$\flabs{*[]\;f^v[r^w=\lapp{f}{(\lapp fr)}]}{\lapp{f}{(\lapp fr)}}$};
	\ltgnode[above=of ffr]{indir}{\indir}; \addPos{indir}{w};
	\draw[->](indir) to (ffr.north);
}
\transbreak
\transstep{\sslapp,\arecvar}{
	\ltgnode{f}{$\sslabs$}; \addPos{f}{v}; \addPrefix{f}{};
	\draw[<-](f.north) to +(0mm,4mm);
	\node[below=of f](in){};
	\node[left=of in](inwest){};
	\draw(f) |- (inwest.center);
	\node[left=of f](indireast){};
	\ltgnode[left=of indireast]{indir}{\indir}; \addPos{indir}{w};
	\draw[->](indireast.center) to (indir.east);
	\draw(indireast.center) -| (inwest.center);
	\ltgnode[below=of indir]{f@fr}{$\sslapp$}; \addPrefix{f@fr}{v};
	\draw[->](indir) to (f@fr);
	\node[below=of f@fr](ffr){};
	\node[node distance=0mm, left=of ffr,draw,shape=transbox](l){$\flabs{*[]\;f^v[r^w=\lapp{f}{(\lapp fr)}]}{f}$};
	\node[node distance=0mm,right=of ffr,draw,shape=transbox](ffr){$\flabs{*[]\;f^v[r^w=\lapp{f}{(\lapp fr)}]}{\lapp fr}$};
	\draw[->](f@fr) to (l.north);
	\draw[->](f@fr) to (ffr.north);
}
\transstep{\snlvar,\sslapp}{
	\ltgnode{f}{$\sslabs$}; \addPos{f}{v}; \addPrefix{f}{};
	\draw[<-](f.north) to +(0mm,4mm);
	\node[below=of f](in){};
	\ltgnode[below=of f]{indir}{\indir}; \addPos{indir}{w};
	\draw[->](f) to (indir.north);
	\ltgnode[below=of indir]{f@fr}{$\sslapp$}; \addPrefix{f@fr}{v};
	\draw[->](indir) to (f@fr);
	\node[below=of f@fr](ffr){};
	\ltgnode[left=of ffr]{l}{$\snlvar$}; \addPrefix{l}{v};
	\draw[->](f@fr) to (l);
	\node[left=of l](lwest){}; \draw(l) to (lwest.center); \draw[->](lwest.center) |- (f);
	\ltgnode[right=of ffr]{f@r}{$\sslapp$}; \addPrefix{f@r}{v};
	\node[below=of f@r](fr){};
	\draw[->](f@fr) to (f@r.north);
	\node[node distance=0mm, left=of fr,draw,shape=transbox](f1){$\flabs{*[]\;f^v[r^w=\lapp{f}{(\lapp fr)}]}{f}$};
	\draw[->](f@r) to (f1.north);
	\node[node distance=0mm,right=of fr,draw,shape=transbox](r){$\flabs{*[]\;f^v[r^w=\lapp{f}{(\lapp fr)}]}{r}$};
	\draw[->](f@r) to (r.north);
}
\transbreak
\transstep{\snlvar,\arecvar}{
	\ltgnode{f}{$\sslabs$}; \addPos{f}{v}; \addPrefix{f}{};
	\draw[<-](f.north) to +(0mm,4mm);
	\node[below=of f](in){};
	\ltgnode[below=of f]{indir}{\indir}; \addPos{indir}{w};
	\draw[->](f) to (indir.north);
	\ltgnode[below=of indir]{f@fr}{$\sslapp$}; \addPrefix{f@fr}{v};
	\draw[->](indir) to (f@fr);
	\node[below=of f@fr](ffr){};
	\ltgnode[left=of ffr]{l}{$\snlvar$}; \addPrefix{l}{v};
	\draw[->](f@fr) to (l);
	\node[left=of l](lwest){}; \draw(l) to (lwest.center); \draw[->](lwest.center) |- (f);
	\ltgnode[right=of ffr]{f@r}{$\sslapp$}; \addPrefix{f@r}{v};
	\node[below=of f@r](fr){};
	\draw[->](f@fr) to (f@r.north);
	\ltgnode[left=of fr]{f1}{$\snlvar$}; \addPrefix{f1}{v};
	\draw[->](f@r) to (f1);
	\draw(f1) -| (lwest.center);
	\node[right=of fr](r){};
	\draw(f@r) to (r.center);
	\node[node distance=1mm,right=of r](reast){};
	\draw(r.center) to (reast.center);
	\draw[->](reast.center) |- (indir.east);
}
\transstep{\text{eliminate indirection vertices}}{
	\ltgnode{f}{$\sslabs$}; \addPos{f}{v}; \addPrefix{f}{};
	\draw[<-](f.north) to +(0mm,4mm);
	\node[below=of f](in){};
	\ltgnode[below=of f]{f@fr}{$\sslapp$}; \addPrefix{f@fr}{v};
	\draw[->](f) to (f@fr);
	\node[below=of f@fr](ffr){};
	\ltgnode[left=of ffr]{l}{$\snlvar$}; \addPrefix{l}{v};
	\draw[->](f@fr) to (l);
	\node[left=of l](lwest){}; \draw(l) to (lwest.center); \draw[->](lwest.center) |- (f);
	\ltgnode[right=of ffr]{f@r}{$\sslapp$}; \addPrefix{f@r}{v};
	\node[below=of f@r](fr){};
	\draw[->](f@fr) to (f@r.north);
	\ltgnode[left=of fr]{f1}{$\snlvar$}; \addPrefix{f1}{v};
	\draw[->](f@r) to (f1);
	\draw(f1) -| (lwest.center);
	\node[right=of fr](r){};
	\draw(f@r) to (r.center);
	\node[node distance=1mm,right=of r](reast){};
	\draw(r.center) to (reast.center);
	\draw[->](reast.center) |- (f@fr.east);
}

\newpage

\begin{multicols}{2}
  We continue with the term  $\allteri{2}$ from Example~\ref{ex:translations} on
  page~\pageref{ex:translations}. We translate it in two different ways
  correspoding to the first-order term graph semantics
  $\sgraphsemCmin{\classltgs}$ from Def.~\ref{def:graphsem-min:ltgs} and
  $\sgraphsemC{\classltgs}$ from Def.~\ref{def:graphsem:ltgs}, respectively.
  Both sequences of translation steps yield
  the same \lambdahotg~$\graphsemC{\classlhotgs}{\allteri{2}}$, but obtain
  the different \lambdatgs\ $\graphsemCmin{\classltgs}{\allteri{2}}$ and
  $\graphsemC{\classltgs}{\allteri{2}}$ in Fig.~\ref{fig:translations}.
  
  
  These two stepwise translation processes exemplify, on the one hand (when
  ignoring the dotted $\snlvarsucc$\nb-vertices that do not occur in
  \lambdahotgs) the translation into \lambdahotgs\ according to the rules
  $\rulestranslambdaletreccaltolhotgs$ from Fig.~\ref{fig:def:graphsem:lhotgs},
  and on the other hand (now taking the dotted delimiter $\snlvarsucc$\nb-vertices into account)
  the translation into \lambdatgs\ according to the modified rules
  $\rulestranslambdaletreccaltoltgs$.
  
  In each step, one or more translation rules (whose
  names are indicated as subscripts in the steps) are applied to the
  translation boxes in the graph.
  When no more rules are applicable, indirection vertices are erased.
  Both translations are \eagscope\
  (i.e. applications of $\snlvarsucc$\nb-rules are given priority, 
  and prefixes lengths are chosen too be small enough in the $\sslet$-rule, see Sec.~\ref{sec:lhotgs:subsec:eagscope})
  but differ in how they resolve the non-determinism due to different choices
  for the prefix lengths $l_1,\dots,l_k$ in the $\sslet$-rule. 

  First we consider the translations of $\allteri{2}$ with
  $\sgraphsemC{\classlhotgs}$ and $\sgraphsemCmin{\classltgs}$, i.e.\ the
  translation process in which prefix lengths are chosen minimally when
  applying the $\sslet$-rule.
\end{multicols}

\vspace{2cm}

\newcommand\vx{u}
\newcommand\vy{v}
\newcommand\vz{w}
\newcommand\vI{s}
\newcommand\vf{t}

\transinit{
	\node[draw,shape=transbox](t){$\femptylabs{\labs{x}{\labs{y}{\letin{I=\labs{z}{z},f=x}{\lapp{\lapp{\lapp{y}{I}}{(\lapp{I}{y})}}{(\lapp{f}{f})}}}}}$};
	\draw[<-](t.north) to +(0mm,6mm);
}
\transbreak
\transstep{\sslabs}{
	\node[draw,shape=transbox](t){$\flabs{*[]\;x^\vx[]}{\labs{y}{\letin{I=\labs{z}{z},f=x}{\lapp{\lapp{\lapp{y}{I}}{(\lapp{I}{y})}}{(\lapp{f}{f})}}}}$};
	\ltgnode[above=of t.north]{x}{$\sslabs$}; \addPos{x}{\vx}; \addPrefix{x}{};
	\draw[<-](x.north) to +(0mm,4mm);
	\draw[->](x) to (t.north);
}
\transbreak
\transstep{\sslabs}{
	\node[draw,shape=transbox](t){$\flabs{*[]\;x^\vx[]\;y^\vy[]}{\letin{I=\labs{z}{z},f=x}{\lapp{\lapp{\lapp{y}{I}}{(\lapp{I}{y})}}{(\lapp{f}{f})}}}$};
	\ltgnode[above=of t.north]{y}{$\sslabs$}; \addPos{y}{\vy}; \addPrefix{y}{\vx};
	\draw[->](y) to (t.north);
	\ltgnode[above=of y.north]{x}{$\sslabs$}; \addPos{x}{\vx}; \addPrefix{x}{};
	\draw[->](x) to (y.north);
	\draw[<-](x.north) to +(0mm,4mm);
}
\transbreak
\transstep{\sslet}{
	\node[draw,shape=transbox](yIIyff){$\flabs{*[I^\vI=\labs{z}{z}]\;x^\vx[f^\vf=x]\;y^\vy[ ]}{\lapp{\lapp{\lapp{y}{I}}{(\lapp{I}{y})}}{(\lapp{f}{f})}}$};
	\ltgnode[above=of yIIyff.north]{y}{$\sslabs$}; \addPos{y}{\vy}; \addPrefix{y}{\vy};
	\draw[->](y) to (yIIyff.north);
	\ltgnode[above=of y.north]{x}{$\sslabs$}; \addPos{x}{\vx}; \addPrefix{x}{};
	\draw[->](x) to (y.north);
	\draw[<-](x.north) to +(0mm,4mm);
	\ltgnode[left=of yIIyff]{Iindir}{\indir}; \addPos{Iindir}{\vI};
	\node[draw,shape=transbox,below=of Iindir](I){$\flabs{*[I^\vI={\labs{z}{z}}]}{\labs{z}{z}}$};
	\draw[->](Iindir) to (I.north);
	\ltgnode[right=of yIIyff]{findir}{\indir}; \addPos{findir}{\vf};
	\node[draw,shape=transbox,below=of findir](f){$\flabs{*[I^\vI={\labs{z}{z}}]\;x^\vx[f^\vf=x]}{x}$};
	\draw[->](findir) to (f.north);
}
\transbreak
\transstep{\sslabs,\sslapp,\snlvar}{
	\node(yIIyff){};
	\node[node distance=0mm,draw,shape=transbox,left=of yIIyff](yIIy){$\flabs{*[I^\vI=\labs{z}{z}]\;x^\vx[f^\vf=x]\;y^\vy[ ]}{\lapp{\lapp{y}{I}}{(\lapp{I}{y})}}$};
	\node[node distance=0mm,draw,shape=transbox,right=of yIIyff](ff){$\flabs{*[I^\vI=\labs{z}{z}]\;x^\vx[f^\vf=x]\;y^\vy[ ]}{\lapp{f}{f}}$};
	\ltgnode[above=of yIIyff.north]{yIIy@ff}{$\sslapp$}; \addPrefix{yIIy@ff}{\vx\;\vy};
	\draw[->](yIIy@ff) to (yIIy.north); \draw[->](yIIy@ff) to (ff.north);
	\ltgnode[above=of yIIy@ff.north]{y}{$\sslabs$}; \addPos{y}{\vy}; \addPrefix{y}{\vx};
	\draw[->](y) to (yIIy@ff.north);
	\ltgnode[above=of y.north]{x}{$\sslabs$}; \addPos{x}{\vx}; \addPrefix{x}{};
	\draw[->](x) to (y.north);
	\draw[<-](x.north) to +(0mm,4mm);
	\ltgnode[node distance=8mm,left=of yIIy]{z}{$\sslabs$}; \addPos{z}{\vz}; \addPrefix{z}{};
	\ltgnode[above=of z]{Iindir}{\indir}; \addPos{Iindir}{\vI};
	\draw[->](Iindir) to (z.north);
	\node[draw,shape=transbox,below=of z](I){$\flabs{*[I^\vI={\labs{z}{z}}]\;z^\vz[]}{z}$};
	\draw[->](z) to (I.north);
	\ltgnode[node distance=8mm,right=of ff]{f}{$\snlvar$}; \addPrefix{f}{\vx};
	\ltgnode[above=of f]{findir}{\indir}; \addPos{findir}{\vf};
	\draw[->](findir) to (f.north);
	\node[right=of f](feast){}; \draw(f.east) to (feast.center); \draw[->] (feast.center) |- (x.east);
}
\transbreak
\raggedbottom
\transstep{\snlvar, \sslapp, \sslapp}{
	\node(yIIyff){};
	\ltgnode[node distance=2cm,left=of yIIyff]{yI@Iy}{$\sslapp$}; \addPrefix{yI@Iy}{\vx\;\vy}; \node[node distance=8mm,below=of yI@Iy](yIIy){};
	\node[node distance=0mm,draw,shape=transbox, left=of yIIy](yI){$\flabs{*[I^\vI=\labs{z}{z}]\;x^\vx[f^\vf=x]\;y^\vy[ ]}{\lapp{y}{I}}$};
	\draw[->](yI@Iy) to (yI.north);
	\node[node distance=0mm,draw,shape=transbox,right=of yIIy](Iy){$\flabs{*[I^\vI=\labs{z}{z}]\;x^\vx[f^\vf=x]\;y^\vy[ ]}{\lapp{I}{y}}$};
	\draw[->](yI@Iy) to (Iy.north);
	\ltgnode[dotted,node distance=2cm,right=of Iy]{Sff}{$\snlvarsucc$}; \addPrefix{Sff}{\vx\;\vy};
	\node[draw,shape=transbox,below=of Sff](ff){$\flabs{*[I^\vI=\labs{z}{z}]\;x^\vx[f^\vf=x]}{\lapp{f}{f}}$};
	\draw[->](Sff) to (ff.north);
	\ltgnode[above=of yIIyff.north]{yIIy@ff}{$\sslapp$}; \addPrefix{yIIy@ff}{\vx\;\vy};
	\draw[->](yIIy@ff) to (yI@Iy.north); \draw[->](yIIy@ff) to (Sff.north);
	\ltgnode[above=of yIIy@ff.north]{y}{$\sslabs$}; \addPos{y}{\vy}; \addPrefix{y}{\vx};
	\draw[->](y) to (yIIy@ff.north);
	\ltgnode[above=of y.north]{x}{$\sslabs$}; \addPos{x}{\vx}; \addPrefix{x}{};
	\draw[->](x) to (y.north);
	\draw[<-](x.north) to +(0mm,4mm);
	\ltgnode[node distance=7cm,left=of y]{Iindir}{\indir}; \addPos{Iindir}{\vI};
	\ltgnode[below=of Iindir]{z}{$\sslabs$}; \addPos{z}{\vz}; \addPrefix{z}{};
	\draw[->](Iindir) to (z.north);
	\ltgnode[below=of z]{I}{$\snlvar$}; \addPrefix{I}{\vz};
	\node[left=of I](Iwest){}; \draw(I.west) to (Iwest.center); \draw[->] (Iwest.center) |- (z.west);
	\draw[->](z) to (I.north);
	\ltgnode[node distance=5cm,right=of y]{findir}{\indir}; \addPos{findir}{\vf};
	\ltgnode[below=of findir]{f}{$\snlvar$}; \addPrefix{f}{\vx};
	\draw[->](findir) to (f.north);
	\node[right=of f](feast){}; \draw(f.east) to (feast.center); \draw[->] (feast.center) |- (x.east);
}
\transbreak
\transstep{\sslapp,\sslapp,\sslapp}{
\hspace{-5mm}
	\node(yIIyff){};
	\ltgnode[node distance=3cm,left=of yIIyff]{yI@Iy}{$\sslapp$}; \addPrefix{yI@Iy}{\vx\;\vy}; \node[node distance=8mm,below=of yI@Iy](yIIy){};
	\ltgnode[node distance=1cm,right=of yIIy]{I@y}{$\sslapp$}; \addPrefix{I@y}{\vx\;\vy}; \node[below=of I@y](Iy){};
	\draw[->](yI@Iy) to (I@y.north);
	\node[node distance=0mm,draw,shape=transbox, left=of Iy](I2){$\flabs{*[I^\vI=\labs{z}{z}]\;x^\vx[f^\vf=x]\;y^\vy[ ]}{I}$};
	\draw[->](I@y) to (I2.north);
	\node[node distance=0mm,draw,shape=transbox,right=of Iy](y2){$\flabs{*[I^\vI=\labs{z}{z}]\;x^\vx[f^\vf=x]\;y^\vy[ ]}{y}$};
	\draw[->](I@y) to (y2.north);
	\ltgnode[node distance=5mm, left=of I2]{y@I}{$\sslapp$}; \addPrefix{y@I}{\vx\;\vy}; \node[node distance=1cm,below=of y@I](yI){};
	\draw[->](yI@Iy) to (y@I.north);
	\node[node distance=0mm,draw,shape=transbox, left=of yI](y1){$\flabs{*[I^\vI=\labs{z}{z}]\;x^\vx[f^\vf=x]\;y^\vy[ ]}{y}$};
	\draw[->](y@I) to (y1.north);
	\node[node distance=0mm,draw,shape=transbox, right=of yI](I1){$\flabs{*[I^\vI=\labs{z}{z}]\;x^\vx[f^\vf=x]\;y^\vy[ ]}{I}$};
	\draw[->](y@I) to (I1.north);
	\ltgnode[dotted,node distance=4.5cm,right=of I@y]{Sff}{$\snlvarsucc$}; \addPrefix{Sff}{\vx\;\vy};
	\ltgnode[node distance=1cm,below=of Sff]{f@f}{$\sslapp$}; \addPrefix{f@f}{\vx\;\vy}; \node[below=of f@f](ff){};
	\draw[->](Sff) to (f@f.north);
	\node[node distance=0mm,draw,shape=transbox, left=of ff](f1){$\flabs{*[I^\vI=\labs{z}{z}]\;x^\vx[f^\vf=x]}{f}$};
	\draw[->](f@f) to (f1.north);
	\node[node distance=0mm,draw,shape=transbox,right=of ff](f2){$\flabs{*[I^\vI=\labs{z}{z}]\;x^\vx[f^\vf=x]}{f}$};
	\draw[->](f@f) to (f2.north);
	\ltgnode[above=of yIIyff.north]{yIIy@ff}{$\sslapp$}; \addPrefix{yIIy@ff}{\vx\;\vy};
	\draw[->](yIIy@ff) to (yI@Iy.north); \draw[->](yIIy@ff) to (Sff.north);
	\ltgnode[above=of yIIy@ff.north]{y}{$\sslabs$}; \addPos{y}{\vy}; \addPrefix{y}{\vx};
	\draw[->](y) to (yIIy@ff.north);
	\ltgnode[above=of y.north]{x}{$\sslabs$}; \addPos{x}{\vx}; \addPrefix{x}{};
	\draw[->](x) to (y.north);
	\draw[<-](x.north) to +(0mm,4mm);
	\ltgnode[node distance=6cm,left=of y]{Iindir}{\indir}; \addPos{Iindir}{\vI};
	\ltgnode[below=of Iindir]{z}{$\sslabs$}; \addPos{z}{\vz}; \addPrefix{z}{};
	\draw[->](Iindir) to (z.north);
	\ltgnode[below=of z]{I}{$\snlvar$}; \addPrefix{I}{\vz};
	\node[left=of I](Iwest){}; \draw(I.west) to (Iwest.center); \draw[->] (Iwest.center) |- (z.west);
	\draw[->](z) to (I.north);
	\ltgnode[node distance=4cm,right=of y]{findir}{\indir}; \addPos{findir}{\vf};
	\ltgnode[below=of findir]{f}{$\snlvar$}; \addPrefix{f}{\vx};
	\draw[->](findir) to (f.north);
	\node[right=of f](feast){}; \draw(f.east) to (feast.center); \draw[->] (feast.center) |- (x.east);
}
\transbreak
\transstep{\snlvar, \snlvarsucc, \snlvarsucc, \snlvar, \arecvar, \arecvar}{
	\node(yIIyff){};
	\ltgnode[node distance=1.5cm,left=of yIIyff]{yI@Iy}{$\sslapp$}; \addPrefix{yI@Iy}{\vx\;\vy}; \node[below=of yI@Iy](yIIy){};
	\ltgnode[node distance=2cm,right=of yIIy]{I@y}{$\sslapp$}; \addPrefix{I@y}{\vx\;\vy}; \node[below=of I@y](Iy){};
	\draw[->](yI@Iy) to (I@y.north);
	\ltgnode[dotted,left=of Iy]{SI2}{$\snlvarsucc$}; \addPrefix{SI2}{\vx\;\vy};
	\draw[->](I@y) to (SI2.north);
	\node[draw,shape=transbox, below=of SI2](I2){$\flabs{*[I^\vI=\labs{z}{z}]\;x^\vx[f^\vf=x]}{I}$};
	\draw[->](SI2) to (I2.north);
	\ltgnode[right=of Iy]{y2}{$\snlvar$}; \addPrefix{y2}{\vx\;\vy};
	\node[right=of y2](y2l){}; \draw(y2) to (y2l.center);
	\draw[->](I@y) to (y2.north);
	\ltgnode[node distance=3cm, left=of SI2]{y@I}{$\sslapp$}; \addPrefix{y@I}{\vx\;\vy}; \node[below=of y@I](yI){};
	\draw[->](yI@Iy) to (y@I.north);
	\ltgnode[left=of yI]{y1}{$\snlvar$}; \addPrefix{y1}{\vx\;\vy};
	\node[left=of y1](y1l){}; \draw(y1) to (y1l.center);
	\draw[->](y@I) to (y1.north);
	\ltgnode[dotted,right=of yI]{SI1}{$\snlvarsucc$}; \addPrefix{SI1}{\vx\;\vy};
	\draw[->](y@I) to (SI1.north);
	\node[draw,shape=transbox, below=of SI1](I1){$\flabs{*[I^\vI=\labs{z}{z}]\;x^\vx[f^\vf=x]}{I}$};
	\draw[->](SI1) to (I1.north);
	\ltgnode[dotted,node distance=2cm,right=of I@y]{Sff}{$\snlvarsucc$}; \addPrefix{Sff}{\vx\;\vy};
	\ltgnode[below=of Sff]{f@f}{$\sslapp$}; \addPrefix{f@f}{\vx\;\vy}; \node[below=of f@f](ff){};
	\draw[->](Sff) to (f@f.north);
	\node[node distance=1.5cm,right=of ff](fmerge){};
	\node[ left=of ff](f1){}; \draw(f@f) to (f1.south); \draw(f1.south) |- (fmerge.south);
	\node[right=of ff](f2){}; \draw(f@f) to (f2.center); \draw(f2.center) to (fmerge.center);
	\ltgnode[above=of yIIyff.north]{yIIy@ff}{$\sslapp$}; \addPrefix{yIIy@ff}{\vx\;\vy};
	\draw[->](yIIy@ff) to (yI@Iy.north); \draw[->](yIIy@ff) to (Sff.north);
	\ltgnode[above=of yIIy@ff.north]{y}{$\sslabs$}; \addPos{y}{\vy}; \addPrefix{y}{\vx};
	\draw[->](y2l.center) |- (y.east); \draw[->](y1l.center) |- (y.west);
	\draw[->](y) to (yIIy@ff.north);
	\ltgnode[above=of y.north]{x}{$\sslabs$}; \addPos{x}{\vx}; \addPrefix{x}{};
	\draw[->](x) to (y.north);
	\draw[<-](x.north) to +(0mm,4mm);
	\ltgnode[node distance=6cm,left=of y]{Iindir}{\indir}; \addPos{Iindir}{\vI};
	\ltgnode[below=of Iindir]{z}{$\sslabs$}; \addPos{z}{\vz}; \addPrefix{z}{};
	\draw[->](Iindir) to (z.north);
	\ltgnode[below=of z]{I}{$\snlvar$}; \addPrefix{I}{\vz};
	\node[left=of I](Iwest){}; \draw(I.west) to (Iwest.center); \draw[->] (Iwest.center) |- (z.west);
	\draw[->](z) to (I.north);
	\ltgnode[node distance=5cm,right=of y]{findir}{\indir}; \addPos{findir}{\vf};
	\ltgnode[below=of findir]{f}{$\snlvar$}; \addPrefix{f}{\vx};
	\draw[->](findir) to (f.north);
	\node[right=of f](feast){}; \draw(f.east) to (feast.center); \draw[->] (feast.center) |- (x.east);
	\draw[->](fmerge.south) |- (findir.west);
}
\transbreak
\transstep{\snlvarsucc,\snlvarsucc}{
	\node(yIIyff){};
	\ltgnode[node distance=1.5cm,left=of yIIyff]{yI@Iy}{$\sslapp$}; \addPrefix{yI@Iy}{\vx\;\vy}; \node[below=of yI@Iy](yIIy){};
	\ltgnode[node distance=2cm,right=of yIIy]{I@y}{$\sslapp$}; \addPrefix{I@y}{\vx\;\vy}; \node[below=of I@y](Iy){};
	\draw[->](yI@Iy) to (I@y.north);
	\ltgnode[dotted,left=of Iy]{SI2}{$\snlvarsucc$}; \addPrefix{SI2}{\vx\;\vy};
	\draw[->](I@y) to (SI2.north);
	\ltgnode[dotted,below=of SI2]{SSI2}{$\snlvarsucc$}; \addPrefix{SSI2}{\vx};
	\draw[->](SI2) to (SSI2.north);
	\node[draw,shape=transbox, below=of SSI2](I2){$\flabs{*[I^\vI=\labs{z}{z}]}{I}$};
	\draw[->](SSI2) to (I2.north);
	\ltgnode[right=of Iy]{y2}{$\snlvar$}; \addPrefix{y2}{\vx};
	\node[right=of y2](y2l){}; \draw(y2) to (y2l.center);
	\draw[->](I@y) to (y2.north);
	\ltgnode[node distance=2.5cm, left=of SI2]{y@I}{$\sslapp$}; \addPrefix{y@I}{\vx\;\vy}; \node[below=of y@I](yI){};
	\draw[->](yI@Iy) to (y@I.north);
	\ltgnode[left=of yI]{y1}{$\snlvar$}; \addPrefix{y1}{\vx\;\vy};
	\node[left=of y1](y1l){}; \draw(y1) to (y1l.center);
	\draw[->](y@I) to (y1.north);
	\ltgnode[dotted,right=of yI]{SI1}{$\snlvarsucc$}; \addPrefix{SI1}{\vx\;\vy};
	\draw[->](y@I) to (SI1.north);
	\ltgnode[dotted,below=of SI1]{SSI1}{$\snlvarsucc$}; \addPrefix{SSI1}{\vx};
	\draw[->](SI1) to (SSI1.north);
	\node[draw,shape=transbox, below=of SSI1](I1){$\flabs{*[I^\vI=\labs{z}{z}]}{I}$};
	\draw[->](SSI1) to (I1.north);
	\ltgnode[dotted,node distance=2cm,right=of I@y]{Sff}{$\snlvarsucc$}; \addPrefix{Sff}{\vx};
	\ltgnode[below=of Sff]{f@f}{$\sslapp$}; \addPrefix{f@f}{\vx\;\vy}; \node[below=of f@f](ff){};
	\draw[->](Sff) to (f@f.north);
	\node[node distance=1.5cm,right=of ff](fmerge){};
	\node[ left=of ff](f1){}; \draw(f@f) to (f1.south); \draw(f1.south) |- (fmerge.south);
	\node[right=of ff](f2){}; \draw(f@f) to (f2.center); \draw(f2.center) to (fmerge.center);
	\ltgnode[above=of yIIyff.north]{yIIy@ff}{$\sslapp$}; \addPrefix{yIIy@ff}{\vx\;\vy};
	\draw[->](yIIy@ff) to (yI@Iy.north); \draw[->](yIIy@ff) to (Sff.north);
	\ltgnode[above=of yIIy@ff.north]{y}{$\sslabs$}; \addPos{y}{\vy}; \addPrefix{y}{\vx};
	\draw[->](y2l.center) |- (y.east); \draw[->](y1l.center) |- (y.west);
	\draw[->](y) to (yIIy@ff.north);
	\ltgnode[above=of y.north]{x}{$\sslabs$}; \addPos{x}{\vx}; \addPrefix{x}{};
	\draw[->](x) to (y.north);
	\draw[<-](x.north) to +(0mm,4mm);
	\ltgnode[node distance=6cm,left=of y]{Iindir}{\indir}; \addPos{Iindir}{\vI};
	\ltgnode[below=of Iindir]{z}{$\sslabs$}; \addPos{z}{\vz}; \addPrefix{z}{};
	\draw[->](Iindir) to (z.north);
	\ltgnode[below=of z]{I}{$\snlvar$}; \addPrefix{I}{\vz};
	\node[left=of I](Iwest){}; \draw(I.west) to (Iwest.center); \draw[->] (Iwest.center) |- (z.west);
	\draw[->](z) to (I.north);
	\ltgnode[node distance=5cm,right=of y]{findir}{\indir}; \addPos{findir}{\vf};
	\ltgnode[below=of findir]{f}{$\snlvar$}; \addPrefix{f}{\vx};
	\draw[->](findir) to (f.north);
	\node[right=of f](feast){}; \draw(f.east) to (feast.center); \draw[->] (feast.center) |- (x.east);
	\draw[->](fmerge.south) |- (findir.west);
}
\transbreak
\transstep{\arecvar,\arecvar}{
	\node(yIIyff){};
	\ltgnode[node distance=1.5cm,left=of yIIyff]{yI@Iy}{$\sslapp$}; \addPrefix{yI@Iy}{\vx\;\vy}; \node[below=of yI@Iy](yIIy){};
	\ltgnode[node distance=1cm,right=of yIIy]{I@y}{$\sslapp$}; \addPrefix{I@y}{\vx\;\vy}; \node[below=of I@y](Iy){};
	\draw[->](yI@Iy) to (I@y.north);
	\ltgnode[dotted,left=of Iy]{SI2}{$\snlvarsucc$}; \addPrefix{SI2}{\vx\;\vy};
	\draw[->](I@y) to (SI2.north);
	\ltgnode[dotted,below=of SI2]{SSI2}{$\snlvarsucc$}; \addPrefix{SSI2}{\vx};
	\draw[->](SI2) to (SSI2.north);
	\ltgnode[right=of Iy]{y2}{$\snlvar$}; \addPrefix{y2}{\vx};
	\node[right=of y2](y2l){}; \draw(y2) to (y2l.center);
	\draw[->](I@y) to (y2.north);
	\ltgnode[node distance=1cm, left=of yIIy]{y@I}{$\sslapp$}; \addPrefix{y@I}{\vx\;\vy}; \node[below=of y@I](yI){};
	\draw[->](yI@Iy) to (y@I.north);
	\ltgnode[left=of yI]{y1}{$\snlvar$}; \addPrefix{y1}{\vx\;\vy};
	\node[left=of y1](y1l){}; \draw(y1) to (y1l.center);
	\draw[->](y@I) to (y1.north);
	\ltgnode[dotted,right=of yI]{SI1}{$\snlvarsucc$}; \addPrefix{SI1}{\vx\;\vy};
	\draw[->](y@I) to (SI1.north);
	\ltgnode[dotted,below=of SI1]{SSI1}{$\snlvarsucc$}; \addPrefix{SSI1}{\vx};
	\draw[->](SI1) to (SSI1.north);
	\node[below=of SSI1](Imerge){};
	\ltgnode[dotted,node distance=2cm,right=of I@y]{Sff}{$\snlvarsucc$}; \addPrefix{Sff}{\vx};
	\ltgnode[below=of Sff]{f@f}{$\sslapp$}; \addPrefix{f@f}{\vx\;\vy}; \node[below=of f@f](ff){};
	\draw[->](Sff) to (f@f.north);
	\node[node distance=1.5cm,right=of ff](fmerge){};
	\node[ left=of ff](f1){}; \draw(f@f) to (f1.south); \draw(f1.south) |- (fmerge.south);
	\node[right=of ff](f2){}; \draw(f@f) to (f2.center); \draw(f2.center) to (fmerge.center);
	\ltgnode[above=of yIIyff.north]{yIIy@ff}{$\sslapp$}; \addPrefix{yIIy@ff}{\vx\;\vy};
	\draw[->](yIIy@ff) to (yI@Iy.north); \draw[->](yIIy@ff) to (Sff.north);
	\ltgnode[above=of yIIy@ff.north]{y}{$\sslabs$}; \addPos{y}{\vy}; \addPrefix{y}{\vx};
	\draw[->](y2l.center) |- (y.east); \draw[->](y1l.center) |- (y.west);
	\draw[->](y) to (yIIy@ff.north);
	\ltgnode[above=of y.north]{x}{$\sslabs$}; \addPos{x}{\vx}; \addPrefix{x}{};
	\draw[->](x) to (y.north);
	\draw[<-](x.north) to +(0mm,4mm);
	\ltgnode[node distance=6cm,left=of y]{Iindir}{\indir}; \addPos{Iindir}{\vI};
	\ltgnode[below=of Iindir]{z}{$\sslabs$}; \addPos{z}{\vz}; \addPrefix{z}{};
	\draw[->](Iindir) to (z.north);
	\ltgnode[below=of z]{I}{$\snlvar$}; \addPrefix{I}{\vz};
	\node[left=of I](Iwest){}; \draw(I.west) to (Iwest.center); \draw[->] (Iwest.center) |- (z.west);
	\draw[->](z) to (I.north);
	\ltgnode[node distance=4cm,right=of y]{findir}{\indir}; \addPos{findir}{\vf};
	\ltgnode[below=of findir]{f}{$\snlvar$}; \addPrefix{f}{\vx};
	\draw[->](findir) to (f.north);
	\node[right=of f](feast){}; \draw(f.east) to (feast.center); \draw[->] (feast.center) |- (x.east);
	\draw[->](fmerge.south) |- (findir.west);
	\draw(SSI1) to (Imerge.center);
	\draw(SSI2) |- (Imerge.center);
	\node[right=of Iindir](Iindireast){};
	\draw[->](Iindireast.center) to (Iindir.east);
	\draw(Imerge.center) -| (Iindireast.center);
}
\transbreak
\transstep{\text{eliminate indirection vertices}}{
	\node(yIIyff){};
	\ltgnode[node distance=1.5cm,left=of yIIyff]{yI@Iy}{$\sslapp$}; \addPrefix{yI@Iy}{\vx\;\vy}; \node[below=of yI@Iy](yIIy){};
	\ltgnode[node distance=1cm,right=of yIIy]{I@y}{$\sslapp$}; \addPrefix{I@y}{\vx\;\vy}; \node[below=of I@y](Iy){};
	\draw[->](yI@Iy) to (I@y.north);
	\ltgnode[dotted,left=of Iy]{SI2}{$\snlvarsucc$}; \addPrefix{SI2}{\vx\;\vy};
	\draw[->](I@y) to (SI2.north);
	\ltgnode[dotted,below=of SI2]{SSI2}{$\snlvarsucc$}; \addPrefix{SSI2}{\vx};
	\draw[->](SI2) to (SSI2.north);
	\ltgnode[right=of Iy]{y2}{$\snlvar$}; \addPrefix{y2}{\vx};
	\node[right=of y2](y2l){}; \draw(y2) to (y2l.center);
	\draw[->](I@y) to (y2.north);
	\ltgnode[node distance=1cm, left=of yIIy]{y@I}{$\sslapp$}; \addPrefix{y@I}{\vx\;\vy}; \node[below=of y@I](yI){};
	\draw[->](yI@Iy) to (y@I.north);
	\ltgnode[left=of yI]{y1}{$\snlvar$}; \addPrefix{y1}{\vx\;\vy};
	\node[left=of y1](y1l){}; \draw(y1) to (y1l.center);
	\draw[->](y@I) to (y1.north);
	\ltgnode[dotted,right=of yI]{SI1}{$\snlvarsucc$}; \addPrefix{SI1}{\vx\;\vy};
	\draw[->](y@I) to (SI1.north);
	\ltgnode[dotted,below=of SI1]{SSI1}{$\snlvarsucc$}; \addPrefix{SSI1}{\vx};
	\draw[->](SI1) to (SSI1.north);
	\ltgnode[dotted,node distance=1.5cm,right=of yIIyff]{Sff}{$\snlvarsucc$}; \addPrefix{Sff}{\vx};
	\ltgnode[below=of Sff]{f@f}{$\sslapp$}; \addPrefix{f@f}{\vx\;\vy}; \node[below=of f@f](ff){};
	\draw[->](Sff) to (f@f.north);
	\node[node distance=1.5cm,right=of ff](fmerge){};
	\ltgnode[above=of yIIyff.north]{yIIy@ff}{$\sslapp$}; \addPrefix{yIIy@ff}{\vx\;\vy};
	\draw[->](yIIy@ff) to (yI@Iy.north); \draw[->](yIIy@ff) to (Sff.north);
	\ltgnode[above=of yIIy@ff.north]{y}{$\sslabs$}; \addPos{y}{\vy}; \addPrefix{y}{\vx};
	\draw[->](y2l.center) |- (y.east); \draw[->](y1l.center) |- (y.west);
	\draw[->](y) to (yIIy@ff.north);
	\ltgnode[above=of y.north]{x}{$\sslabs$}; \addPos{x}{\vx}; \addPrefix{x}{};
	\draw[->](x) to (y.north);
	\draw[<-](x.north) to +(0mm,4mm);
	\ltgnode[node distance=2.5cm,below=of yIIy]{z}{$\sslabs$}; \addPos{z}{\vz}; \addPrefix{z}{};
	\ltgnode[below=of z]{I}{$\snlvar$}; \addPrefix{I}{\vz};
	\node[left=of I](Iwest){}; \draw(I.west) to (Iwest.center); \draw[->] (Iwest.center) |- (z.west);
	\draw[->](z) to (I.north);
	\ltgnode[below=of ff]{f}{$\snlvar$}; \addPrefix{f}{\vx};
	\draw[->, bend right](f@f) to (f.north);
	\draw[->, bend left](f@f) to (f.north);
	\node[right=of f](feast){}; \draw(f.east) to (feast.center); \draw[->] (feast.center) |- (x.east);
	\draw[->](SSI1) to (z.north);
	\draw[->](SSI2) to (z.north);
}

\newpage

\vspace*{-1.5ex}
\begin{multicols}{2}
Second, we give the translation of the same term $\allteri{2}$ (from Example~\ref{ex:translations} on page~\pageref{ex:translations}) 
with the process needed for the first-order \termgraph\ semantics $\sgraphsemC{\classltgs}$,
yielding $\graphsemC{\classltgs}{\allteri{2}}$ in Fig.~\ref{fig:def:graphsem:lhotgs},
and, at first sight incidentally\footnotemark\addtocounter{footnote}{-1}, 
the corresponding \lambdahotg~$\graphsemC{\classlhotgs}{\allteri{2}}$. 
Note that the resulting \lambdahotg\ (ignore the dotted delimiter
$\snlvarsucc$-vertices) is again $\graphsemC{\classlhotgs}{\allteri{2}}$, that
is, it is identical%
  \footnote{The fact that this is actually not just a coincidence in this specific example 
            is an easy consequence of Prop.~\ref{prop:corr:laphotgs:ltgs}, \ref{prop:corr:laphotgs:ltgs:item:isomorphism},
            Prop.~\ref{prop:graphsem:ltgs:min:graphsem:ltgs}, and Prop.~\ref{prop:graphsem:ltgs:min:graphsem:ltgs:S}.}
with the one that was produced by the translation process
above. Yet the obtained \lambdatg\ (now taking the dotted
$\snlvarsucc$-vertices into account) $\graphsemC{\classltgs}{\allteri{2}}$
differs from the \lambdatg~$\graphsemCmin{\classltgs}{\allteri{2}}$ obtained
above by exhibiting a higher degree of $\snlvarsucc$-sharing.
\end{multicols}

\transinit{
	\node[draw,shape=transbox](t){$\femptylabs{\labs{x}{\labs{y}{\letin{I=\labs{z}{z},f=x}{\lapp{\lapp{\lapp{y}{I}}{(\lapp{I}{y})}}{(\lapp{f}{f})}}}}}$};
	\draw[<-](t.north) to +(0mm,6mm);
}
\transstep{\sslabs\hspace{1cm}}{
	\node[draw,shape=transbox](t){$\flabs{*[]\;x^\vx[]}{\labs{y}{\letin{I=\labs{z}{z},f=x}{\lapp{\lapp{\lapp{y}{I}}{(\lapp{I}{y})}}{(\lapp{f}{f})}}}}$};
	\ltgnode[above=of t.north]{x}{$\sslabs$}; \addPos{x}{\vx}; \addPrefix{x}{};
	\draw[<-](x.north) to +(0mm,4mm);
	\draw[->](x) to (t.north);
}
\transbreak
\transstep{\sslabs}{
	\node[draw,shape=transbox](t){$\flabs{*[]\;x^\vx[]\;y^\vy[]}{\letin{I=\labs{z}{z},f=x}{\lapp{\lapp{\lapp{y}{I}}{(\lapp{I}{y})}}{(\lapp{f}{f})}}}$};
	\ltgnode[above=of t.north]{y}{$\sslabs$}; \addPos{y}{\vy}; \addPrefix{y}{\vx};
	\draw[->](y) to (t.north);
	\ltgnode[above=of y.north]{x}{$\sslabs$}; \addPos{x}{\vx}; \addPrefix{x}{};
	\draw[->](x) to (y.north);
	\draw[<-](x.north) to +(0mm,4mm);
}
\transbreak
\transstep{\sslet}{
	\node[draw,shape=transbox](yIIyff){$\flabs{*[]\;x^\vx[f^\vf=x]\;y^\vy[I^\vI=\labs{z}{z}]}{\lapp{\lapp{\lapp{y}{I}}{(\lapp{I}{y})}}{(\lapp{f}{f})}}$};
	\ltgnode[above=of yIIyff.north]{y}{$\sslabs$}; \addPos{y}{\vy}; \addPrefix{y}{\vy};
	\draw[->](y) to (yIIyff.north);
	\ltgnode[above=of y.north]{x}{$\sslabs$}; \addPos{x}{\vx}; \addPrefix{x}{};
	\draw[->](x) to (y.north);
	\draw[<-](x.north) to +(0mm,4mm);
	\ltgnode[left=of yIIyff]{Iindir}{\indir}; \addPos{Iindir}{\vI};
	\node[draw,shape=transbox,below=of Iindir](I){$\flabs{*[]\;x^\vx[f^\vf=x]\;y^\vy[I^\vI=\labs{z}{z}]}{\labs{z}{z}}$};
	\draw[->](Iindir) to (I.north);
	\ltgnode[right=of yIIyff]{findir}{\indir}; \addPos{findir}{\vf};
	\node[draw,shape=transbox,below=of findir](f){$\flabs{*[]\;x^\vx[f^\vf=x]}{x}$};
	\draw[->](findir) to (f.north);
}
\transbreak
\transstep{\snlvarsucc,\sslapp,\snlvar}{
	\node(yIIyff){};
	\node[node distance=0mm,draw,shape=transbox,left=of yIIyff](yIIy){$\flabs{*[]\;x^\vx[f^\vf=x]\;y^\vy[I^\vI=\labs{z}{z}]}{\lapp{\lapp{y}{I}}{(\lapp{I}{y})}}$};
	\node[node distance=0mm,draw,shape=transbox,right=of yIIyff](ff){$\flabs{*[]\;x^\vx[f^\vf=x]\;y^\vy[I^\vI=\labs{z}{z}]}{\lapp{f}{f}}$};
	\ltgnode[above=of yIIyff.north]{yIIy@ff}{$\sslapp$}; \addPrefix{yIIy@ff}{\vx\;\vy};
	\draw[->](yIIy@ff) to (yIIy.north); \draw[->](yIIy@ff) to (ff.north);
	\ltgnode[above=of yIIy@ff.north]{y}{$\sslabs$}; \addPos{y}{\vy}; \addPrefix{y}{\vx};
	\draw[->](y) to (yIIy@ff.north);
	\ltgnode[above=of y.north]{x}{$\sslabs$}; \addPos{x}{\vx}; \addPrefix{x}{};
	\draw[->](x) to (y.north);
	\draw[<-](x.north) to +(0mm,4mm);
	\ltgnode[dotted,node distance=8mm,left=of yIIy]{z}{$\snlvarsucc$}; \addPrefix{z}{\vx\;\vy};
	\ltgnode[above=of z]{Iindir}{\indir}; \addPos{Iindir}{\vI};
	\draw[->](Iindir) to (z.north);
	\node[draw,shape=transbox,below=of z](I){$\flabs{*[]\;x^\vx[f^\vf=x]}{\labs{z}{z}}$};
	\draw[->](z) to (I.north);
	\ltgnode[node distance=8mm,right=of ff]{f}{$\snlvar$}; \addPrefix{f}{\vx};
	\ltgnode[above=of f]{findir}{\indir}; \addPos{findir}{\vf};
	\draw[->](findir) to (f.north);
	\node[right=of f](feast){}; \draw(f.east) to (feast.center); \draw[->] (feast.center) |- (x.east);
}
\transbreak
\transstep{\snlvarsucc, \sslapp, \snlvarsucc}{
	\node(yIIyff){};
	\ltgnode[node distance=1.5cm,left=of yIIyff]{yI@Iy}{$\sslapp$}; \addPrefix{yI@Iy}{\vx\;\vy}; \node[node distance=8mm,below=of yI@Iy](yIIy){};
	\node[node distance=0mm,draw,shape=transbox, left=of yIIy](yI){$\flabs{*[]\;x^\vx[f^\vf=x]\;y^\vy[I^\vI=\labs{z}{z}]}{\lapp{y}{I}}$};
	\draw[->](yI@Iy) to (yI.north);
	\node[node distance=0mm,draw,shape=transbox,right=of yIIy](Iy){$\flabs{*[]\;x^\vx[f^\vf=x]\;y^\vy[I^\vI=\labs{z}{z}]}{\lapp{I}{y}}$};
	\draw[->](yI@Iy) to (Iy.north);
	\ltgnode[dotted,node distance=1cm,right=of Iy]{Sff}{$\snlvarsucc$}; \addPrefix{Sff}{\vx\;\vy};
	\node[draw,shape=transbox,below=of Sff](ff){$\flabs{*[]\;x^\vx[f^\vf=x]}{\lapp{f}{f}}$};
	\draw[->](Sff) to (ff.north);
	\ltgnode[above=of yIIyff.north]{yIIy@ff}{$\sslapp$}; \addPrefix{yIIy@ff}{\vx\;\vy};
	\draw[->](yIIy@ff) to (yI@Iy.north); \draw[->](yIIy@ff) to (Sff.north);
	\ltgnode[above=of yIIy@ff.north]{y}{$\sslabs$}; \addPos{y}{\vy}; \addPrefix{y}{\vx};
	\draw[->](y) to (yIIy@ff.north);
	\ltgnode[above=of y.north]{x}{$\sslabs$}; \addPos{x}{\vx}; \addPrefix{x}{};
	\draw[->](x) to (y.north);
	\draw[<-](x.north) to +(0mm,4mm);
	\ltgnode[dotted,node distance=5cm,left=of yI@Iy]{z}{$\snlvarsucc$}; \addPrefix{z}{\vx\;\vy};
	\ltgnode[above=of z]{Iindir}{\indir}; \addPos{Iindir}{\vI};
	\draw[->](Iindir) to (z.north);
	\ltgnode[dotted,below=of z]{zS}{$\snlvarsucc$}; \addPrefix{zS}{\vx};
	\node[draw,shape=transbox,below=of zS](I){$\flabs{*[]}{\labs{z}{z}}$};
	\draw[->](z) to (zS.north);
	\draw[->](zS) to (I.north);
	\ltgnode[node distance=5cm,right=of y]{findir}{\indir}; \addPos{findir}{\vf};
	\ltgnode[below=of findir]{f}{$\snlvar$}; \addPrefix{f}{\vx};
	\draw[->](findir) to (f.north);
	\node[right=of f](feast){}; \draw(f.east) to (feast.center); \draw[->] (feast.center) |- (x.east);
}
\transbreak
\transstep{\sslabs,\sslapp,\sslapp,\sslapp}{
	\node(yIIyff){};
	\ltgnode[node distance=3cm,left=of yIIyff]{yI@Iy}{$\sslapp$}; \addPrefix{yI@Iy}{\vx\;\vy}; \node[node distance=8mm,below=of yI@Iy](yIIy){};
	\ltgnode[node distance=1cm,right=of yIIy]{I@y}{$\sslapp$}; \addPrefix{I@y}{\vx\;\vy}; \node[below=of I@y](Iy){};
	\draw[->](yI@Iy) to (I@y.north);
	\node[node distance=0mm,draw,shape=transbox, left=of Iy](I2){$\flabs{*[]\;x^\vx[f^\vf=x]\;y^\vy[I^\vI=\labs{z}{z}]}{I}$};
	\draw[->](I@y) to (I2.north);
	\node[node distance=0mm,draw,shape=transbox,right=of Iy](y2){$\flabs{*[]\;x^\vx[f^\vf=x]\;y^\vy[I^\vI=\labs{z}{z}]}{y}$};
	\draw[->](I@y) to (y2.north);
	\ltgnode[node distance=5mm, left=of I2]{y@I}{$\sslapp$}; \addPrefix{y@I}{\vx\;\vy}; \node[node distance=1cm,below=of y@I](yI){};
	\draw[->](yI@Iy) to (y@I.north);
	\node[node distance=0mm,draw,shape=transbox, left=of yI](y1){$\flabs{*[]\;x^\vx[f^\vf=x]\;y^\vy[I^\vI=\labs{z}{z}]}{y}$};
	\draw[->](y@I) to (y1.north);
	\node[node distance=0mm,draw,shape=transbox, right=of yI](I1){$\flabs{*[]\;x^\vx[f^\vf=x]\;y^\vy[I^\vI=\labs{z}{z}]}{I}$};
	\draw[->](y@I) to (I1.north);
	\ltgnode[dotted,node distance=5cm,right=of I@y]{Sff}{$\snlvarsucc$}; \addPrefix{Sff}{\vx\;\vy};
	\ltgnode[node distance=1cm,below=of Sff]{f@f}{$\sslapp$}; \addPrefix{f@f}{\vx\;\vy}; \node[below=of f@f](ff){};
	\draw[->](Sff) to (f@f.north);
	\node[node distance=0mm,draw,shape=transbox, left=of ff](f1){$\flabs{*[]\;x^\vx[f^\vf=x]}{f}$};
	\draw[->](f@f) to (f1.north);
	\node[node distance=0mm,draw,shape=transbox,right=of ff](f2){$\flabs{*[]\;x^\vx[f^\vf=x]}{f}$};
	\draw[->](f@f) to (f2.north);
	\ltgnode[above=of yIIyff.north]{yIIy@ff}{$\sslapp$}; \addPrefix{yIIy@ff}{\vx\;\vy};
	\draw[->](yIIy@ff) to (yI@Iy.north); \draw[->](yIIy@ff) to (Sff.north);
	\ltgnode[above=of yIIy@ff.north]{y}{$\sslabs$}; \addPos{y}{\vy}; \addPrefix{y}{\vx};
	\draw[->](y) to (yIIy@ff.north);
	\ltgnode[above=of y.north]{x}{$\sslabs$}; \addPos{x}{\vx}; \addPrefix{x}{};
	\draw[->](x) to (y.north);
	\draw[<-](x.north) to +(0mm,4mm);
	\ltgnode[node distance=10cm,left=of y]{Iindir}{\indir}; \addPos{Iindir}{\vI};
	\ltgnode[dotted,below=of Iindir]{SSz}{$\snlvarsucc$}; \addPrefix{SSz}{\vx\;\vy};
	\draw[->](Iindir) to (SSz.north);
	\ltgnode[dotted,below=of SSz]{Sz}{$\snlvarsucc$}; \addPrefix{Sz}{\vx};
	\draw[->](SSz) to (Sz.north);
	\ltgnode[below=of Sz]{z}{$\sslabs$}; \addPos{z}{\vz}; \addPrefix{z}{};
	\draw[->](Sz) to (z.north);
	\node[draw,shape=transbox,below=of z](I){$\flabs{*[]\;z^\vz[]}{z}$};
	\draw[->](z) to (I.north);
	\ltgnode[node distance=4cm,right=of y]{findir}{\indir}; \addPos{findir}{\vf};
	\ltgnode[below=of findir]{f}{$\snlvar$}; \addPrefix{f}{\vx};
	\draw[->](findir) to (f.north);
	\node[right=of f](feast){}; \draw(f.east) to (feast.center); \draw[->] (feast.center) |- (x.east);
}
\transbreak
\transstep{\snlvar, \snlvar, \arecvar, \arecvar, \snlvar, \arecvar, \arecvar}{
	\node(yIIyff){};
	\ltgnode[node distance=1.2cm,left=of yIIyff]{yI@Iy}{$\sslapp$}; \addPrefix{yI@Iy}{\vx\;\vy}; \node[below=of yI@Iy](yIIy){};
	\ltgnode[right=of yIIy]{I@y}{$\sslapp$}; \addPrefix{I@y}{\vx\;\vy}; \node[below=of I@y](Iy){};
	\draw[->](yI@Iy) to (I@y.north);
	\node[left=of Iy](SI2){};
	\draw(I@y) to (SI2.center);
	\ltgnode[right=of Iy]{y2}{$\snlvar$}; \addPrefix{y2}{\vx\;\vy};
	\node[right=of y2](y2l){}; \draw(y2) to (y2l.center); \draw[->](y2l.center) |- (y.east);
	\draw[->](I@y) to (y2.north);
	\ltgnode[left=of yIIy]{y@I}{$\sslapp$}; \addPrefix{y@I}{\vx\;\vy}; \node[below=of y@I](yI){};
	\draw[->](yI@Iy) to (y@I.north);
	\ltgnode[left=of yI]{y1}{$\snlvar$}; \addPrefix{y1}{\vx\;\vy};
	\node[left=of y1](y1l){}; \draw(y1) to (y1l.center); \draw[->](y1l.center) |- (y.west);
	\draw[->](y@I) to (y1.north);
	\node[right=of yI](SI1){};
	\draw(y@I) to (SI1.center);
	\node[below=of SI1](Imerge){};
	\draw(SI1.center) to (Imerge.center);
	\draw(SI2.center) |- (Imerge.center);
	\ltgnode[dotted,node distance=1.3cm,right=of yIIyff]{Sff}{$\snlvarsucc$}; \addPrefix{Sff}{\vx\;\vy};
	\ltgnode[below=of Sff]{f@f}{$\sslapp$}; \addPrefix{f@f}{\vx\;\vy}; \node[below=of f@f](ff){};
	\draw[->](Sff) to (f@f.north);
	\node[node distance=1cm,right=of ff](fmerge){};
	\node[ left=of ff](f1){}; \draw(f@f) to (f1.south); \draw(f1.south) |- (fmerge.south);
	\node[right=of ff](f2){}; \draw(f@f) to (f2.center); \draw(f2.center) to (fmerge.center);
	\ltgnode[above=of yIIyff.north]{yIIy@ff}{$\sslapp$}; \addPrefix{yIIy@ff}{\vx\;\vy};
	\draw[->](yIIy@ff) to (yI@Iy.north); \draw[->](yIIy@ff) to (Sff.north);
	\ltgnode[above=of yIIy@ff.north]{y}{$\sslabs$}; \addPos{y}{\vy}; \addPrefix{y}{\vx};
	\draw[->](y) to (yIIy@ff.north);
	\ltgnode[above=of y.north]{x}{$\sslabs$}; \addPos{x}{\vx}; \addPrefix{x}{};
	\draw[->](x) to (y.north);
	\draw[<-](x.north) to +(0mm,4mm);
	\ltgnode[node distance=5cm,left=of y]{Iindir}{\indir}; \addPos{Iindir}{\vI};
	\node[right=of Iindir](Ieast){};
	\draw[->](Ieast.center) to (Iindir);
	\ltgnode[dotted,below=of Iindir]{SSz}{$\snlvarsucc$}; \addPrefix{SSz}{\vx\;\vy};
	\draw[->](Iindir) to (SSz.north);
	\ltgnode[dotted,below=of SSz]{Sz}{$\snlvarsucc$}; \addPrefix{Sz}{\vx};
	\draw[->](SSz) to (Sz.north);
	\ltgnode[below=of Sz]{z}{$\sslabs$}; \addPos{z}{\vz}; \addPrefix{z}{};
	\draw[->](Sz) to (z.north);
	\ltgnode[below=of z]{I}{$\snlvar$}; \addPrefix{I}{\vz};
	\node[left=of I](Iwest){}; \draw(I.west) to (Iwest.center); \draw[->] (Iwest.center) |- (z.west);
	\draw[->](z) to (I.north);
	\ltgnode[node distance=3cm,right=of y]{findir}{\indir}; \addPos{findir}{\vf};
	\ltgnode[below=of findir]{f}{$\snlvar$}; \addPrefix{f}{\vx};
	\draw[->](findir) to (f.north);
	\node[right=of f](feast){}; \draw(f.east) to (feast.center); \draw[->] (feast.center) |- (x.east);
	\draw[->](fmerge.south) |- (findir.west);
	\draw(Imerge.center) -| (Ieast.center);
}
\transbreak
\transstep{\text{erasure of indirection vertices}}{
	\node(yIIyff){};
	\ltgnode[node distance=1.2cm,left=of yIIyff]{yI@Iy}{$\sslapp$}; \addPrefix{yI@Iy}{\vx\;\vy}; \node[below=of yI@Iy](yIIy){};
	\ltgnode[right=of yIIy]{I@y}{$\sslapp$}; \addPrefix{I@y}{\vx\;\vy}; \node[below=of I@y](Iy){};
	\draw[->](yI@Iy) to (I@y.north);
	\node[left=of Iy](SI2){};
	\ltgnode[right=of Iy]{y2}{$\snlvar$}; \addPrefix{y2}{\vx\;\vy};
	\node[right=of y2](y2l){}; \draw(y2) to (y2l.center); \draw[->](y2l.center) |- (y.east);
	\draw[->](I@y) to (y2.north);
	\ltgnode[left=of yIIy]{y@I}{$\sslapp$}; \addPrefix{y@I}{\vx\;\vy}; \node[below=of y@I](yI){};
	\draw[->](yI@Iy) to (y@I.north);
	\ltgnode[left=of yI]{y1}{$\snlvar$}; \addPrefix{y1}{\vx\;\vy};
	\node[left=of y1](y1l){}; \draw(y1) to (y1l.center); \draw[->](y1l.center) |- (y.west);
	\draw[->](y@I) to (y1.north);
	\node[right=of yI](SI1){};
	\ltgnode[dotted,node distance=1.3cm,right=of yIIyff]{Sff}{$\snlvarsucc$}; \addPrefix{Sff}{\vx\;\vy};
	\ltgnode[below=of Sff]{f@f}{$\sslapp$}; \addPrefix{f@f}{\vx\;\vy}; \node[below=of f@f](ff){};
	\draw[->](Sff) to (f@f.north);
	\ltgnode[above=of yIIyff.north]{yIIy@ff}{$\sslapp$}; \addPrefix{yIIy@ff}{\vx\;\vy};
	\draw[->](yIIy@ff) to (yI@Iy.north); \draw[->](yIIy@ff) to (Sff.north);
	\ltgnode[above=of yIIy@ff.north]{y}{$\sslabs$}; \addPos{y}{\vy}; \addPrefix{y}{\vx};
	\draw[->](y) to (yIIy@ff.north);
	\ltgnode[above=of y.north]{x}{$\sslabs$}; \addPos{x}{\vx}; \addPrefix{x}{};
	\draw[->](x) to (y.north);
	\draw[<-](x.north) to +(0mm,4mm);
	\ltgnode[distance=2.5cm,dotted,below=of yIIy]{SSz}{$\snlvarsucc$}; \addPrefix{SSz}{\vx\;\vy};
	\draw[->](y@I) to (SSz);
	\draw[->](I@y) to (SSz);
	\ltgnode[dotted,below=of SSz]{Sz}{$\snlvarsucc$}; \addPrefix{Sz}{\vx};
	\draw[->](SSz) to (Sz.north);
	\ltgnode[below=of Sz]{z}{$\sslabs$}; \addPos{z}{\vz}; \addPrefix{z}{};
	\draw[->](Sz) to (z.north);
	\ltgnode[below=of z]{I}{$\snlvar$}; \addPrefix{I}{\vz};
	\node[left=of I](Iwest){}; \draw(I.west) to (Iwest.center); \draw[->] (Iwest.center) |- (z.west);
	\draw[->](z) to (I.north);
	\ltgnode[below=of ff]{f}{$\snlvar$}; \addPrefix{f}{\vx};
	\node[right=of f](feast){}; \draw(f.east) to (feast.center); \draw[->] (feast.center) |- (x.east);
	\draw[->,bend right](f@f) to (f);
	\draw[->,bend left](f@f) to (f);
}

\newpage
\section{Implementation Showcase}
  \label{app:impl:showcase}

\begin{multicols}{2}
To demonstrate the realisability of our method, and for further illustration,
we include the output of our implementation for the examples used in the paper.
The implementation is called
\href{http://hackage.haskell.org/package/maxsharing/}{maxsharing} and is
available on \href{http://hackage.haskell.org/package/maxsharing/}{Hackage}.
It is written in Haskell and therefore requires the
\href{http://www.haskell.org/platform/}{Haskell Platform} to be installed.

Then, maxsharing can be installed via
\href{http://www.haskell.org/haskellwiki/Cabal-Install}{cabal-install} using
the commands \texttt{cabal update} and and \texttt{cabal install maxsharing}
from the terminal. Invoke the executable \texttt{maxsharing} in your
cabal-directory with a file as an argument that contains a
\txtlambdaletreccal-term. Run \texttt{maxsharing -h} for help on run-time flags.
\end{multicols}

\subsection{Example~\ref{ex:cse}}
\tooltip{(5.525,0)}{6cm}{the original term as recognised by the parser}
\tooltip{(5.2,0.5)}{11cm}{
    The user can specify which \termgraph\  semantics shall be used in processing the term.
    All further output is with respect to that translation. The two options are:
    \begin{itemize} 
      \item $\sgraphsemC{\classltgs}$
        is indicated by:
        {\tt maximal prefix lengths while maintaining eager scope-closure}
      \item $\sgraphsemCmin{\classltgs}$
        is indicated by: 
        {\tt minimal prefix lengths}.
    \end{itemize}
    For this example both semantics yield the same output.
 }
\tooltip{(5.2,2.3)}{6cm}{term with scope delimiters; see Example~\ref{ex:counternat} for a more interesting case and more explanations}
\tooltip{(5.2,3.2)}{11cm}{
     Derivation according to the proof system in
     Fig.~\ref{fig:trans-lambdaletreccal-lhotgs-proof-system}
     that shows the translation as a stepwise process and includes all subterms
     with their abstraction prefixes.
     Note that in the notation of the prefixed terms the binding annotation of
     a variable is omitted if it is empty. Also, it only includes
     the names of the function variables, not their entire definition.

     Even though the correspondence between the derivations in the proof system
     in Fig.~\ref{fig:trans-lambdaletreccal-lhotgs-proof-system} and the
     translations $\sgraphsemC{\classltgs}$ and
     $\sgraphsemCmin{\classltgs}$ is not provided here, we think that the
     derivation can help as an illustration of the translation process and its
     result.
}
\tooltip{(5.2,4.9)}{11cm}{
     The implementation produces a graphical depiction for the term's
     \lambdatg\ in DFA form, as well as for the minimised form of the DFA;
     see below for pictures.
}
\tooltip{(5.2,5.7)}{11cm}{
     A textual representation of the minimised DFA's spanning tree, used for
     readback. It is displayed in first-order term rewriting notation, i.e.\ with a
     unary function symbols {\tt L}, {\tt S} for abstraction and scope
     delimiters, a binary function symbol {\tt A} for application and a nullary
     symbol {\tt 0} for abstraction variable occurrences. Furthermore there is a
     class of function symbols written as a vertical bar followed by an
     upper-case variable name {\tt |F}, {\tt |G}, etc.\ which signify a vertex
     with multiple incoming non-backlink edges, and therefore a vertex that will
     be the root of a shared subterm. There is also a class of corresponding
     nullary function symbols {\tt F}, {\tt G}, etc.\ which represent
     non-backlink, non-spanning-tree edges to these shared vertices.
}
\tooltip{(5.2,7.8)}{6cm}{
    Readback of the minimised DFA.
}
\showcase{01}{ex-cse}{minprefix}{(0.5,9.0)}{(2.5,9.0)}
\newpage

\subsection{Example~\ref{ex:fix}}
Also for the terms $\allter$ and $\bllter$ of Ex.~\ref{ex:fix} on page
\pageref{ex:fix} the translations are identical for $\sgraphsemC{\classlhotgs}$
and $\sgraphsemCmin{\classltgs}$.\\
\begin{multicols}{2}
\showcase{02}{ex-fix-L}{minprefix}{(4.5,0)}{(4.5,2)}
\columnbreak
\showcase{02}{ex-fix-P}{minprefix}{(4.7,0)}{(4.7,3)}
\end{multicols}

\newpage

\subsection{Example~\ref{ex:counternat}}
Again, for the terms $\allteri1$ and $\allteri2$ of Ex.~\ref{ex:counternat}
on page \pageref{ex:counternat} the translations are identical for
$\sgraphsemC{\classlhotgs}$ and $\sgraphsemCmin{\classltgs}$.\\\\
\begin{multicols}{2}
\tooltip{(5.0,0.9)}{11cm}{
    The term enriched by abdmals \cite{hend:oost:2003}. The adbmal
    (\reflectbox{$\lambda$}) is to be read as a scope delimiter that
    explicitly includes the name of the \lambda-variable whose scope it
    delimits. The adbmals are placed in accordance to the translation used.
}
\tooltip{(5.0,1.7)}{11cm}{
     A nameless scoped representation, where the names of
     abstraction variables are omitted for lambdas as well as for abstraction
     variable occurrences, shown as a $\svar$-symbol. The scoping
     is expressed by scope-delimiters in the shape of an $\snlvarsucc$-symbol.
     Note that these
     terms can be obtained by a simple syntactical transformation from the adbmal-terms.
}
\showcase{03}{ex-counternat-L1}{minprefix}{(0,7.5)}{(3,7.5)}

\columnbreak

\mbox{}\vspace{6cm}
\showcase{03}{ex-counternat-L2}{minprefix}{(4.4,0)}{(4.5,4.8)}
\end{multicols}
\newpage

\subsection{Figure~\ref{fig:rigid_let_rule}}
Also for the term
$\labs{\aavar}{\labs{\bbvar}{\letin{\arecvar=\aavar}{\lapp{\lapp{\lapp{\aavar}{\aavar}}{(\lapp{\arecvar}{\aavar})}}{\bbvar}}}}$ 
from Fig.~\ref{fig:rigid_let_rule} on page
\pageref{fig:rigid_let_rule} the translations are identical for
$\sgraphsemC{\classlhotgs}$ and $\sgraphsemCmin{\classltgs}$.\\
\showcase{04}{fig-rigid_let_rule}{minprefix}{(9,0)}{(9,5)}
\newpage

 \subsection{Example~\ref{ex:translations}}
 For the terms $\allteri1$, $\allteri2$, and $\allteri3$ from
 Ex.~\ref{ex:translations} on page \pageref{ex:translations} the
 translations differ for $\sgraphsemC{\classlhotgs}$ and
 $\sgraphsemCmin{\classltgs}$. Thus on the follwing pages we provide the output
 for both translations for each of the terms.\\
\showcase{05}{ex-translations-L1}{minprefix}{(3.6,7.7)}{(8.7,6.3)}
\newpage
\showcase{05}{ex-translations-L1}{maxeagpre}{(4,7.2)}{(9,6.3)}
\newpage
\showcase{05}{ex-translations-L2}{minprefix}{(3,8)}{(9,6.2)}
\newpage
\showcase{05}{ex-translations-L2}{maxeagpre}{(4,7)}{(9,6)}
\newpage
\showcase{05}{ex-translations-L3}{minprefix}{(5,8)}{(10,7.25)}
\newpage
\showcase{05}{ex-translations-L3}{maxeagpre}{(5,8)}{(9.5,7)}
\newpage
\noindent
For the term $\allter'$ from the same example (Ex.~\ref{ex:translations}, page
\pageref{ex:translations}) the translations for $\sgraphsemC{\classlhotgs}$ and
$\sgraphsemCmin{\classltgs}$ are identical again.\\
\showcase{05}{ex-translations-Lp}{minprefix}{(3,7.9)}{(9,7)}
\newpage

\newpage

\subsection{Term $\allteri2$ from Example~\ref{ex:translations} (page \pageref{ex:translations})}
This term has different translations for $\sgraphsemC{\classltgs}$ and
$\sgraphsemCmin{\classltgs}$. Thus on the following two pages we provide the output
for both translations. See also, and compare with, the stepwise translations in Appendix~\ref{app:translation}.\\
\showcase{05}{ex-translations-L2}{maxeagpre}{(4,6.8)}{(9,5.8)}
\newpage
\showcase{05}{ex-translations-L2}{minprefix}{(3,7.5)}{(9,6.2)}

\end{document}